\documentclass[12]{article}

\usepackage{amssymb,amsmath,amsthm}
\usepackage{mathrsfs}
\usepackage{epsfig}
\usepackage{diagrams}
\usepackage{graphicx}

\newcommand\ie{{\em i.e.}}

\def\A{\mathcal{A}}
\def\B{\mathcal B}
\def\C{\mathbb C}
\def\D{\mathcal D}
\def\d{\mathrm d}
\def\e{\mathop{\mathrm{e}}\nolimits}
\def\F{\mathscr F}
\def\G{\mathcal G}
\def\H{\mathcal H}
\def\K{\mathcal K}
\def\M{\mathscr M}
\def\MM{\mathcal M}
\def\N{\mathbb N}

\def\P{\mathcal P}
\def\R{\mathbb R}
\def\S{\mathbb S}
\def\SS{\mathcal S}
\def\T{\mathbb T}
\def\U{\mathcal U}
\def\UU{\mathscr U}
\def\V{\mathscr V}
\def\X{\mathcal X}
\def\Y{\mathcal Y}
\def\bW{\mathbf W}
\def\bG{\mathbf \Gamma}
\def\bE{\mathbf E_{\rm p}}
\def\Z{\mathbb Z}
\def\Hrond{\mathscr H}
\def\Pv{\mathrm{Pv}}
\def\Ran{{\mathrm{Ran}}}
\def\Wind{{\mathrm{Wind}}}
\def\Ker{{\mathrm{Ker}}}
\def\Im{{\mathrm{Im}}}
\def\Re{{\mathrm{Re}}}
\def\det{{\mathrm{det}}}
\def\sgn{{\rm sgn}}
\def\Tr{{\rm Tr}}
\def\tr{{\rm tr}}
\def\ind{{\rm ind}}
\def\bpm{\begin{smallmatrix}}
\def\epm{\end{smallmatrix}}
\def\dom{\mathcal D}
\def\f{f}
\def\HS{\mathfrak h}

\def\ltwo{\mathsf{L}^{\:\!\!2}}
\def\linf{\mathsf{L}^{\:\!\!\infty}}
\def\12{{\textstyle\frac12}}
\def\i2{{\textstyle\frac{i}{2}}}
\def\({\scalebox{1.15}{$\hspace{-0.0mm}\big($}}
\def\){\scalebox{1.15}{$\hspace{-0.25mm}\big)$}}
\newcommand{\bR}{[-\infty,+\infty]}
\newcommand{\bRp}{[0,+\infty]}
\def\sgn{\mathop{\mathrm{sgn}}\nolimits}
\def\Gammaq{\Gamma_{\!\hbox{\scriptsize $\square$}}}

\def\AB{A\!B}
\def\CD{C\!D}
\def\2{\mathfrak{int}}
\def\ud{{\textstyle \frac{1}{2}}}

\def\b{{\tt b}}
\def\dd{{\tt d}}

\def\dim{{\rm dim}}

\def\EE{\mathcal E}
\def\QQ{\mathcal Q}
\def\JJ{\mathcal J}

\newtheorem{Theorem}{Theorem}[section]
\newtheorem{Remark}[Theorem]{Remark}
\newtheorem{Lemma}[Theorem]{Lemma}
\newtheorem{Assumption}[Theorem]{Assumption}
\newtheorem{Corollary}[Theorem]{Corollary}
\newtheorem{Proposition}[Theorem]{Proposition}
\newtheorem{Definition}[Theorem]{Definition}
\newtheorem{Example}[Theorem]{Example}

\begin{document}


\title{Levinson's theorem: an index theorem in scattering theory}

\author{S. Richard\footnote{Supported by JSPS Grant-in-Aid for Young Scientists A no 26707005.}}

\date{\small}
\maketitle \vspace{-1cm}

\begin{quote}
\emph{
\begin{itemize}
\item[] Graduate school of mathematics, Nagoya University,
Chikusa-ku, Nagoya 464-8602, Japan; On leave of absence from
Universit\'e de Lyon; Universit\'e
Lyon 1; CNRS, UMR5208, Institut Camille Jordan,
43 blvd du 11 novembre 1918, F-69622
Villeur\-ban\-ne-Cedex, France
\item[] \emph{E-mail:} richard@math.nagoya-u.ac.jp
\end{itemize}
}
\end{quote}


\begin{abstract}
A topological version of Levinson's theorem is presented.
Its proof relies on a $C^*$-algebraic framework which is introduced in detail.
Various scattering systems are considered in this framework, and more coherent
explanations for the corrections due to thresholds effects or for the regularization
procedure are provided.
Potential scattering, point interactions, Friedrichs model and Aharonov-Bohm operators
are part of the examples which are presented.
Every concepts from scattering theory or from $K$-theory are introduced from scratch.
\end{abstract}


\smallskip


{\small
\tableofcontents
}
\section{Introduction}\label{sec_intro}
\setcounter{equation}{0}

Levinson's theorem is a relation between the number of bound states of a quantum mechanical
system and an expression related to the scattering part of that system.
Its original formulation was established by N. Levinson in \cite{Lev} in the context
of a Schr\"odinger operator with a spherically symmetric potential, but subsequently
numerous authors extend the validity of such a relation in various contexts or for various models.
It is certainly impossible to quote all papers or books containing either Levinson's theorem in
their title or in the subtitle of a section, but let us mention a few key references
\cite{Bolle,Jauch,Ma,New1,Newbook,New2,OB,Rob,ReedS}.
Various methods have also been used for the proof of this relation,
as for example the Jost functions, the Green functions, the Sturm-Liouville theorem,
and most prominently the spectral shift function.
Note that expressions like the phase shift, the Friedel sum rule or some trace formulas
are also associated with Levinon's theorem.

Our aim in this review paper\footnote{
This paper is an extended version of a mini-course given at the \emph{International conference on spectral theory
and mathematical physics} which took place in Santiago (Chile) in November 2014. The author
takes this opportunity to thank the organizers of the conference for their kind invitation and support.}
is to present a radically different approach for Levinson's theorem.
Indeed, during the last couple of years it has been shown that once recast in a $C^*$-algebraic framework
this relation can be understood as an index theorem in scattering theory.
This new approach does not only shed new light on this theorem, but also provides a more coherent and natural way to take
various corrections or regularization processes into account.
In brief, the key point in our proof of Levinson's theorem consists in
evaluating the index of the wave operator by the winding number of an
expression involving not only the scattering operator, but also new operators that describe the
system at thresholds energies.

From this short description, it clearly appears that this new approach relies on two distinct fields of mathematics.
On the one hand, the wave operators and the scattering operator belong to the framework of
spectral and scattering theory, two rather well-known subjects in the mathematical physics community. On the other hand,
index theorem, winding number, and beyond them index map, K-theory and Connes' pairing are familiar tools for
operator algebraists.
For this reason, a special attention has been given to briefly introduce all
concepts which belong only to one of these communities.
One of our motivations in writing this survey is make this approach of Levinson's theorem accessible to both readerships.

Let us now be more precise about the organization of this paper. In Section \ref{sec_baby}
we introduce a so-called ``baby model" on which the essence of our approach can be fully presented.
No prior knowledge on scattering theory or on K-theory is necessary, and all computations can be
explicitly performed. The construction might look quite {\it ad hoc}, but this feeling will hopefully
disappear once the full framework is established.

Section \ref{sec_scat} contains a very short introduction to scattering theory, with the main requirements
imposed on the subsequent scattering systems gathered in Assumption \ref{Assum_1}.
In Section \ref{sec_Kth} we gradually introduce the $C^*$-algebraic framework, starting with a brief
introduction to $K$-theory followed by the introduction of the index map. An abstract
topological Levinson's theorem in then proposed in Theorem \ref{thm_Lev}.
Since this statement still contains an implicit condition, we illustrate our purpose
by introducing in Section \ref{subsec_example} various isomorphic versions of the algebra which is going to play
a key role in the subsequent examples. In the last part of this section we show how
the previous computations performed on the baby model can be explained in this algebraic framework.
Clearly, Sections \ref{sec_scat} and \ref{sec_Kth} can be skipped by experts in these respective fields,
or very rapidly consulted for the notations.

In Section \ref{Sec_1D} we gather several examples of scattering systems which are either one dimensional or essentially
one dimensional. With the word ``essential" we mean that a rather simple reduction of the system under consideration leads
to a system which is not trivial only in a space of dimension one.
Potential scattering on $\R$ is presented and an explanation of the usual $\frac12$-correction is provided.
With another example, we show that embedded or non-embedded eigenvalues play exactly the same role for Levinson's theorem,
a question which had led to some controversies in the past \cite{Dreyfus}.
A sketchy presentation of a few other models is also proposed, and references to the corresponding papers are provided.

With Section \ref{Sec_3D} we start the most analytical section of this review paper. Indeed,
a key role in our approach is played by the wave operators, and a good understanding of them is thus necessary.
Prior to our investigations such a knowledge on these operators was not available in the literature, and
part of our work has consisted in deriving new explicit formulas for these operators. In the previous section
the resulting formulas are presented but not their proofs. In Section \ref{subsec_Wave3} we provide
a rather detailed derivation of these formulas for a system of potential scattering in $\R^3$,
and the corresponding computations are based on a stationary approach of scattering theory.
On the algebraic side this model is also richer than the ones contained in Section \ref{Sec_1D} in the sense that
a slight extension of the algebraic framework introduced in Section \ref{sec_Kth} together with a regularization procedure are necessary.
More precisely, we provide a regularized formula for the computation of the winding number of suitable elements of $C\big(\S;\K_p(\HS)\big)$,
the algebra of continuous functions on the unit circle with value in the $p$-Schatten class of a Hilbert space $\HS$.

A very brief description of the wave operators for potential scattering in $\R^2$ is provided in Section \ref{Sec_2D}.
However, note that for this model a full understanding of the wave operators is not available yet, and that
further investigations are necessary when resonances or eigenvalues take place at the threshold energy $0$.
Accordingly, a full description of a topological Levinson's theorem does not exist yet.

In Section \ref{Sec_AB} we extend the $C^*$-algebraic framework in a different direction, namely to index theorems for families.
First of all, we introduce a rather large family of self-adjoint operators corresponding to so-called Aharonov-Bohm operators.
These operators are obtained as self-adjoint extensions of a closed operator with deficiency indices $(2,2)$.
A Levinson's theorem is then provided for each of them, once suitably compared with the usual Laplace operator on $\R^2$.
For this model, explicit formulas for the wave operators and for the scattering operator are provided,
and a thorough description of the computation of the winding number is also given.
These expressions and computations are presented in Sections \ref{sec_AB_model} and \ref{sec_ped}.

In order to present a Levinson's theorem for families, additional information on cyclic cohomology, $n$-traces
and Connes' pairing are necessary. A very brief survey is provided in Section \ref{sec_cycle}.
A glimpse on dual boundary maps is also given in Section \ref{sec_dual}.
With these information at hands, we derive in Section \ref{sec_higher} a so-called higher degree Levinson's theorem.
The resulting relation corresponds to the equality between the Chern number of a vector bundle given by the
projections on the bound states of the Aharonov-Bohm operators, and a $3$-trace applied to the scattering part of the system.
Even if a physical interpretation of this equality is still lacking, it is likely that it can play a role in the theory of
topological transport and/or adiabatic processes.

Let us now end this Introduction with some final comments.
As illustrated by the multiplicity of the examples, the underlying $C^*$-algebraic
framework for our approach of Levinson's theorem is very flexible and rich.
Beside the extensions already presented in Sections \ref{Sec_3D} and \ref{Sec_AB},
others are appealing. For example, it would certainly be interesting to recast the generalized
Levinson's theorem exhibited in \cite{Rai,Tie} in our framework.
Another challenging extension would be to find out the suitable algebraic framework for
dealing with scattering systems described in a two-Hilbert spaces setting.
Finally, let us mention similar investigations \cite{BS,SchB} which have been performed on discrete
systems with the same $C^*$-algebraic framework in the background.

\section{The baby model}\label{sec_baby}
\setcounter{equation}{0}

In this section we introduce an example of a scattering system for which everything can
be computed explicitly. It will allow us to describe more precisely the kind of results we are looking for,
without having to introduce any $C^*$-algebraic framework or too much information on scattering theory.
In fact, we shall keep the content of this section as simple as possible.

Let us start by considering the Hilbert space $\ltwo(\R_+)$ and the Dirichlet Laplacian $H_{\rm D}$ on $\R_+$. More precisely,
we set $H_{\rm D}=-\frac{\d^2}{\d x^2}$ with the domain $\dom(H_{\rm D})=\{f\in \H^2(\R_+)\mid f(0)=0\}$.
Here $\H^2(\R_+)$ means the usual Sobolev space on $\R_+$ of order $2$.
For any $\alpha \in \R$, let us also consider the operator $H^\alpha$ defined by $H^\alpha = -\frac{\d^2}{\d x^2}$
with $\dom(H^\alpha)=\{f\in \H^2(\R_+)\mid f'(0)=\alpha f(0)\}$.
It is well-known that if $\alpha<0$ the operator $H^\alpha$ possesses only one eigenvalue, namely $-\alpha^2$,
and the corresponding eigenspace is generated by the function $x\mapsto \e^{\alpha x}$.
On the other hand, for $\alpha\geq 0$ the operators $H^\alpha$ have no eigenvalue, and so does $H_{\rm D}$.

As explained in the next section, a common object of scattering theory is defined by the following formula:
$$
W_\pm^\alpha:=s-\lim_{t\to \pm \infty}\e^{itH^\alpha} \e^{-itH_{\rm D}},
$$
and this limit in the strong sense is known to exist for this model, see for example \cite[Sec.~3.1]{Yaf}.
Moreover, we shall provide below a very explicit formula for these operators.
For that purpose, we need to introduce one more operator which is going to play a key role in the sequel.
More precisely, we consider the unitary group $\{U_t\}_{t\in \R}$ acting on any $f\in \ltwo(\R_+)$ as
\begin{equation}\label{eq_dil_group}
[U_tf](x) = \e^{t/2} f\big(\e^t x\big), \qquad \forall x\in \R_+
\end{equation}
which is usually called \emph{the unitary group of dilations},
and denote its self-adjoint generator by $A$ and call it \emph{the generator of dilations}.

Our first result for this model then reads.

\begin{Lemma}\label{lem_baby}
For any $\alpha \in \R$, the following formula holds:
\begin{equation}\label{eq_wave_baby}
W_-^\alpha = 1 + \12\big(1+\tanh(\pi A)-i\cosh(\pi A)^{-1}\big) \Big[\frac{\alpha+i\sqrt{H_{\rm D}}}{\alpha-i\sqrt{H_{\rm D}}}-1\Big].
\end{equation}
\end{Lemma}
Note that a similar formula for $W_+^\alpha$ also holds for this model, see Lemma \ref{lem_baby_W}.
Since the proof of this lemma has never appeared in the literature,
we provide it in the Appendix.
Motivated by the above formula, let us now introduce the function
\begin{equation*}
\Gamma^\alpha_\blacksquare: [0,+\infty]\times [-\infty,+\infty] \ni (x,y) \mapsto  1 + \12\big(1+\tanh(\pi y)-i\cosh(\pi y)^{-1}\big)
\Big[\frac{\alpha+i\sqrt{x}}{\alpha-i\sqrt{x}}-1\Big]\in \C.
\end{equation*}
Since this function is continuous on the square $\blacksquare:=[0,+\infty]\times [-\infty,+\infty]$,
its restriction on the boundary $\square$ of the square is also well defined and continuous.
Note that this boundary is made of four parts: $\square = B_1\cup B_2 \cup B_3 \cup B_4$
with $B_1 = \{0\}\times[-\infty,+\infty]$, $B_2 =[0,+\infty]\times\{+\infty\}$, $B_3 = \{+\infty\}\times[-\infty,+\infty]$,
and $B_4=[0,+\infty]\times\{-\infty\}$. Thus, the algebra $C(\square)$ of continuous functions on $\square$
can be viewed as a subalgebra of
\begin{equation}\label{eq_sum_alg1}
C\big([-\infty,+\infty]\big)\oplus C\big([0,+\infty]\big)\oplus C\big([-\infty,+\infty]\big)\oplus C\big([0,+\infty]\big)
\end{equation}
given by elements $\Gamma=(\Gamma_1, \Gamma_2, \Gamma_3, \Gamma_4)$ which coincide at the corresponding end points, that
is,
$$\Gamma_1(+\infty) = \Gamma_2(0), \ \Gamma_2(+\infty) = \Gamma_3(+\infty), \ \Gamma_3(-\infty) = \Gamma_4(+\infty),
\hbox{ and }\Gamma_4(0) = \Gamma_1(-\infty).
$$
With these notations, the restriction function $\Gammaq^\alpha :=\Gamma^\alpha_\blacksquare\big|_\square$ is given
for $\alpha \neq 0$ by
\begin{equation}\label{eq_baby_Gamma_1}
\Gammaq^\alpha =\Big(1, \frac{\alpha+i\sqrt{\cdot}}{\alpha-i\sqrt{\cdot}},-\tanh(\pi \cdot)+i\cosh(\pi \cdot)^{-1},1\Big)
\end{equation}
and for $\alpha =0$ by
\begin{equation}\label{eq_baby_Gamma_2}
\Gammaq^0 :=\Big(-\tanh(\pi \cdot)+i\cosh(\pi \cdot)^{-1}, -1,-\tanh(\pi \cdot)+i\cosh(\pi \cdot)^{-1},1\Big).
\end{equation}
For simplicity, we have directly written this function in the representation provided by \eqref{eq_sum_alg1}.

Let us now observe that the boundary $\square$  of $\blacksquare$ is homeomorphic to the circle $\S$.
Observe in addition that the function $\Gammaq^\alpha$ takes its values in the unit circle $\T$ of $\C$.
Then, since $\Gammaq^\alpha$ is a continuous function on the closed curve $\square$ and takes
values in $\T$, its winding number $\Wind(\Gammaq^\alpha)$ is well defined and can easily be computed.
So, let us compute separately the contribution $w_j(\Gammaq^\alpha)$ to this winding number on each component $B_j$ of $\square$.
By convention, we shall turn around $\square$ clockwise, starting from the left-down corner,
and the increase in the winding number is also counted clockwise. Let us stress that the contribution on $B_3$
has to be computed from $+\infty$ to $-\infty$, and the contribution on $B_4$ from $+\infty$ to $0$.
Without difficulty one gets:

\begin{center}
\begin{tabular}{c|c|c|c|c|c|}
& $w_1(\Gammaq^\alpha)$ & $w_2(\Gammaq^\alpha)$ & $w_3(\Gammaq^\alpha)$ & $w_4(\Gammaq^\alpha)$ & $\Wind(\Gammaq^\alpha)$ \\
\hline\hline
$\alpha < 0$ & $0$ & $1/2$ & $1/2$ & $0$ & $1$ \\
\hline
$\alpha = 0$ & $-1/2$ & $0$ & $1/2$ & $0$ & $0$ \\
\hline
$\alpha > 0$ & $0$ &$-1/2$ & $1/2$ & $0$ & $0$ \\
\hline
\end{tabular}
\end{center}

By comparing the last column of this table with the information on the eigenvalues of $H^\alpha$ mentioned at the beginning of the section one gets:

\begin{Proposition}\label{prop_baby}
For any $\alpha \in \R$ the following equality holds:
\begin{equation}\label{eq_Lev_1}
\Wind(\Gammaq^\alpha) = \hbox{ number of eigenvalues of } H^\alpha.
\end{equation}
\end{Proposition}

The content of this proposition is an example of Levinson's theorem. Indeed, it relates
the number of bound states of the operator $H^\alpha$ to a quantity computed on the
scattering part of the system. Let us already mention that the contribution $w_2(\Gammaq^\alpha)$ is
the only one usually considered in the literature.
However, we can immediately observe that if $w_1(\Gammaq^\alpha)$ and $w_3(\Gammaq^\alpha)$
are disregarded, then no meaningful statement can be obtained.

Obviously, the above result should now be recast in a more general framework. Indeed, except for very specific models,
it is usually not possible to compute precisely both sides of \eqref{eq_Lev_1}, but our aim is
to show that such an equality still holds in a much more general setting.
For that purpose, a $C^*$-algebraic framework will be constructed in Section \ref{sec_Kth}.

\section{Scattering theory: a brief introduction}\label{sec_scat}
\setcounter{equation}{0}

In this section, we introduce the main objects of spectral and scattering theory which will be used throughout this paper.

Let us start by recalling a few basic facts from spectral theory. We consider a separable Hilbert space $\H$,
with its scalar product denoted by $\langle \cdot, \cdot \rangle$ and its norm by $\|\cdot\|$.
The set of bounded linear operators on $\H$ is denoted by $\B(\H)$.
Now, if $\B(\R)$ denotes the set of Borel sets in $\R$ and if $\P(\H)$ denotes the set of orthogonal projections on $\H$,
then \emph{a spectral measure} is a map $E : \B(\R) \to  \P(\H)$ satisfying the following
properties:
\begin{enumerate}
\item[(i)] $E(\emptyset) = 0$ and $E(\R) = 1$,
\item[(ii)] If $\{\vartheta_n\}_{n\in \N}$ is a family of disjoint Borel sets, then $E(\cup_n \vartheta_n) =
\sum_n E(\vartheta_n)\ $ (convergence in the strong topology).
\end{enumerate}
The importance of spectral measures comes from their relation with the set of self-adjoint operators in $\H$.
More precisely, let $H$ be a self-adjoint operator acting in $\H$, with its domain denoted by $\dom(H)$.
Then, there exists a unique spectral measure $E(\cdot)$ such that $H=\int_\R\lambda E(\d \lambda)$.
Note that this integral has to be understood in the strong sense, and only on elements of $\dom(H)$.

This measure can now be decomposed into three parts, namely its absolutely continuous part, its singular continuous part,
and its pure point part. More precisely, there exists a decomposition of the Hilbert space
$\H=\H_{\rm ac}(H)\oplus\H_{\rm sc}(H)\oplus \H_{\rm p}(H)$ (which depends on $H$) such that
for any $f\in \H_\bullet(H)$, the measure
$$
\B(\R)\ni \vartheta \mapsto \langle E(\vartheta)f,f\rangle \in \R
$$
is of the type $\bullet$, \ie~absolutely continuous, singular continuous or pure point.
It follows that the operator $H$ is reduced by this decomposition of the Hilbert space,
\ie~$H=H_{\rm ac}\oplus H_{\rm sc}\oplus H_{\rm p}$.
In other words, if one sets $E_{\rm ac}(H)$, $E_{\rm sc}(H)$ and $E_{\rm p}(H)$ for the orthogonal projections
on $\H_{\rm ac}(H)$, $\H_{\rm sc}(H)$ and $\H_{\rm p}(H)$ respectively, then these projections commute with $H$
and one has $H_{\rm ac}=H E_{\rm ac}(H)$, $H_{\rm sc}=H E_{\rm sc}(H)$ and $H_{\rm p}=H E_{\rm p}(H)$.
In addition, if $\sigma(H)$ denotes the spectrum of the operator $H$, we then set $\sigma_{\rm ac}(H):=\sigma(H_{\rm ac})$,
$\sigma_{\rm sc}(H):=\sigma(H_{\rm sc})$, and if $\sigma_{\rm p}(H)$ denotes the set of eigenvalues of $H$,
then the equality $\overline{\sigma_{\rm p}(H)} =\sigma(H_{\rm p})$ holds.
In this framework the operator $H$ is said to be purely absolutely continuous if $\H_{\rm sc}(H)= \H_{\rm p}(H) = \{0\}$,
or is said to have a finite point spectrum (counting multiplicity) if $\dim\big(\H_{\rm p}(H)\big)<\infty$.
In this case, we also write $\sharp\;\! \sigma_p(H)<\infty$.

Let us now move to scattering theory. It is a comparison theory, therefore we have to consider two self-adjoint operators
$H_0$ and $H$ in the Hilbert space $\H$. A few requirements will be imposed on these operators and on their relationships.
Let us first state these conditions, and discuss them afterwards.

\begin{Assumption}\label{Assum_1}
The following conditions hold for $H_0$ and $H$:
\begin{enumerate}
\item[(i)] $H_0$ is purely absolutely continuous,
\item[(ii)] $\sharp\;\! \sigma_p(H)<\infty$,
\item[(iii)] the wave operators
$\quad W_\pm:=s-\lim_{t \to\pm \infty} \e^{itH}\e^{-itH_0}\quad$ exist,
\item[(iv)] $\Ran(W_-)=\Ran(W_+)= \H_p(H)^\bot = \big(1-E_{\rm p}(H)\big)\H$.
\end{enumerate}
\end{Assumption}

The assumption \emph{(i)} is a rather common condition in scattering theory. Indeed, since
$H_0$ is often thought as a comparison operator, we expect it to be as simple as possible.
For that reason, any eigenvalue for $H_0$ is automatically ruled out.
In addition, since scattering theory does not really see any singular continuous spectrum,
we can assume without much loss of generality that $H_0$ does not possess such a component.
On the other hand, assumption \emph{(ii)}, which imposes that the point spectrum of $H$ is finite (multiplicity included)
is certainly restrictive, but is natural for our purpose.
Indeed, since at the end of the day we are looking for a relation involving the number of bound states,
the resulting equality is meaningful only if such a number is finite.

Assumption \emph{(iii)} is the main condition on the relation between $H_0$ and $H$.
In fact, this assumption does not directly compare these two operators, but compare their respective evolution group
$\{\e^{-itH_0}\}_{t\in \R}$ and $\{\e^{-itH}\}_{t\in \R}$ for $|t|$ large enough.
This condition is usually rephrased as \emph{the existence of the wave operators}.
Note that $s-\lim$ means the limit in the strong sense, \ie~when
these operators are applied on an element of the Hilbert space.
For a concrete model, checking this existence is a central part
of scattering theory, and can be a rather complicated task. We shall see in the examples developed later on
that this condition can be satisfied if $H$ corresponds to a suitable perturbation of $H_0$.
For the time being, imposing this existence corresponds in fact to the weakest condition necessary for the subsequent construction.
Finally, assumption \emph{(iv)} is usually called \emph{the asymptotic completeness} of the wave operators.
It is a rather natural expectation in the setting of scattering theory.
In addition, since $\Ran(W_\pm)\subset \H_{\rm ac}(H)$ always holds, this assumption
implies in particular that $H$ has no singular continuous spectrum, \ie~$\H_{\rm sc}(H)=\{0\}$.
The main idea behind this notion of asymptotic completeness will be explained in Remark \ref{rem_comp}.

Let us now stress some important consequences of Assumption \ref{Assum_1}.
Firstly, the wave operators $W_\pm$ are isometries, with
\begin{equation}\label{eq_range_W}
W_\pm^*W_\pm =1\quad \hbox{ and } \quad W_\pm W_\pm^*=1-E_{\rm p}(H),
\end{equation}
where ${}^*$ means the adjoint operator.
Secondly, $W_\pm$ are Fredholm operators and
satisfy the so-called \emph{intertwining relation}, namely $W_\pm \e^{-itH_0} = \e^{-itH}W_\pm$
for any $t\in \R$.
Another crucial consequence of our assumptions is that \emph{the scattering operator}
$$
S:=W_+^* W_-
$$
is unitary and commute with $H_0$, \ie~the relation $S\e^{-itH_0}=\e^{-itH_0}S$ holds for any $t\in \R$.
Note that this latter property means that $S$ and $H_0$ can be diagonalized simultaneously.
More precisely, from the general theory of self-adjoint operators, there exists
a unitary map $\F_0:\H\to\int_{\sigma(H_0)}^\oplus \H(\lambda)\;\!\d \lambda$ from $\H$ to a direct integral Hilbert space
such that $\F_0H_0\F_0^*=\int_{\sigma(H_0)}^\oplus \lambda \;\!\d \lambda$.
Then, the mentioned commutation relation implies that
$$
\F_0S\F_0^* = \int_{\sigma(H_0)}^\oplus S(\lambda) \;\!\d \lambda
$$
with $S(\lambda)$ a unitary operator in the Hilbert space $\H(\lambda)$ for almost every $\lambda$.
The operator $S(\lambda)$ is usually called \emph{the scattering matrix at energy $\lambda$}, even when this operator is not
a matrix but an operator acting in a infinite dimensional Hilbert space.

\begin{Remark}\label{rem_comp}
In order to understand the idea behind the asymptotic completeness, let us assume it and consider any $f\in \H_{\rm ac}(H)$.
We then set $f_\pm :=W_\pm^* f$ and observe that
\begin{align*}
\lim_{t\to\pm \infty}\big\| \e^{-itH}f-\e^{-itH_0}f_\pm \big\|
& = \lim_{t\to\pm \infty} \big\|f-\e^{itH}\e^{-itH_0}f_\pm\big\| \\
& = \lim_{t\to \pm \infty} \big\|f-\e^{itH}\e^{-itH_0}W_\pm^*f\big\| \\
& = 0,
\end{align*}
where the second equality in \eqref{eq_range_W} together with the equality $1-E_{\rm p}(H)=E_{\rm ac}(H)$
have been used for the last equality. Thus, the asymptotic completeness of the wave operators
means that for any $f \in \H_{\rm p}(H)^\bot$ the element $\e^{-itH}f$ can be well approximated by the simpler expression $\e^{-itH_0}f_\pm$
for $t$ going to $\pm \infty$. As already mentioned, one usually considers the operator $H_0$
simpler than $H$, and for that reason the evolution group $\{\e^{-itH_0}\}_{t\in \R}$
is considered simpler than the evolution group $\{\e^{-itH}\}_{t\in \R}$.
\end{Remark}

\section{The $C^*$-algebraic framework}\label{sec_Kth}
\setcounter{equation}{0}

In this section we introduce the $C^*$-algebraic framework which is necessary for interpreting Levinson's theorem as an index theorem.
We start by defining the $K$-groups for a $C^*$-algebra.

\subsection{The $K$-groups}\label{subsec_Kgroups}

Our presentation of the $K$-groups is mainly based the first chapters of the book \cite{RLL} to which we refer for details.

For any $C^*$-algebra $\EE$, let us denote by $\MM_n(\EE)$ the set of all $n\times n$ matrices
with entries in $\EE$. Addition, multiplication and involution for such matrices are mimicked
from the scalar case, \ie~when $\EE=\C$.
For defining a $C^*$-norm on $\MM_n(\EE)$, consider any injective $*$-morphism $\phi:\EE\to \B(\H)$
for some Hilbert space $\H$, and extend this morphism to a morphism $\phi:\MM_n(\EE)\to \B(\H^n)$ by defining
\begin{equation}\label{eq_extMn}
\phi
\left(\begin{matrix}
a_{11} & \dots & a_{1n} \\
\vdots  & \ddots & \vdots \\
a_{n1} & \dots &  a_{nn}
\end{matrix}\right)
\left(\begin{matrix}
f_1 \\ \vdots \\ f_n
\end{matrix}\right)
=
\left(\begin{matrix}
\phi(a_{11})f_1 + \dots +\phi(a_{1n})f_n \\
\vdots  \\
\phi(a_{n1})f_1 + \dots +\phi(a_{nn})f_n \\
\end{matrix}\right)
\end{equation}
for any ${}^t(f_1, \dots,  f_n)\in \H^n$ (the notation ${}^t(\dots)$ means the transpose of a vector).
Then a $C^*$-norm on $\MM_n(\EE)$ is obtained by setting $\|a\|:=\|\phi(a)\|$ for any $a\in \MM_n(\EE)$,
and this norm is independent of the choice of $\phi$.

In order to construct the first $K$-group
associated with $\EE$, let us consider the set
$$
\P_\infty(\EE) = \bigcup_{n\in \N} \P_n(\EE)
$$
with $\P_n(\EE):=\{p\in \MM_n(\EE)\mid p=p^* = p^2\}$. Such an element $p$ is called a projection.
$\P_\infty(\EE)$ is then endowed with a relation, namely for $p\in \P_n(\EE)$ and $q\in \P_m(\EE)$ one writes $p\sim_0 q$
if there exists $v\in \MM_{m,n}(\EE)$ such that $p=v^*v$ and $q=vv^*$.
Clearly, $\MM_{m,n}(\EE)$ denotes the set of $m\times n$ matrices with
entries in $\EE$, and the adjoint $v^*$ of $v\in \MM_{m,n}(\EE)$ is obtained by taking the transpose of the matrix,
and then the adjoint of each entry. This relation defines an equivalence relation
which combines the Murray-von Neumann equivalence relation together with an identification of projections in
different sized matrix algebras over $\EE$.
We also endow $\P_\infty(\EE)$ with a binary operation, namely if $p,q\in \P_\infty(\EE)$ we set
$p\oplus q= \left(\begin{smallmatrix} p & 0 \\ 0 & q\end{smallmatrix}\right)$ which
is again an element of $\P_\infty(\EE)$.

We can then define the quotient space
$$
\D(\EE):=\P_\infty(\EE)/\sim_0
$$
with its elements denoted by $[p]$ (the equivalence class containing $p\in \P_\infty(\EE)$).
One also sets
$$
[p]+[q]:= [p\oplus q]
$$
for any $p,q\in \P_\infty(\EE)$, and it turns out that the pair $\big(\D(\EE),+\big)$ defines an Abelian semigroup.

In order to obtain an Abelian group from the semigroup, let us recall that there exists a canonical construction which allows one to add ``the opposites''
to any Abelian semigroup and which is called \emph{the Grothendieck construction}.
More precisely, for an Abelian semigroup $(\D,+)$ we consider on $\D\times \D$ an equivalence relation,
namely $(a_1,b_1)\sim(a_2,b_2)$ if there exists $c\in \D$ such that $a_1 + b_2 +c= a_2 + b_1 + c$.
The elements of the quotient $\D\times \D/\sim$ are denoted by $\langle a,b\rangle$
and this quotient corresponds to an Abelian group with the addition
$$
\langle a_1,b_1\rangle + \langle a_2,b_2\rangle := \langle a_1 + a_2, b_1 + b_2\rangle.
$$
One readily checks that the equalities $-\langle a,b\rangle= \langle b,a\rangle$ and $\langle a,a\rangle = 0$ hold.
This group is called \emph{the Grothendieck group} associated with $(\D,+)$ and is denoted by $\big(\G(\D),+\big)$.

Coming back to a unital $C^*$-algebra $\EE$, we set
$$
K_0(\EE):=\G\big(\D(\EE)\big),
$$
which is thus an Abelian group with the binary operation $+$,
and define the map $[\cdot]_0: \P_\infty(\EE)\to K_0(\EE)$ by $[p]_0:=\langle [p]+[q],[q]\rangle$
for an arbitrary fixed $q\in \P_\infty(\EE)$. Note that this latter map is called the \emph{Grothendieck map} and is independent
of the choice of $q$.
Note also that an alternative description of $K_0(\EE)$ is provided by differences of equivalence classes of projections,
\ie
\begin{equation}\label{eq_stand1}
K_0(\EE)=\big\{[p]_0-[q]_0\mid p,q\in \P_\infty(\EE)\big\}.
\end{equation}
At the end of the day, we have thus obtained an Abelian group $\big(K_0(\EE),+\big)$ canonically associated with the unital $C^*$-algebra
$\EE$ and which is essentially made of equivalence classes of projections.

Before discussing the non-unital case, let us observe that if $\EE_1$, $\EE_2$ are unital $C^*$-algebras, and if
$\phi:\EE_1\to \EE_2$ is a $*$-morphism, then $\phi$ extends to a $*$-morphism $\MM_n(\EE_1)\to \MM_n(\EE_2)$, as already
mentioned just before \eqref{eq_extMn}. Since a $*$-morphism maps projections to projections, it follows that $\phi$ maps $\P_\infty(\EE_1)$
into $\P_\infty(\EE_2)$.
One can then infer from the universal property of the $K_0$-groups that $\phi$ defines a group homomorphism $K_0(\phi):K_0(\EE_1)\to K_0(\EE_2)$
given by
\begin{equation*}
K_0(\phi)([p]_0) = [\phi(p)]_0 \qquad \forall p\in \P_\infty(\EE_1).
\end{equation*}
The existence of this morphism will be necessary right now.

If $\EE$ is not unital, the construction is slightly more involved. Recall first that with any $C^*$-algebra $\EE$ (with or without a unit)
one can associate a unique unital $C^*$-algebra $\EE^+$ that contains $\EE$ as an ideal,
and such that the quotient $\EE^+/\EE$ is isomorphic to $\C$. We do not provide here this explicit construction, but
refer to \cite[Ex.~1.3]{RLL} for a detailed presentation. However, let us mention the fact that
the short exact sequence\footnote{A short exact sequence of $C^*$-algebras
$0 \to \JJ \stackrel{\iota}{\to} \EE \stackrel{q\;\;\!}{\to} \QQ \to 0$
consists in three $C^*$-algebras $\JJ,\EE$ and $\QQ$ and two $*$-morphisms $\iota$ and $q$ such that $\Im(\iota) = \Ker(q)$
and such that $\iota$ is injective while $q$ is surjective.}
$$
0\longrightarrow \EE\longrightarrow \EE^+\stackrel{\pi}{\longrightarrow} \C \longrightarrow 0
$$
is split exact, in the sense that if one sets $\lambda:\C\ni \alpha \mapsto \alpha 1_{\EE^+}\in \EE^+$,
then $\lambda$ is a $*$-morphism and the equality $\pi\big(\lambda(\alpha)\big)=\alpha$ holds for any $\alpha \in \C$.
Observe now that since $\pi:\EE^+\to \C$ is a $*$-morphism between unital $C^*$-algebras, it follows from
the construction made in the previous paragraph that there exists a group morphism $K_0(\pi):K_0(\EE^+)\to K_0(\C)$.
In the case of a non-unital $C^*$-algebra $\EE$, we set $K_0(\EE)$ for the kernel of this morphism $K_0(\pi)$,
which is obviously an Abelian group with the binary operation of $K_0(\EE^+)$. In summary:
$$
K_0(\EE):=\Ker\(K_0(\pi): K_0(\EE^+)\to K_0(\C)\)
$$
which is an Abelian group once endowed with the binary operation $+$ inherited from $K_0(\EE^+)$.

Let us still provide an alternative description of $K_0(\EE)$, in a way similar to the one provided in \eqref{eq_stand1},
but which holds both in the unital and in the non-unital case. For that purpose, let us introduce the scalar mapping $s:=\EE^+\to \EE^+$
obtained by the composition $\lambda \circ \pi$. Note that $\pi\big(s(a)\big) = \pi(a)$ and that $a-s(a)$ belongs to $\EE$
for any $a\in \EE^+$. As before, we keep the same notation for the extension of $s$ to $\MM_n(\EE^+)$. With these notations,
one has for any $C^*$-algebra $\EE$:
$$
K_0(\EE) = \big\{[p]_0-[s(p)]_0\mid p \in \P_\infty(\EE^+)\big\}.
$$

In summary, for any $C^*$-algebra (with or without unit) we have constructed
an Abelian group consisting essentially of equivalence classes of projections.
Since projections are not the only special elements in a $C^*$-algebra $\EE$, it is natural to wonder
if an analogous construction holds for other families of elements of $\EE$ ? The answer is yes, for families of unitary elements of $\EE$,
and fortunately this new construction is simpler.
The resulting Abelian group will be denoted by $K_1(\EE)$, and we are now going to describe how to obtain it.

In order to construct the second $K$-group
associated with a unital $C^*$-algebra $\EE$, let us consider the set
$$
\U_\infty(\EE) = \bigcup_{n\in \N} \U_n(\EE)
$$
with $\U_n(\EE):=\{u\in \MM_n(\EE)\mid u^*=u^{-1} \}$.
This set is endowed with
a binary operation, namely if $u,v\in \U_\infty(\EE)$ we set
$u\oplus v= \left(\begin{smallmatrix} u & 0 \\ 0 & v\end{smallmatrix}\right)$ which
is again an element of $\U_\infty(\EE)$.
We also introduce an equivalence relation on $\U_\infty(\EE)$: if $u\in \U_n(\EE)$ and $v\in \U_m(\EE)$, one sets
$u\sim_1  v$ if there exists a natural number $k\geq \max\{m,n\}$ such that $u\oplus 1_{k-n}$ is homotopic\footnote{
Recall that two elements $u_0,u_1\in \U_k(\EE)$ are homotopic in $\U_k(\EE)$, written $u_0\sim_h u_1$, if there exists
a continuous map $u:[0,1]\ni t\mapsto u(t)\in \U_k(\EE)$ such that $u(0)=u_0$ and $u(1)=u_1$.} to $v\oplus 1_{k-m}$
in $\U_k(\EE)$. Here we have used the notation $1_\ell$ for the identity matrix\footnote{The notation
$1_n$ for the identity matrix in $\MM_n(\EE)$ is sometimes very convenient, and sometimes very annoying (with $1$ much preferable).
In the sequel we shall use both conventions, and this should not lead to any confusion.}  in $\U_\ell(\EE)$.

Based on this construction, for any $C^*$-algebra $\EE$ one sets
$$
K_1(\EE):=\U_\infty(\EE^+)/\sim_1,
$$
and denotes the elements of $K_1(\EE)$ by $[u]_1$ for any $u\in \U_\infty(\EE^+)$.
$K_1(\EE)$ is naturally endowed with a binary operation, by setting for any $u,v\in \U_\infty(\EE^+)$
$$
[u]_1+[v]_1 := [u\oplus v]_1,
$$
which is commutative and associative.
Its zero element is provided by $[1]_1:=[1_n]_1$ for any natural number $n$, and one has
$-[u]_1 = [u^*]_1$ for any $u\in \U_\infty(\EE^+)$.
As a consequence, $\big(K_1(\EE),+\big)$ is an Abelian group, which corresponds to the second $K$-group of $\EE$.

In summary, for any $C^*$-algebra we have constructed
an Abelian group consisting essentially of equivalence classes of unitary elements.
As a result, any $C^*$-algebra is intimately linked with two Abelian groups, one based on projections and one based on unitary elements.
Before going to the next step of the construction, let us provide two examples of $K$-groups which can be
figured out without difficulty.

\begin{Example}\label{ex_K_groups}
\begin{enumerate}
\item[(i)]
Let $C(\S)$ denote the $C^*$-algebra of continuous function on the unit circle $\S$, with the $\linf$-norm,
and let us identify this algebra with $\big\{\zeta\in C([0,2\pi])\mid \zeta(0)=\zeta(2\pi)\big\}$, also endowed with the $\linf$-norm.
Some unitary elements of $C(\S)$ are provided for any $m\in \Z$ by the functions
$$
\zeta_m:[0,2\pi]\ni \theta \mapsto \e^{-i m \theta} \in \T.
$$
Clearly, for two different values of $m$ the functions $\zeta_m$ are not homotopic, and thus define different classes in $K_1\big(C(\S)\big)$.
With some more efforts one can show that these elements define in fact all elements of $K_1\big(C(\S)\big)$, and indeed one has
$$
K_1\big(C(\S)\big) \cong \Z.
$$
Note that this isomorphism is implemented by the winding number $\Wind(\cdot)$, which is roughly defined for any continuous function on $\S$ with values in $\T$
as the number of times this function turns around $0$ along the path from $0$ to $2\pi$. Clearly, for any $m\in \Z$ one has $\Wind(\zeta_m)=m$.
More generally, if $\det$ denotes the determinant on $\MM_n(\C)$
then the mentioned isomorphism is given by $\Wind\circ \det$ on $\U_n\big(C(\S)\big)$.

\item[(ii)] Let $\K(\H)$ denote the $C^*$-algebra of all compact operators on a infinite dimensional and separable Hilbert space $\H$.
For any $n$ one can consider the orthogonal projections on subspaces of dimension $n$ of $\H$, and these finite dimensional projections belong to $\K(\H)$.
It is then not too difficult to show that two projections of the same dimension are Murray-von Neumann equivalent,
while projections corresponding to two different values of $n$ are not.
With some more efforts, one shows that the dimension of these projections plays the crucial role for the definition of
$K_0\big(\K(\H)\big)$, and one has again
$$
K_0\big(\K(\H)\big) \cong \Z.
$$
In this case, the isomorphism is provided by the usual trace $\Tr$ on finite dimensional projections, and by the tensor product
of this trace with the trace $\tr$ on $\MM_n(\C)$. More precisely, on any element of $\P_n\big(\K(\H)\big)$ the mentioned isomorphism
is provided by $\Tr\circ \tr$.
\end{enumerate}
\end{Example}

\subsection{The boundary maps}\label{subsec_bound}

We shall now consider three $C^*$-algebras, with some relations between them. Since two $K$-groups
can be associated with each of them, we can expect that the relations between the algebras have
a counterpart between the $K$-groups. This is indeed the case.

Consider the short exact sequence of $C^*$-algebras
\begin{equation}\label{eq_short}
0 \to \JJ \stackrel{\iota}{\hookrightarrow} \EE \stackrel{q\;\;\!}{\to} \QQ \to 0
\end{equation}
where the notation $\hookrightarrow$ means that $\JJ$ is an ideal in $\EE$, and therefore $\iota$
corresponds to the inclusion map. In this setting,
$\QQ$ corresponds either to the quotient $\EE/\JJ$ or is isomorphic to this quotient.
The relations between the $K$-groups of these algebras can then be summarized with the following
six-term exact sequence
\begin{diagram}
K_1(\JJ) & \rTo & K_1(\EE) & \rTo & K_1(\QQ) \\
\uTo^{{\rm exp}} & &  & & \dTo_{\rm ind} \\
K_0(\QQ) & \lTo & K_0(\EE) & \lTo&  K_0(\JJ)\ .
\end{diagram}
In this diagram, each arrow corresponds to a group morphism, and the range of an arrow is equal to the kernel of the following one.
Note that we have indicated the name of two special arrows, one is called \emph{the exponential map}, and the other one \emph{the index map}.
These two arrows are generically called \emph{boundary maps}.
In this paper, we shall only deal with the index map, but let us mention that the exponential map has also played a central role
for exhibiting other index theorems in the context of solid states physics \cite{Johannes,KS}.

We shall not recall the construction of the index map in the most general framework,
but consider a slightly restricted setting (see \cite[Chap.~9]{RLL} for a complete presentation).
For that purpose, let us assume that the algebra $\EE$ is unital, in which case
$\QQ$ is unital as well and the morphism $q$ is unit preserving.
Then, a reformulation of \cite[Prop.~9.2.4.(ii)]{RLL} in our context reads:

\begin{Proposition}\label{prop_top_lev_0}
Consider the short exact sequence \eqref{eq_short} with $\EE$ unital.
Assume that $\Gamma$ is a unitary element of $\MM_n(\QQ)$ and that there exists a partial isometry
$W\in \MM_n(\EE)$ such that $q(W)=\Gamma$. Then $1_n-W^*W$
and $1_n-W W^*$ are projections in $\MM_n(\JJ)$, and
$$
\ind([\Gamma]_1) = [1_n-W^*W]_0-[1_n-W W^*]_0\ .
$$
\end{Proposition}

Let us stress the interest of this statement. Starting from a unitary element $\Gamma$ of $\MM_n(\QQ)$, one can
naturally associate to it an element of $K_0(\JJ)$.
In addition, since the elements of the $K$-groups are made of equivalence classes of objects,
such an association is rather stable under small deformations.

Before starting with applications of this formalism to scattering systems, let us
add one more reformulation of the previous proposition.
The key point in the next statement is that the central role is played by the partial isometry $W$
instead of the unitary element $\Gamma$. In fact, the following statement is at the root of our
topological approach of Levinson's theorem.

\begin{Proposition}\label{prop_top_lev}
Consider the short exact sequence \eqref{eq_short} with $\EE$ unital.
Let $W$ be a partial isometry in $\MM_n(\EE)$ and assume that $\Gamma:= q(W)$ is a unitary element of $\MM_n(\QQ)$.
Then $1_n-W^* W$ and $1_n-W W^*$ are projections in $\MM_n(\JJ)$, and
$$
\ind([q(W)]_1) = [1_n-W^* W]_0-[1_n-W W^*]_0\ .
$$
\end{Proposition}

\subsection{The abstract topological Levinson's theorem}\label{subsec_index}

Let us now add the different pieces of information we have presented so far, and get an abstract version of our Levinson's theorem.
For that purpose, we consider a separable Hilbert space $\H$ and a unital $C^*$-subalgebra $\EE$ of $\B(\H)$ which contains
the ideal of $\K(\H)$ of compact operators.
We can thus look at the short exact sequence of $C^*$-algebras
\begin{equation*}
0 \to \K(\H) \hookrightarrow \EE \stackrel{q\;\;\!}{\to} \EE/\K(\H) \to 0.
\end{equation*}
If we assume in addition that $\EE/\K(\H)$ is isomorphic to $C(\S)$, and if we take
the results presented in Example \ref{ex_K_groups} into account, one infers that
$$
\Z\cong K_1\big(C(\S)\big) \stackrel{\ind}{\longrightarrow} K_0\big(\K(\H)\big) \cong \Z
$$
with the first isomorphism realized by the winding number and the second isomorphism realized by the trace.
As a consequence, one infers from this together with Proposition \ref{prop_top_lev} that there exists $n\in \Z$
such that for any partial isometry $W\in \EE$ with unitary $\Gamma:=q(W)\in C(\S)$ the following equality holds:
\begin{equation*}
\Wind(\Gamma) = n \Tr\big([1-W^* W]-[1-W W^*]\big).
\end{equation*}
We emphasize once again that the interest in this equality is that the left hand side is independent of the choice of any special
representative in $[\Gamma]_1$.
Let us also mention that the number $n$ depends on the choice of the extension of $\K(\H)$ by $C(\S)$, see \cite[Chap.~3.2]{WO},
but also on the convention chosen for the computation of the winding number.

If we summarize all this in a single statement, one gets:

\begin{Theorem}[Abstract topological Levinson's theorem]\label{thm_Lev}
Let $\H$ be a separable Hilbert space, and let $\EE\subset \B(\H)$ be a unital $C^*$-algebra such that
$\K(\H)\subset \EE$ and $\EE/\K(\H)\cong C(\S)$ (with quotient morphism denoted by $q$). Then there exists $n\in \Z$
such that for any partial isometry $W\in \EE$ with unitary $\Gamma:=q(W)\in C(\S)$ the following equality holds:
\begin{equation}\label{eq_Lev_2}
\Wind(\Gamma) = n \Tr\big([1-W^* W]-[1-W W^*]\big).
\end{equation}
In particular if $W= W_-$ for some scattering system satisfying Assumption \ref{Assum_1}, the previous equality reads
$$
\Wind\big(q(W_-)\big) = -n \Tr\big([E_{\rm p}]\big).
$$
\end{Theorem}

Note that in applications, the factor $n$ will be determined by computing both sides of the equality on an explicit example.

\subsection{The leading example}\label{subsec_example}

We shall now provide a concrete short exact sequence of $C^*$-algebras, and illustrate
the previous constructions on this example.

In the Hilbert space $\ltwo(\R)$ we consider the two canonical self-adjoint operators $X$ of
multiplication by the variable, and $D=-i\frac{\d}{\d x}$ of differentiation.
These operators satisfy the canonical commutation relation written formally $[iD,X]=1$, or more precisely
$\e^{-isX}\e^{-itD} = \e^{-ist}\e^{-itD}\e^{-isX}$.
We recall that the spectrum of both operators is $\R$. Then, for any functions
$\varphi, \eta\in \linf(\R)$, one can consider by bounded functional calculus
the operators $\varphi(X)$ and $\eta(D)$ in $\B\big(\ltwo(\R)\big)$.
And by mixing some operators $\varphi_i(X)$ and $\eta_i(D)$ for suitable functions $\varphi_i$ and $\eta_i$,
we are going to produce an algebra $\EE$ which will be useful in many applications. In fact, the first algebras
which we are going to construct have been introduced in \cite{GI} for a different purpose, and these algebras
have been an important source of inspiration for us.
We also  mention that related algebras had already been introduced a long time ago in
\cite{BC1,BC2, Cordes,CH}.

Let us consider the closure in $\B\big(\ltwo(\R)\big)$ of the $*$-algebra generated by
elements of the form $\varphi_i(D)\eta_i(X)$, where $\varphi_i, \eta_i$ are continuous functions on $\R$ which have limits at $\pm \infty$.
Stated differently, $\varphi_i, \eta_i$ belong to $C(\bR)$. Note that this algebra is clearly unital.
In the sequel, we shall use the following notation:
$$
\EE_{(D,X)}:=C^*\Big(\varphi_i(D)\eta_i(X)\mid \varphi_i,\eta_i\in C(\bR)\Big).
$$
Let us also consider the $C^*$-algebra generated by $\varphi_i(D)\eta_i(X)$ with $\varphi_i,\eta_i\in C_0(\R)$,
which means that these functions are continuous and vanish at $\pm \infty$.
As easily observed, this algebra is a closed ideal in $\EE_{(D,X)}$ and is equal to the $C^*$-algebra
$\K\big(\ltwo(\R)\big)$ of compact operators in $\ltwo(\R)$, see for example \cite[Corol.~2.18]{GI}.

Implicitly, the description of the quotient $\EE_{(D,X)}/\K\big(\ltwo(\R)\big)$ has already been provided in Section \ref{sec_baby}.
Let us do it more explicitly now.
We consider the square $\blacksquare:=\bR\times \bR$ whose boundary $\square$ is the
union of four parts: $\square =C_1\cup C_2\cup C_3\cup C_4$, with $C_1 = \{-\infty\}\times \bR$, $C_2 = \bR \times \{+\infty\}$,
$C_3 = \{+\infty\}\times \bR$ and $C_4 = \bR \times \{-\infty\}$. We can also view $C(\square)$ as the subalgebra of
\begin{equation}\label{eq_sum_alg2}
C(\bR)\oplus C(\bR)\oplus C(\bR)\oplus C(\bR)
\end{equation}
given by elements
$\Gamma:=(\Gamma_1,\Gamma_2,\Gamma_3,\Gamma_4)$ which coincide at the
corresponding end points, that is,
$\Gamma_1(+\infty)=\Gamma_2(-\infty)$, $\Gamma_2(+\infty) = \Gamma_3(+\infty)$, $\Gamma_3(-\infty) = \Gamma_4(+\infty)$,
and $\Gamma_4(-\infty) = \Gamma_1(-\infty)$.
Then $\EE_{(D,X)}/\K\big(\ltwo(\R)\big)$ is isomorphic to
$C(\square)$, and if we denote the quotient map by
$$
q: \EE_{(D,X)}\to \EE_{(D,X)}/\K\big(\ltwo(\R)\big) \cong  C(\square)
$$
then the image $q\big(\varphi(D)\eta(X)\big)$ in \eqref{eq_sum_alg2}
is given by $\Gamma_1 = \varphi(-\infty)\eta(\cdot)$, $\Gamma_{2} =
\varphi(\cdot)\eta(+\infty)$, $\Gamma_{3} = \varphi(+\infty)\eta(\cdot)$ and $\Gamma_{4} =
\varphi(\cdot)\eta(-\infty)$. Note that this isomorphism is proved in \cite[Thm.~3.22]{GI}.
In summary, we have obtained the short exact sequence
$$
0 \to \K\big(\ltwo(\R)\big) \hookrightarrow \EE_{(D,X)} \stackrel{q\;\;\!}{\to} C(\square) \to 0
$$
with $\K\big(\ltwo(\R)\big)$ and $\EE_{(D,X)}$ represented in $\B\big(\ltwo(\R)\big)$, but with $C(\square)$
which is not naturally represented in $\B\big(\ltwo(\R)\big)$.
Note however that each of the four functions summing up in an element of $C(\square)$ can individually be represented in
$\B\big(\ltwo(\R)\big)$, either as a multiplication operator or as a convolution operator.

We shall now construct several isomorphic versions of these algebras. Indeed, if one looks
back at the baby model, the wave operator is expressed in \eqref{eq_wave_baby}
with bounded functions of the two operators $H_{\rm D}$ and $A$, but not in terms of $D$ and $X$.
In fact, we shall first use a third pair of operators, namely $L$ and $A$, acting in $\ltwo(\R_+)$,
and then come back to the pair $(H_{\rm D},A)$ also acting in $\ltwo(\R_+)$.

Let us consider the Hilbert space $\ltwo(\R_+)$, and as in \eqref{eq_dil_group} the action of the dilation group
with generator $A$. Let also $B$ be the operator of multiplication in $\ltwo(\R_+)$ by the function $-\ln$,
\ie~$[Bf](\lambda)=-\ln(\lambda)f(\lambda)$ for any $f\in C_{\rm c}(\R_+)$ and $\lambda\in \R_+$.
Note that if one sets $L$ for the self-adjoint operator of multiplication by the variable in $\ltwo(\R_+)$, \ie
\begin{equation}\label{eq_def_L}
[Lf](\lambda):= \lambda f(\lambda)\qquad f\in C_{\rm c}(\R_+) \hbox{ and }\lambda \in \R_+ ,
\end{equation}
then one has $B=-\ln(L)$.
Now, the equality $[iB,A]=1$ holds (once suitable defined), and the relation between the pair of operators $(D,X)$ in $\ltwo(\R)$
and the pair $(B,A)$ in $\ltwo(\R_+)$ is well-known and corresponds to the Mellin transform.
Indeed, let $\V:\ltwo(\R_+)\to \ltwo(\R)$ be defined by $(\V f)(x):=\e^{x/2}f(\e^x)$ for $x\in\R$, and remark that $\V$ is a
unitary map with adjoint $\V^*$ given by $(\V^*g)(\lambda)=\lambda^{-1/2}g(\ln \lambda)$ for $\lambda\in\R_+$.
Then, the Mellin transform $\M:\ltwo(\R_+)\to \ltwo(\R)$ is defined by $\M:=\F\V$ with $\F$ the usual unitary
Fourier transform\footnote{For $f\in C_{\rm c}(\R)$ and $x\in \R$ we set $[\F f](x)=(2\pi)^{-1/2}\int_\R \e^{-ixy}f(y)\;\!\d y$.} in $\ltwo(\R)$.
The main property of $\M$ is that it diagonalizes the generator of dilations, namely,
$\M A \M^*=X$. Note that one also has $\M B \M^*=D$.

Before introducing a first isomorphic algebra, observe that if $\eta\in C(\bR)$, then
$$
\M^* \eta(D)\M = \eta(\M^* D\M) = \eta(B) = \eta\big(-\ln(L)\big) \equiv \psi(L)
$$
for some $\psi \in C(\bRp)$.
Thus, by taking these equalities into account, it is natural to define in $\B\big(\ltwo(\R_+)\big)$ the $C^*$-algebra
$$
\EE_{(L,A)}:=C^*\Big(\psi_i(L)\eta_i(A)\mid \psi_i\in C(\bRp)\hbox{ and }\eta_i\in C(\bR)\Big),
$$
and clearly this algebra is isomorphic to the $C^*$-algebra $\EE_{(D,X)}$ in $\B\big(\ltwo(\R)\big)$.
Thus, through this isomorphism one gets again a short exact sequence
$$
0 \to \K\big(\ltwo(\R_+)\big) \hookrightarrow \EE_{(L,A)} \stackrel{q\;\;\!}{\to} C(\square) \to 0
$$
with the square $\square$ made of the four parts $\square = B_1\cup B_2 \cup B_3 \cup B_4$
with $B_1 = \{0\}\times[-\infty,+\infty]$, $B_2 =[0,+\infty]\times\{+\infty\}$, $B_3 = \{+\infty\}\times[-\infty,+\infty]$,
and $B_4=[0,+\infty]\times\{-\infty\}$. In addition, the algebra $C(\square)$ of continuous functions on $\square$
can be viewed as a subalgebra of
\begin{equation}\label{eq_sum_alg3}
C\big([-\infty,+\infty]\big)\oplus C\big([0,+\infty]\big)\oplus C\big([-\infty,+\infty]\big)\oplus C\big([0,+\infty]\big)
\end{equation}
given by elements $\Gamma:=(\Gamma_1, \Gamma_2, \Gamma_3, \Gamma_4)$ which coincide at the corresponding end points, that
is, $\Gamma_1(+\infty) = \Gamma_2(0)$, $\Gamma_2(+\infty) = \Gamma_3(+\infty)$, $\Gamma_3(-\infty) = \Gamma_4(+\infty)$,
and $\Gamma_4(0) = \Gamma_1(-\infty)$.

Finally, if one sets $\F_{\rm s}$ for the unitary Fourier sine transformation in $\ltwo(\R_+)$, as recalled in \eqref{Fouriers},
then the equalities $-A = \F_{\rm s}^* A \F_{\rm s}$ and $\sqrt{H_{\rm D}} = \F_{\rm s}^* L \F_{\rm s}$ hold, where $H_{\rm D}$ corresponds to
the Dirichlet Laplacian on $\R_+$ introduced in Section \ref{sec_baby}. As a consequence, note that the formal equality
$[i\12 \ln(H_{\rm D}),A]=1$ can also be fully justified. Moreover, by using this new unitary transformation
one gets that the $C^*$-subalgebra of $\B\big(\ltwo(\R_+)\big)$ defined by
\begin{equation}\label{alg_A_H_D}
\EE_{(H_{\rm D},A)}:=C^*\Big(\psi_i(H_{\rm D})\eta_i(A)\mid \psi_i\in C(\bRp) \hbox{ and } \varphi_i\in C(\bR)\Big),
\end{equation}
is again isomorphic to $\EE_{(D,X)}$, and that the quotient $\EE_{(H_{\rm D},A)}/\K\big(\ltwo(\R_+)\big)$ can naturally be viewed
as a subalgebra of the algebra introduced in \eqref{eq_sum_alg3} with similar compatibility conditions.
Let us mention that if the Fourier cosine transformation $\F_{\rm c}$ had been chosen instead of $\F_{\rm s}$
(see \eqref{Fourierc} for the definition of $\F_{\rm c}$) an isomorphic
algebra $\EE_{(H_{\rm N},A)}$ would have been obtained, with $H_{\rm N}$ the Neumann Laplacian on $\R_+$.

\begin{Remark}
Let us stress that the presence of some minus signs in the above expressions, as for example in
$B =-\ln(L)$ or in $-A = \F_{\rm s}^* A \F_{\rm s}$, are completely harmless and unavoidable.
However, one can not simply forget them because they play a (minor) role in the conventions related to the computation of the winding number.
\end{Remark}

\subsection{Back to the baby model}\label{subsec_Teddy}

Let us briefly explain how the previous framework can be used in the context of the baby model.
This will also allow us to compute explicitly the value of $n$ in Theorem \ref{thm_Lev}.

We consider the Hilbert space $\ltwo(\R_+)$ and the unital $C^*$-algebra $\EE_{(H_{\rm D},A)}$ introduced in \eqref{alg_A_H_D}.
Let us first observe that the wave operator $W_-^\alpha$ of \eqref{eq_wave_baby} is an isometry which
clearly belongs to the $C^*$-algebra $\EE_{(H_{\rm D},A)}\subset \B\big(\ltwo(\R_+)\big)$. In addition, the image of $W_-^\alpha$
in the quotient algebra $\EE_{(H_{\rm D},A)}/\K\big(\ltwo(\R_+)\big)\cong C(\square)$
is precisely the function $\Gammaq^\alpha$, defined in
\eqref{eq_baby_Gamma_1} for $\alpha\neq 0$ and in \eqref{eq_baby_Gamma_2} for $\alpha=0$, which are unitary elements of $C(\square)$.
Finally, since $C(\square)$ and $C(\S)$ are clearly isomorphic, the winding number
$\Wind(\Gammaq^\alpha)$ of $\Gammaq^\alpha$ can be computed, and in fact this has been performed and recorded in the table
of Section \ref{sec_baby}.

On the other hand, it follows from \eqref{eq_range_W} that $1-(W_-^\alpha)^*W_-^\alpha = 0$ and that $1-W_-^\alpha (W_-^\alpha)^*=E_p(H^\alpha)$,
which is trivial if $\alpha\geq 0$ and which is a projection of dimension $1$ if $\alpha<0$.
It follows that
\begin{equation}\label{eq_Tr_explicit}
\Tr\big([1-(W_-^\alpha)^* W_-^\alpha]-[1-W_-^\alpha (W_-^\alpha)^*]\big) = -\Tr\big(E_p(H^\alpha)\big) =
\left\{\begin{matrix} -1 & \hbox{ if } \alpha<0 \;\!,\\ 0 & \hbox{ if } \alpha\geq 0\;\!. \end{matrix}\right.
\end{equation}

Thus, this example fits in the framework of Theorem \ref{thm_Lev}, and in addition both sides
of \eqref{eq_Lev_2} have been computed explicitly.
By comparing \eqref{eq_Tr_explicit} with the results obtained for $\Wind(\Gammaq^\alpha)$, one gets that the factor $n$
mentioned in \eqref{eq_Lev_2} is equal to $-1$ for these algebras.
Finally, since $E_p(H^\alpha)$ is related to the point spectrum of $H^\alpha$, the content of Proposition \ref{prop_baby}
can be rewritten as
$$
\Wind(\Gammaq^\alpha) = \sharp \sigma_{\rm p}(H^\alpha).
$$
This equality corresponds to a topological version of Levinson's theorem for the baby model.
Obviously, this result was already obtained in Section \ref{sec_baby} and all the above framework was not necessary
for its derivation. However, we have now in our hands a very robust framework which will be applied to several
other situations.

\section{Quasi $1$D examples}\label{Sec_1D}
\setcounter{equation}{0}

In this section, we gather various examples of scattering systems which can be recast in
the framework introduced in the previous section. Several topological versions of Levinson's theorem will be deduced
for these models. Note that we shall avoid in this section the technicalities required for
obtaining more explicit formulas for the wave operators.
An example of such a rather detailed proof will be provided for Schr\"odinger operators on $\R^3$.

\subsection{Schr\"odinger operator with one point interaction}\label{subsec_point}

In this section we recall the results which have been obtained for
Schr\"odinger operators with one point interaction. In fact, such operators were the
first ones on which the algebraic framework has been applied. More information
about this model can be found in \cite{KR_PI}. Note that the construction
and the results depend on the space dimension, we shall therefore present successively the results in dimension
$1$, $2$ and $3$. However, even in dimension $2$ and $3$, the problem is essentially one dimensional, as we shall observe.

Let us consider the Hilbert space $\ltwo(\R^d)$ and the operator $H_0=-\Delta$ with domain the Sobolev space $\H^2(\R^d)$.
For the operator $H$ we shall consider the perturbation of $H_0$ by a one point interaction located at
the origin of $\R^d$. We shall not recall the precise definition of a one point interaction since
this subject is rather well-known, and since the literature on the subject is easily accessible. Let us
just mention that such a perturbation of $H_0$ corresponds to the addition of a boundary condition at $0\in \R^d$
which can be parameterized by a single real parameter family in $\R^d$ for $d=2$ and $d=3$. In dimension
$1$ a four real parameters family is necessary for describing all corresponding operators.
In the sequel and in dimension $1$ we shall deal only with either a so-called $\delta$-interaction
or a $\delta'$-interaction. We refer for example to the monograph \cite{AGHH} for a thorough presentation
of operators with a finite or an infinite number of point interactions.

Beside the action of dilations in $\ltwo(\R_+)$, we shall often use the dilation groups in $\ltwo(\R^d)$ whose action is defined by
$$
[U_t f](x) = \e^{d t/2} f\big(\e^t x\big), \qquad f\in \ltwo(\R^d), \ x \in \R^d.
$$
Generically, its generator will be denoted by $A$ in all these spaces.

\subsubsection{The dimension $d=1$}\label{subsec_even_odd}

For any $\alpha, \beta \in \R$, let us denote by $H^\alpha$ the operator in $\ltwo(\R)$ which formally corresponds to $H_0+\alpha \delta$
and by $H^\beta$ the operator which formally corresponds to $H_0+\beta \delta'$.
Note that for $\alpha<0$ and for $\beta<0$ the operators $H^\alpha$ and $H^\beta$ have both a single eigenvalue of multiplicity one,
while for $\alpha\geq 0$ and for $\beta \geq 0$ the corresponding operators have no eigenvalue.
It is also known that the wave operators $W_\pm^\alpha$ for the pair $(H^\alpha, H_0)$ exist, and that
the wave operators $W_\pm^\beta$ for the pair $(H^\beta, H_0)$ also exist.
Some explicit expressions for them have been computed in \cite{KR_PI}.

\begin{Lemma}
For any $\alpha, \beta \in \R$ the following equalities hold in $\B\big(\ltwo(\R)\big)$~:
\begin{align*}
W_-^\alpha & = 1 + \12\big(1+\tanh(\pi A)+i\cosh(\pi A)^{-1}\big) \Big[\frac{2\sqrt{H_0}-i\alpha}{2\sqrt{H_0}+i\alpha}-1\Big] P_{\rm e}, \\
W_-^\beta & = 1 + \12\big(1+\tanh(\pi A)-i\cosh(\pi A)^{-1}\big) \Big[\frac{2+i\beta\sqrt{H_0}}{2-i\beta\sqrt{H_0}}-1\Big] P_{\rm o},
\end{align*}
where $P_{\rm e}$ denotes the projection onto the set of even functions of $\ltwo(\R)$,
while $P_{\rm o}$ denotes the projection onto the set of odd functions of $\ltwo(\R)$.
\end{Lemma}

In order to come back precisely to the framework introduced in Section \ref{sec_Kth},
we need to introduce the even$\;\!$/$\;\!$odd representation of $\ltwo(\R)$. Given any
function $m$ on $\R$, we write $m_{\rm e}$ and $m_{\rm o}$ for the even part and the
odd part of $m$. We also set $\Hrond:=\ltwo(\R_+;\C^2)$ and
introduce the unitary map $\UU:\ltwo(\R) \to \Hrond$ given on any $f\in\ltwo(\R),~\big(\begin{smallmatrix}
f_1\\
f_2
\end{smallmatrix}\big)\in\Hrond,~x\in\R$ by
\begin{equation}\label{eq_e_o}
\UU f:=\sqrt2
\left(\begin{smallmatrix}
f_{\rm e}\\
f_{\rm o}
\end{smallmatrix}\right)\quad{\rm and}\quad
\big[\UU^*
\big(\begin{smallmatrix}
f_1\\
f_2
\end{smallmatrix}\big)\big](x)
:=\textstyle\frac1{\sqrt2}
[f_1(|x|)+\sgn(x)f_2(|x|)].
\end{equation}
Now, observe that if $m$ is a function on $\R$ and $m(X)$ denotes the corresponding multiplication operator on $\ltwo(\R)$,
then we have
$$
\UU m(X)\UU^*=
\left(\begin{smallmatrix}
m_{\rm e}(L)~ & m_{\rm o}(L)\\
m_{\rm o}(L)~ & m_{\rm e}(L)
\end{smallmatrix}\right)
$$
where $L$ is the operator of multiplication by the variable in $\ltwo(\R_+)$ already introduced in \eqref{eq_def_L}.

By taking these formulas and the previous lemma into account, one gets
\begin{align*}
\UU W_-^\alpha\UU^* & =
\left(\begin{smallmatrix} 1+\12\big(1+\tanh(\pi A)+i\cosh(\pi A)^{-1}\big) \big[\frac{2\sqrt{H_{\rm N}}-i\alpha}{2\sqrt{H_{\rm N}}+i\alpha}-1\big]  & \ 0 \\
0 &\  1\end{smallmatrix}\right) \\
\UU W_-^\beta  \UU^*& =
\left(\begin{smallmatrix}  1 &\  0 \\ 0 & \ 1+ \12\big(1+\tanh(\pi A)-i\cosh(\pi A)^{-1}\big) \big[\frac{2+i\beta\sqrt{H_{\rm D}}}{2-i\beta\sqrt{H_{\rm D}}}-1\big]
\end{smallmatrix}\right)
\end{align*}
It clearly follows from these formulas that $\UU W_-^\alpha\UU^* \in \MM_2\big(\EE_{(H_{\rm N},A)}\big)$ and
$\UU W_-^\beta \UU^* \in \MM_2\big(\EE_{(H_{\rm D},A)}\big)$, and as a consequence
the algebraic framework introduced in Section \ref{sec_Kth} can be applied straightforwardly.
In particular, one can define the functions $\Gammaq^\alpha$, $\Gammaq^\beta$ as the image of
$\UU W_-^\alpha\UU^* $ and $\UU W_-^\beta  \UU^*$ in the respective quotient algebras, and get:

\begin{Corollary}
For any $\alpha, \beta \in \R^*$, one has
\begin{align*}
\Gammaq^\alpha & =
\Big(
\left(\begin{smallmatrix} 1+\frac{1}{2}(1+\tanh(\pi \cdot)+i\cosh(\pi \cdot)^{-1}) [s^\alpha(0)-1]  & \ 0 \\
0 &\  1\end{smallmatrix}\right), \left(\begin{smallmatrix} s^\alpha(\cdot)  & \ 0 \\ 0 &\  1\end{smallmatrix}\right), \\
&\qquad\qquad\qquad\qquad
\left(\begin{smallmatrix} 1+\frac{1}{2} (1+\tanh(\pi \cdot)+i\cosh(\pi \cdot)^{-1}) [s^\alpha(+\infty)-1]  & \ 0 \\
0 &\  1\end{smallmatrix}\right),
\left(\begin{smallmatrix} 1  & 0 \\ 0 & 1\end{smallmatrix}\right) \Big)
\end{align*}
with $s^\alpha(\cdot) = \frac{2\sqrt{\cdot}-i\alpha}{2\sqrt{\cdot}+i\alpha}$,
\begin{align*}
\Gammaq^\beta  & =
\Big(
\left(\begin{smallmatrix}  1 &\  0 \\ 0 & \ 1+ \frac{1}{2}(1+\tanh(\pi \cdot)-i\cosh(\pi \cdot)^{-1})[s^\beta(0)-1]\end{smallmatrix}\right),
\left(\begin{smallmatrix}  1 &\  0 \\ 0 & \ s^\beta(\cdot)\end{smallmatrix}\right), \\
&\qquad\qquad\qquad\qquad
\left(\begin{smallmatrix}  1 &\  0 \\ 0 & \ 1+ \frac{1}{2}(1+\tanh(\pi \cdot)-i\cosh(\pi \cdot)^{-1}) [s^\beta(+\infty)-1]\end{smallmatrix}\right),
\left(\begin{smallmatrix} 1  & 0 \\ 0 & 1\end{smallmatrix}\right) \Big)
\end{align*}
with $s^\beta(\cdot)= \frac{2+i\beta\sqrt{\cdot}}{2-i\beta\sqrt{\cdot}}$, and $\Gammaq^0=(1_2,1_2,1_2,1_2)$ (both for $\alpha=0$ and $\beta = 0$).
In addition, one infers that for any $\alpha, \beta\in \R$:
$$
\Wind(\Gammaq^\alpha)=\sharp\sigma_{\rm p}(H^\alpha), \qquad \hbox{ and } \qquad
\Wind(\Gammaq^\beta)=\sharp\sigma_{\rm p}(H^\beta).
$$
\end{Corollary}

\begin{Remark}\label{rem_PI}
Let us mention that another convention had been taken in \cite{KR_PI} for the computation of the winding number,
leading to a different sign in the previous equalities.
Note that the same remark holds for equations \eqref{equ1} and \eqref{equ2} below.
\end{Remark}

\subsubsection{The dimension $d=2$}

As already mentioned above, in dimension $2$ there is only one type of self-adjoint extensions, and thus
only one real parameter family of operators $H^\alpha$ which formally correspond to $H_0+\alpha \delta$.
The main difference with dimensions $1$ and $3$ is that $H^\alpha$ always possesses a single eigenvalue of
multiplicity one. As before, the wave operators $W_\pm^\alpha$ for the pair $(H^\alpha,H_0)$ exist, and it has been
shown in the reference paper that:

\begin{Lemma}
For any $\alpha \in \R$ the following equality holds:
\begin{equation*}
W_-^\alpha = 1 + \12\big(1+\tanh(\pi A/2)\big) \Big[\frac{2\pi\alpha -\Psi(1)-\ln(2)+\ln(\sqrt{H_0})+i\pi/2}{2\pi\alpha -\Psi(1)-\ln(2)+\ln(\sqrt{H_0})-i\pi/2}-1\Big] P_0,
\end{equation*}
where $P_0$ denotes the projection on the rotation invariant functions of $\ltwo(\R^2)$,
and where $\Psi$ corresponds to the digamma function.
\end{Lemma}

Note that in this formula, $A$ denotes the generator of dilations in $\ltwo(\R^2)$.
It is then sufficient to restrict our attention to $P_0 \ltwo(\R^2)$ since the subspace of $\ltwo(\R^2)$
which is orthogonal to $P_0\ltwo(\R^2)$ does not play any role for this model
(and it is the reason why this model is quasi one dimensional).
Thus, let us introduce the unitary map $\UU : P_0 \ltwo(\R^2)\to \ltwo(\R_+, r\d r)$ defined by $[\UU f](r):= \sqrt{2\pi} f(r)$
which is well defined since $f\in P_0 \ltwo(\R^2)$ depends only on the radial coordinate.
Since the dilation group as well as the operator $H_0$ leave the subspace $P_0 \ltwo(\R^2)$ of $\ltwo(\R^2)$ invariant, one gets
in $\ltwo(\R_+, r\d r)$~:
\begin{equation}\label{eq_2D}
\UU W_-^\alpha P_0\UU^* =
1 + \12\big(1+\tanh(\pi A/2)\big) \Big[\frac{2\pi\alpha -\Psi(1)-\ln(2)+\ln(\sqrt{H_0})+i\pi/2}{2\pi\alpha -\Psi(1)-\ln(2)+\ln(\sqrt{H_0})-i\pi/2}-1\Big].
\end{equation}

\begin{Remark}
Let us stress that the above formula does not take place in any of the representations introduced in
Section \ref{subsec_example} but in a unitarily equivalent one. Indeed, one can come back to the algebra $\EE_{(L,A)}$
by using the spectral representation of $H_0$.
More precisely let us first introduce $\F_0:\ltwo(\R^2)\to \ltwo\big(\R_+;\ltwo(\S)\big)$
defined by
$$
\big([\F_0 f](\lambda)\big)(\omega)=2^{-1/2} [\F f](\sqrt{\lambda}\omega), \qquad f\in C_{\rm c}(\R^2),
\ \lambda \in \R_+, \ \omega\in \S
$$
with $\F$ the unitary Fourier transform in $\ltwo(\R^2)$,
and recall that $[\F_0 H_0 f](\lambda) = \lambda [\F_0 f](\lambda)$ for any $f\in \H^2(\R^2)$
and a.e.~$\lambda \in \R_+$.
Then, if one defines the unitary map $\UU' : P_0 \ltwo(\R^2)\to \ltwo(\R_+)$ by $[\UU' f](\lambda):= \sqrt{\pi} [\F f](\sqrt{\lambda})$,
one gets $\UU' H_0 \UU'^* = L$, and a short computation using the dilation group
in $\ltwo(\R^2)$ and in $\ltwo(\R_+)$ leads to the relation $\UU' A \UU'^* = -2 A$.
As a consequence of this alternative construction, the following equality holds in $\ltwo(\R_+)$~:
\begin{equation*}
\UU' W_-^\alpha P_0 \UU'^* =
1 + \12\big(1-\tanh(\pi A)\big) \Big[\frac{2\pi\alpha -\Psi(1)-\ln(2)+\ln(\sqrt{L})+i\pi/2}{2\pi\alpha -\Psi(1)-\ln(2)+\ln(\sqrt{L})-i\pi/2}-1\Big]
\end{equation*}
and it is then clear that this operator belongs to $\EE_{(L,A)}$.
\end{Remark}

By coming back to the expression \eqref{eq_2D} one can compute the image $\Gammaq^\alpha$ of this operator
in the quotient algebra and obtain the following statement:

\begin{Corollary}
For any $\alpha \in \R$, one has
$\Gammaq^\alpha = \big(1, s^\alpha(\cdot),1,1\big)$
and
\begin{equation}\label{equ1}
\Wind(\Gammaq^\alpha) = \Wind(s^\alpha)=\sharp\sigma_{\rm p}(H^\alpha) =1,
\end{equation}
with $s^\alpha(\cdot) =\frac{2\pi\alpha -\Psi(1)-\ln(2)+\ln(\sqrt{\cdot})+i\pi/2}{2\pi\alpha -\Psi(1)-\ln(2)+\ln(\sqrt{\cdot})-i\pi/2}$.
\end{Corollary}

\subsubsection{The dimension $d=3$}

In dimension $3$, there also exists only one real parameter family of self-adjoint operators $H^\alpha$
formally represented as $H_0+\alpha \delta$, and this operator has a single eigenvalue if $\alpha<0$
and no eigenvalue if $\alpha \geq 0$. As for the other two dimensions, the wave operators $W_\pm^\alpha$
for the pair $(H^\alpha,H_0)$ exist, and it has been shown in the reference paper that:

\begin{Lemma}
For any $\alpha \in \R$ the following equality holds
\begin{equation*}
W_-^\alpha = 1+ \12\big(1+\tanh(\pi A)-i \cosh(\pi A)^{-1}\big)
\Big[\frac{4\pi \alpha +i\sqrt{H_0}}{4\pi\alpha -i\sqrt{H_0}}-1\Big] P_0\ .
\end{equation*}
where $P_0$ denotes the projection on the rotation invariant functions of $\ltwo(\R^3)$.
\end{Lemma}

Note that in these formulas, $A$ denotes the generator of dilations in $\ltwo(\R^3)$.
As for the two dimensional case, it is sufficient to restrict our attention to $P_0 \ltwo(\R^3)$
since the subspace of $\ltwo(\R^3)$ which is orthogonal to $P_0\ltwo(\R^3)$ does not play any role for this model
(and it is again the reason why this model is quasi one dimensional).
Let us thus introduce the unitary map $\UU : P_0 \ltwo(\R^3)\to \ltwo(\R_+, r^2\d r)$ defined by $[\UU f](r):= 2\sqrt{\pi} f(r)$
which is well defined since $f\in P_0 \ltwo(\R^3)$ depends only on the radial coordinate.
Since the dilation group as well as the operator $H_0$ leave the subspace $P_0 \ltwo(\R^3)$ of $\ltwo(\R^3)$ invariant, one gets
in $\ltwo(\R_+, r^2\d r)$~:
\begin{equation}\label{eq_3D}
\UU W_-^\alpha P_0\UU^*= 1+ \12\big(1+\tanh(\pi A)-i \cosh(\pi A)^{-1}\big)
\Big[\frac{4\pi \alpha +i\sqrt{H_0}}{4\pi\alpha -i\sqrt{H_0}}-1\Big] \ .
\end{equation}

\begin{Remark}
As in the two dimensional case, the above formula does not take place in any of the representations introduced in
Section \ref{subsec_example} but in a unitarily equivalent one. In this case again, one can come back to the algebra $\EE_{(L,A)}$ by using
the spectral representation of $H_0$. We refer to the $2$-dimensional case for the details.
\end{Remark}

By coming back to the expression \eqref{eq_3D} one can compute the image $\Gammaq^\alpha$ of this operator
in the quotient algebra. If one sets $s^\alpha(\cdot) = \frac{4\pi \alpha +i\sqrt{\cdot}}{4\pi\alpha -i\sqrt{\cdot}} $
one gets:

\begin{Corollary}
For any $\alpha \in \R^*$, one has
$$
\Gammaq^\alpha = \big(
1, s^\alpha(\cdot), -\tanh(\pi \cdot)+i \cosh(\pi \cdot)^{-1},1 \big)
$$
while $\Gammaq^0 = \big(-\tanh(\pi \cdot)+i \cosh(\pi \cdot)^{-1}, -1, -\tanh(\pi \cdot)+i \cosh(\pi \cdot)^{-1},1\big)$.
In addition, for any $\alpha \in \R$ it follows that
\begin{equation}\label{equ2}
\Wind(\Gammaq^\alpha) =\sharp\sigma_{\rm p}(H^\alpha).
\end{equation}
\end{Corollary}

As before, we refer to \cite{KR_PI} for the details of the computations, but stress that some conventions
had been chosen differently.

\subsection{Schr\"odinger operator on $\R$}\label{subsec_Sch1D}

The content of this section is mainly borrowed from \cite{KR1D} but some minor adaptations
with respect to this paper are freely made. We refer to this reference and to the papers mentioned
in it for more information on scattering theory for Schr\"odinger operators on $\R$.

We consider the Hilbert space $\ltwo(\R)$, and the self-adjoint operators $H_0 = -\Delta$ with domain $\H^2(\R)$ and $H=H_0 + V$
with $V$ a multiplication operator by a real function which satisfies the condition
\begin{equation}\label{condV1}
\int_\R (1+|x|)^{\rho}|V(x)|\;\!\d x < \infty,
\end{equation}
for some $\rho\geq 1$.
For such a pair of operators, it is well-known that the conditions required by Assumption \ref{Assum_1}
are satisfied, and thus that the wave operators $W_\pm$ are Fredholm operators and the scattering operator $S$ is unitary.

In order to use the algebraic framework introduced in Section \ref{sec_Kth},
more information on the wave operators are necessary.
First of all, let us recall the following statement which has been proved in \cite{KR1D}.

\begin{Proposition}
Assume that $V$ satisfies \eqref{condV1} with $\rho>5/2$, then the following representation of the wave operator holds:
$$
W_-  = 1 + \12\big(1+\tanh(\pi A)+i\cosh(\pi A)^{-1}(P_{\rm e}-P_{\rm o})\big)
[S-1] + K
$$
with $K$ a compact operator in $\ltwo(\R)$, and $P_{\rm e}, P_{\rm o}$ the projections on the even elements, respectively odd elements, of $\ltwo(\R)$.
\end{Proposition}

Let us now look at this result in  the even$\;\!$/$\;\!$odd representation introduced in Section \ref{subsec_even_odd}.
More precisely, by using the map $\UU:\ltwo(\R)\to \ltwo(\R_+;\C^2)$ introduced in \eqref{eq_e_o}, one gets
\begin{equation}\label{eq_Schro1}
\UU W_- \UU^* =  1_2
 + \12\left(\begin{smallmatrix}1+\tanh(\pi A)+i\cosh(\pi A)^{-1} & 0 \\ 0 & 1+\tanh(\pi A)-i\cosh(\pi A)^{-1} \end{smallmatrix}\right)
\left[ S\left(\begin{smallmatrix}H_{\rm N} & 0 \\ 0 & H_{\rm D} \end{smallmatrix}\right)-1_2\right] +K'
\end{equation}
with $K'\in \K\big(\ltwo(\R_+;\C^2)\big)$.

\begin{Remark}
As in the previous example, the operator $\UU W_- \UU^*$ does not belong directly to one of the algebras
introduced in Section \ref{subsec_example}, but in a unitarily equivalent one which
can be constructed with the spectral representation of $H_0$.
More precisely, we set $\F_0:\ltwo(\R)\to \ltwo(\R_+;\C^2)$ defined by
$$
[\F_0 f](\lambda) = 2^{-1/2}\lambda^{-1/4}\left(\begin{smallmatrix}[\F f](-\sqrt{\lambda}) \\ [\F f](\sqrt{\lambda}) \end{smallmatrix}\right)
\qquad f\in C_{\rm c}(\R), \ \lambda\in \R_+
$$
with $\F$ the unitary Fourier transform in $\ltwo(\R)$.
As usual, one has $[\F_0 H_0 f](\lambda) = \lambda [\F_0 f](\lambda)$ for any $f\in \H^2(\R)$
and a.e.~$\lambda \in \R_+$. Accordingly, one writes $L\otimes 1_2 = \F_0 H_0 \F_0^*$.
Similarly, the equality $\F_0 A \F_0^* = -2A\otimes 1_2$ holds, where the operator $A$
on the l.h.s.~corresponds to the generator of dilation in $\ltwo(\R)$, while the operator $A$ on the r.h.s.~corresponds
to the generator of dilations in $\ltwo(\R_+)$.
Finally, a short computation leads to the equalities $\F_0 P_{\rm e}\F_0^* = \12\left(\begin{smallmatrix}1 & 1 \\ 1 & 1\end{smallmatrix}\right)$
and $\F_0 P_{\rm o}\F_0^* = \12\left(\begin{smallmatrix}1 & -1 \\ -1 & 1\end{smallmatrix}\right)$.
By summing up these information one gets
\begin{equation}\label{eq_R1}
\F_0 W_- \F_0^* = 1_2
+\12 \left(\begin{smallmatrix} 1-\tanh(2\pi A)\  & \ i\cosh(2\pi A)^{-1} \\  i\cosh(2\pi A)^{-1}\  & \ 1-\tanh(2\pi A)\end{smallmatrix}\right)
\big[S(L)-1_2 \big] + \F_0 K\F_0^*\ .
\end{equation}
Based on this formula, it is clear that $\F_0 W_- \F_0^*$ belongs to $\MM_2(\EE_{(L,A)})$, as it should be.
\end{Remark}

Let us however come back to formula \eqref{eq_Schro1} and compute the image $\Gammaq$ of this operator in the
quotient algebra. One clearly gets
\begin{equation}\label{eq_Gamma1D}
\Gammaq  = \Big(
1_2
+ \12\left(\begin{smallmatrix}1+\tanh(\pi \cdot)+i\cosh(\pi \cdot)^{-1} & 0 \\ 0 & 1+\tanh(\pi \cdot)-i\cosh(\pi \cdot)^{-1} \end{smallmatrix}\right)
[S(0)-1_2], S(\cdot), 1_2,1_2 \Big)
\end{equation}
In addition, let us note that under our condition on $V$,
the map $\R_+\ni \lambda \mapsto S(\lambda)\in \MM_2(\C)$
is norm continuous and has a limit at $0$ and converges to $1_2$ at $+\infty$.
Then, by the algebraic formalism, one would automatically obtain that the winding number of the pointwise determinant of the function $\Gammaq$ is equal
to the number of bound states of $H$. However, let us add some more comments on this model, and in particular on the matrix $S(0)$.
In fact, it is well-known that the matrix $S(0)$ depends on the existence or the absence of a so-called half-bound state for $H$ at $0$.
Before explaining this statement, let us recall a result which has been proved in \cite[Prop.~9]{KR1D}, and which is based
only on the explicit expression \eqref{eq_Gamma1D} and its unitarity.

\begin{Lemma}
Either $\det\big(S(0)\big)=-1$ and then $S(0)=\pm \big(\begin{smallmatrix}-1 & 0 \\ 0 & 1\end{smallmatrix}\big)$,
or $\det\big(S(0)\big)=1$ and then $S(0)=\big(\begin{smallmatrix}a & b \\ -\overline{b} & a\end{smallmatrix}\big)$
with $a\in \R$, $b\in \C$ and $|a|^2+|b|^2=1$. Moreover, the contribution to the winding number of the first term
of $\Gammaq$ is equal to $\pm \12$ in the first case, and to $0$ in the second case.
\end{Lemma}

Let us now mention that when $H$ possesses a half-bound state, \ie~a solution of the equation $Hf=0$ with $f$ not in $\ltwo(\R)$
but in a slightly bigger weighted Hilbert space, then $\det\big(S(0)\big)=1$. This case is called
the exceptional case, and thus the first term in $\Gammaq$ does not provide any contribution to the winding number in this case.
On the other hand, when $H$ does not possess such a half-bound state, then
$S(0)=\big(\begin{smallmatrix}-1 & 0 \\ 0 & 1\end{smallmatrix}\big)$.
This case is referred as the generic case, and in this situation the first term in $\Gammaq$ provides a contribution of $\12$
to the winding number.
By taking these information into account, Levinson's theorem can be rewritten for this model as
$$
\Wind(S) =  \left\{\begin{array}{ll}
{\sharp\sigma_{\rm p}(H) -\frac{1}{2}} & \hbox{ in the generic case} \\
{\sharp\sigma_{\rm p}(H)} & \hbox{ in the exceptional case}
\end{array}\right.
$$
Such a result is in accordance with the classical literature on the subject, see \cite{KR1D} and references therein
for the proof of the above statements and for more explanations.
Note finally that one asset of our approach has been to show that the correction $-\frac{1}{2}$ should be located on the other side of the
above equality (with a different sign), and that the rearranged equality is in fact an index theorem.

\subsection{Rank one interaction}\label{subsec_rank1}

In this section, we present another scattering system which has been studied in \cite{RT_rank}.
Our interest in this model comes from the spectrum of $H_0$ which is equal to $\R$.
This fact implies in particular that if $H$ possesses some eigenvalues, then these eigenvalues are
automatically included in the spectrum of $H_0$. In our approach, this fact does not cause any
problem, but some controversies for the original Levinson's theorem with embedded eigenvalues can be found
in the literature, see \cite{Dreyfus}. Note that the following presentation is reduced to the key features only,
all the details can be found in the original paper.

We consider the Hilbert space $\ltwo(\R)$, and let $H_0$ be the operator of multiplication by the variable, \ie~$H_0=X$,
as introduced at the beginning of Section \ref{subsec_example}.
For the perturbation, let $u\in \ltwo(\R)$ and consider the rank one perturbation of $H_0$ defined by
$$
H_u f = H_0 f + \langle u,f\rangle u, \qquad f\in \dom(H_0).
$$
It is well-known that for such a rank one perturbation the wave operators
exist and that the scattering operator is unitary.
Note that for this model, the scattering operator $S \equiv S(X)$ is simply an operator of multiplication
by a function defined on $\R$ and taking values in $\T$.
Let us also stress that for such a general $u$ singular continuous spectrum for $H$ can exist.
In order to ensure the asymptotic completeness, an additional condition on $u$ is necessary.
More precisely, let us introduce this additional assumption:

\begin{Assumption}\label{poupette}
The function $u\in\ltwo(\R)$ is H\"older continuous with exponent $\alpha>1/2$.
\end{Assumption}

It is known that under Assumption \ref{poupette}, the operator $H_u$ has at
most a finite number of eigenvalues of multiplicity one \cite[Sec.~2]{Als80}.
In addition, it is proved in \cite[Lem.~2.2]{RT_rank} that under this assumption
the map
$$
S: \R\ni x\mapsto S(x)\in \T
$$
is continuous and satisfies $S(\pm\infty)=1$.

In order to state the main result about the wave operators for this model, let us use again the
even$\;\!$/$\;\!$odd representation of $\ltwo(\R)$ introduced in Section \ref{subsec_even_odd}.
Let us also recall that we set $m_{\rm e}$, $m_{\rm o}$ for the even part and the odd part of any function $m$ defined on $\R$.

\begin{Theorem}[Theorem 1.2 of \cite{RT_rank}]\label{thm_rank}
Let $u$ satisfy Assumption \ref{poupette}. Then, one has
\begin{equation}\label{nice}
\UU W_- \UU^*=
\left(\begin{smallmatrix}
1 & 0\\
0 & 1
\end{smallmatrix}\right)
+ \12\left(\begin{smallmatrix}
1 & -\tanh(\pi A)+i\cosh(\pi A)^{-1}\\
-\tanh(\pi A)-i\cosh(\pi A)^{-1} & 1
\end{smallmatrix}\right)
\left(\begin{smallmatrix}
S_{\rm e}(L)-1 & S_{\rm o}(L)\\
S_{\rm o}(L) & S_{\rm e}(L)-1
\end{smallmatrix}\right) + K,
\end{equation}
where $K$ is a compact operator in $\ltwo(\R_+;\C^2)$.
\end{Theorem}

Let us immediately mention that a similar formula holds for $W_+$ and that this formula is
exhibited in the reference paper. In addition, it follows from \eqref{nice} that $W_-\in \MM_2\big(\EE_{(L,A)}\big)$,
and that the algebraic framework introduced in Section \ref{sec_Kth} can be applied straightforwardly.
Without difficulty, the formalism leads us directly to the following consequence of Theorem \ref{thm_rank}:

\begin{Corollary}\label{surLev}
Let $u$ satisfy Assumption \ref{poupette}. Then the following
equality holds:
$$
\Wind(S)=\sharp\sigma_{\rm p}(H_u).
$$
\end{Corollary}

Let us stress that another convention had been taken in \cite{RT_rank} for the computation of the winding number,
leading to a different sign in the previous equality.
Note also that such a result was already known for more general perturbations but under
stronger regularity conditions \cite{Buslaev,Drey78}. We stress that the above result does require neither the
differentiability of the scattering matrix nor the differentiability of $u$.
It is also interesting that for this model, only the winding number of the scattering operator
contributes to the left hand side of the equality.

\subsection{Other examples}\label{subsec_Dirac}

In this section, we simply mention two additional models on which some investigations have been performed
in relation with our topological approach of Levinson's theorem.

In reference \cite{IR}, the so-called Friedrichs-Faddeev model has been studied. In
this model, the operator $H_0$ corresponds to the multiplication by the variable but only on an interval $[a,b]$,
and not on $\R$. The perturbation of $H_0$ is defined in terms of an integral operator which satisfies
some H\"older continuity conditions, and some additional conditions on the restriction of the kernel at the values $a$ and $b$ are imposed.
Explicit expressions for the wave operators for this model have been provided in \cite{IR},
but the use of these formulas for deducing a topological Levinson's theorem has
not been performed yet. Note that one of the interests in this model is that the spectrum of $H_0$ is
equal to $[a,b]$, which is different from $\R_+$ or $\R$ which appear in the models developed above.

In reference \cite{PR_Dirac}, the spectral and scattering theory for $1$-dimensional Dirac operators with mass $m>0$
and with a zero-range interaction are fully investigated. In fact, these operators are described by a four real parameters family
of self-adjoint extensions of a symmetric operator.
Explicit expressions for the wave operators and for the scattering
operator are provided. Let us note that these new formulas take place in a representation which links, in a suitable way,
the energies $-\infty$ and $+\infty$, and which emphasizes the role of the thresholds $\pm m$.
Based on these formulas, a topological version of Levinson's theorem is deduced, with the threshold effects at $\pm m$
automatically taken into account. Let us also emphasize that in our investigations on Levinson's theorem,
this model was the first one for which the spectrum of $H_0$ consisted into two disjoint parts, namely $(-\infty,-m]\cup [+m,\infty)$.
It was not clear at the very beginning what could be the suitable algebra for nesting the wave operators and how
the algebraic construction could then be used. The results of these investigations are
thoroughly presented in \cite{PR_Dirac}, and it is expected that the same results hold for less singular perturbations
of $H_0$. Finally, a surprising feature of this model is that the contribution to the winding number from
the scattering matrix is computed from $-m$ to $-\infty$, and then from $+m$ to $+\infty$.
In addition, contributions due to thresholds effects can appear at $-m$ and/or at $+m$.

\section{Schr\"odinger on $\R^3$ and regularized Levinson's theorem}\label{Sec_3D}
\setcounter{equation}{0}

In this section, we illustrate our approach on the example of a Schr\"odinger operator on $\R^3$.
In the first part, we explain with some details how new formulas for the wave operators
can be obtained for this model. In a second part, the algebraic framework is slightly
enlarged in order to deal with a spectrum with infinite multiplicity.
A method of regularization for the computation of the winding number is also presented.

\subsection{New expressions for the wave operators}\label{subsec_Wave3}

In this section, we derive explicit formulas for the wave operators
based on the stationary approach of scattering theory.
Let us immediately stress that the following presentation is deeply inspired from the paper \cite{RT_R3} to which
we refer for the proofs and for more details.
Thus, our aim is to justify the following statement:

\begin{Theorem}\label{Java}
Let $V \in \linf(\R^3)$ be real and satisfy $|V(x)|\le{\rm Const.}\;\!(1+|x|)^{-\rho}$ with $\rho>7$ for
almost every $x\in\R^3$. Then,
the wave operators $W_\pm$ for the pair of operators $(-\Delta+V,-\Delta)$ exist and the following equalities hold in $\B\big(\ltwo(\R^3)\big)$~:
\begin{equation*}
W_-=1+\12\big(1+\tanh(\pi A)-i\cosh(\pi A)^{-1}\big)[S-1]+K
\end{equation*}
and
\begin{equation*}
W_+=1+\12\big(1-\tanh(\pi A)+i\cosh(\pi A)^{-1}\big)[S^*-1]+K',
\end{equation*}
with $A$ is the generator of dilations in $\R^3$, $S$ the scattering operator, and $K,K'\in\K\big(\ltwo(\R^3)\big)$.
\end{Theorem}

In order to prove this statement, let us be more precise about the framework.
We first introduce the Hilbert space $\H:=\ltwo(\R^3)$ and the self-adjoint operator $H_0=-\Delta$
with domain the usual Sobolev space $\H^2 \equiv \H^2(\R^3)$.
We also set $\Hrond:= \ltwo(\R_+;\HS)$ with $\HS:= \ltwo\big(\S^2\big)$, and $\SS(\R^3)$
for the Schwartz space on $\R^3$.
The spectral representation for $H_0$ is constructed as follows:
we define $\F_0:\ltwo(\R^3)\to \ltwo(\R_+;\HS)$ by
\begin{equation*}\label{def_F_0}
\big([\F_0 f](\lambda)\big)(\omega)
=\textstyle\big(\frac\lambda4\big)^{1/4}[\F f](\sqrt\lambda\;\!\omega)
=\textstyle\big(\frac\lambda4\big)^{1/4}
\big[\gamma(\sqrt\lambda)\;\!\F f\big](\omega),
\quad f\in \SS(\R^3),~\lambda\in\R_+,~\omega\in\S^2,
\end{equation*}
with $\gamma(\lambda):\SS(\R^3)\to\HS$ the trace operator given by
$\big[\gamma(\lambda)f\big](\omega):=f(\lambda\;\!\omega)$, and $\F$ the unitary Fourier transform on $\R^3$.
The map $\F_0$ is unitary and satisfies for $f\in \H^2$ and a.e.~$\lambda \in \R_+$
$$
[\F_0 H_0 f](\lambda)= \lambda [\F_0 f](\lambda)\equiv [L\F_0 f](\lambda),
$$
where $L$ denotes the multiplication operator in $\Hrond$ by the variable in $\R_+$.

Let us now introduce the operator $H:=H_0 + V$ with a potential $V\in\linf(\R^3;\R)$ satisfying
for some $\rho>1$ the condition
\begin{equation}\label{condV}
|V(x)|\le{\rm Const.}\;\!\langle x\rangle^{-\rho},\quad\hbox{a.e. }x\in\R^3,
\end{equation}
with $\langle x\rangle:=\sqrt{1+x^2}$. Since $V$ is bounded, $H$ is self-adjoint with
domain $\dom(H)=\H^2$. Also, it is well-known \cite[Thm.~12.1]{Pea88} that the
wave operators $W_\pm$ exist and are asymptotically complete.
In stationary scattering theory one defines the wave operators in terms
of suitable limits of the resolvents of $H_0$ and $H$ on the real axis. We shall
mainly use this second approach, noting that for this model both definitions for the
wave operators do coincide (see \cite[Sec.~5.3]{Yaf}).

Let us thus recall from \cite[Eq.~2.7.5]{Yaf} that for suitable $f,g\in\H$ the
stationary expressions for the wave operators are given by\footnote{In this section, the various scalar products
are indexed by the corresponding Hilbert spaces.}
$$
\big\langle W_\pm f,g\big\rangle_\H
=\int_\R\d\lambda\,\lim_{\varepsilon\searrow0}\frac\varepsilon\pi
\big\langle R_0(\lambda\pm i\varepsilon)f,
R(\lambda\pm i\varepsilon)g\big\rangle_\H\;\!,
$$
where $R_0(z):=(H_0-z)^{-1}$ and $R(z):=(H-z)^{-1}$, $z\in\C\setminus\R$, are the
resolvents of the operators $H_0$ and $H$. We also recall from \cite[Sec.~1.4]{Yaf}
that the limit
$
\lim_{\varepsilon\searrow 0}
\big\langle\delta_\varepsilon(H_0-\lambda)f,g\big\rangle_\H
$
with
$
\delta_\varepsilon(H_0-\lambda)
:=\frac\varepsilon\pi R_0(\lambda\mp i\varepsilon)\;\!R_0(\lambda\pm i\varepsilon)
$
exists for a.e. $\lambda\in\R$ and that
$$
\big\langle f,g\big\rangle_\H
=\int_\R\d\lambda\,\lim_{\varepsilon\searrow0}
\big\langle\delta_\varepsilon(H_0-\lambda)f,g\big\rangle_\H\;\!.
$$
Thus, taking into account the second resolvent equation, one infers that
\begin{equation*}
\big\langle(W_\pm-1)f,g\big\rangle_\H
=-\int_\R\d\lambda\,\lim_{\varepsilon\searrow0}\big\langle
\delta_\varepsilon(H_0-\lambda)f,\big(1+VR_0(\lambda\pm i\varepsilon)\big)^{-1}\;\!
VR_0(\lambda\pm i\varepsilon)g\big\rangle_\H\;\!.
\end{equation*}

We now derive expressions for the wave operators in the spectral representation
of $H_0$; that is, for the operators $\F_0(W_\pm-1)\F_0^*$. So, let
$\varphi,\psi$ be suitable elements of $\Hrond$ (precise conditions will be specified
in Theorem \ref{BigMama} below), then one obtains that
\begin{align*}
&\big\langle\F_0(W_\pm-1)\F_0^*\varphi,\psi\big\rangle_{\!\Hrond}\\
&=-\int_\R\d\lambda\,\lim_{\varepsilon\searrow0}
\big\langle V\big(1+R_0(\lambda\mp i\varepsilon)V\big)^{-1}\F_0^*\;\!
\delta_\varepsilon(L-\lambda)\varphi,
\F_0^*\;\!(L-\lambda\mp i\varepsilon)^{-1}\psi\big\rangle_\H\\
&=-\int_\R\d\lambda\,\lim_{\varepsilon\searrow0}\int_0^\infty\d\mu\,
\big\langle\big\{\F_0V\big(1+R_0(\lambda\mp i\varepsilon)V\big)^{-1}\F_0^*\;\!
\delta_\varepsilon(L-\lambda)\varphi\big\}(\mu),(\mu-\lambda\mp i\varepsilon)^{-1}
\psi(\mu)\big\rangle_\HS.
\end{align*}
Using the short hand notation $T(z):=V\big(1+R_0(z)V\big)^{-1}$, $z\in\C\setminus\R$,
one thus gets the equality
\begin{align}
&\big\langle\;\!\F_0(W_\pm-1)\F_0^*\varphi,\psi\big\rangle_{\!\Hrond}\nonumber\\
&=-\int_\R\d\lambda\,\lim_{\varepsilon\searrow0}\int_0^\infty\d\mu\,
\big\langle\big\{\F_0\;\!T(\lambda\mp i\varepsilon)\;\!\F_0^*\;\!
\delta_\varepsilon(L-\lambda)\varphi\big\}(\mu),
(\mu-\lambda\mp i\varepsilon)^{-1}\psi(\mu)\big\rangle_\HS.\label{start}
\end{align}

This formula will be our starting point for computing new expressions for the wave operators.
The next step is to exchange the integral over $\mu$ and the limit
$\varepsilon\searrow0$. To do it properly, we need a
series of preparatory lemmas. First of all, we recall that for $\lambda>0$ the trace
operator $\gamma(\lambda)$ extends to an element of $\B(\H^s_t,\HS)$ for each $s>1/2$
and $t\in\R$, where $\H^s_t = \H^s_t(\R^3)$ denotes the weighted Sobolev space over $\R^3$
with index $s\in \R$ and with the index $t\in \R$ associated with the weight\footnote{We also use the convention
$\H^s = \H^s_0$ and $\H_t = \H^0_t$.}.
In addition, the map $\R_+\ni\lambda\mapsto\gamma(\lambda)\in\B(\H^s_t,\HS)$
is continuous, see for example \cite[Sec.~3]{Jen81}. As a consequence, the operator
$\F_0(\lambda):\SS(\R^3)\to\HS$ given by $\F_0(\lambda)f:=(\F_0f)(\lambda)$ extends to an
element of $\B(\H^s_t,\HS)$ for each $s\in\R$ and $t>1/2$, and the map
$\R_+\ni\lambda\mapsto\F_0(\lambda)\in\B(\H^s_t,\HS)$ is continuous.

We recall now three technical lemmas which have been proved in \cite{RT_R3} and which
strengthen some standard results.

\begin{Lemma}\label{lem1}
Let $s\ge 0$ and $t>3/2$. Then, the functions
$$
(0,\infty)\ni\lambda\mapsto\lambda^{\pm1/4}\F_0(\lambda)\in\B(\H^s_t,\HS)
$$
are continuous and bounded.
\end{Lemma}

One immediately infers from Lemma \ref{lem1} that the function
$\R_+\ni\lambda\mapsto\|\F_0(\lambda)\|_{\B(\H^s_t,\HS)}\in\R$ is continuous and
bounded for any $s\ge0$ and $t>3/2$. Also, one can strengthen the statement of Lemma
\ref{lem1} in the case of the minus sign\;\!:

\begin{Lemma}\label{lem2}
Let $s >-1$ and $t>3/2$. Then, $\F_0(\lambda)\in\K(\H^s_t,\HS)$ for each
$\lambda\in\R_+$, and the function
$\R_+\ni\lambda\mapsto\lambda^{-1/4}\F_0(\lambda)\in\K(\H^s_t,\HS)$ is continuous,
admits a limit as $\lambda\searrow0$ and vanishes as $\lambda\to\infty$.
\end{Lemma}

From now on, we use the notation $C_{\rm c}(\R_+;\G)$ for the set of compactly
supported and continuous functions from $\R_+$ to some Hilbert space $\G$. With this
notation and what precedes, we note that the multiplication operator
$M:C_{\rm c}(\R_+;\H^s_t)\to\Hrond$ given by
\begin{equation}\label{defdeM}
(M\xi)(\lambda):=\lambda^{-1/4}\F_0(\lambda)\;\!\xi(\lambda),
\quad\xi\in C_{\rm c}(\R_+;\H^s_t),~\lambda\in\R_+,
\end{equation}
extends for $s\ge 0$ and $t>3/2$ to an element of
$\B\big(\ltwo(\R_+;\H^s_t),\Hrond\big)$.

The next step is to deal with the limit $\varepsilon\searrow0$ of the operator
$\delta_\varepsilon(L-\lambda)$ in Equation \eqref{start}. For that purpose, we shall
use the continuous extension of the scalar product
$\langle\;\!\cdot\;\!,\;\!\cdot\;\!\rangle_\H$ to a duality
$\langle\;\!\cdot\;\!,\;\!\cdot\;\!\rangle_{\H^s_t,\H^{-s}_{-t}}$ between $\H^s_t$
and $\H^{-s}_{-t}$.

\begin{Lemma}\label{lemlimite}
Take $s\ge0$, $t>3/2$, $\lambda\in\R_+$ and $\varphi\in C_{\rm c}(\R_+;\HS)$. Then,
we have
$$
\lim_{\varepsilon\searrow 0}\big\|\F_0^*\;\!\delta_\varepsilon(L-\lambda)\varphi
-\F_0(\lambda)^*\varphi(\lambda)\big\|_{\H^{-s}_{-t}}=0.
$$
\end{Lemma}

The next necessary result concerns the limits
$T(\lambda\pm i0):=\lim_{\varepsilon\searrow0}T(\lambda\pm i\varepsilon)$,
$\lambda\in\R_+$. Fortunately, it is already known (see for example
\cite[Lemma 9.1]{JK79}) that if $\rho>1$ in \eqref{condV} then the limit
$\big(1+R_0(\lambda+i0)V\big)^{-1}
:=\lim_{\varepsilon\searrow0}\big(1+R_0(\lambda+i\varepsilon)V\big)^{-1}$
exists in $\B(\H_{-t},\H_{-t})$ for any $t\in(1/2,\rho-1/2)$, and that the map
$\R_+\ni\lambda\mapsto\big(1+R_0(\lambda+i0)V\big)^{-1}\in\B(\H_{-t},\H_{-t})$ is
continuous. Corresponding results for $T(\lambda+i\varepsilon)$ follow immediately.
Note that only the limits from the upper half-plane have been computed in
\cite{JK79}, even though similar results for $T(\lambda-i0)$ could have been derived.
Due to this lack of information in the literature and for the simplicity of the
exposition, we consider from now on only the wave operator $W_-$.

\begin{Proposition}\label{prop_on_sigma}
Take $\rho>5$ in \eqref{condV} and let $t\in(5/2,\rho-5/2)$. Then, the function
$$
\R_+\ni\lambda\mapsto
\lambda^{1/4}\;\!T(\lambda+i0)\F_0(\lambda)^*\in\B(\HS,\H_{\rho-t})
$$
is continuous and bounded, and the multiplication operator
$B:C_{\rm c}\big(\R_+;\HS\big)\to\ltwo(\R_+;\H_{\rho-t})$ given by
\begin{equation}\label{defdeB}
(B\;\!\varphi)(\lambda)
:=\lambda^{1/4}\;\!T(\lambda+i0)\F_0(\lambda)^*\varphi(\lambda)\in\H_{\rho-t},
\quad\varphi\in C_{\rm c}\big(\R_+;\HS\big),~\lambda\in\R_+,
\end{equation}
extends to an element of $\B\big(\Hrond,\ltwo(\R_+;\H_{\rho-t})\big)$.
\end{Proposition}

\begin{Remark}
If one assumes that $H$ has no $0$-energy eigenvalue and/or no $0$-energy resonance,
then one can prove Proposition \ref{prop_on_sigma} under a weaker assumption on the decay of
$V$ at infinity. However, even if the absence of $0$-energy eigenvalue and $0$-energy
resonance is generic, we do not want to make such an implicit assumption in the
sequel. The condition on $V$ is thus imposed adequately.
\end{Remark}

We are ready for stating the main result of this section. Let us simply recall that the dilation group
in $\ltwo(\R_+)$ has been introduced in \eqref{eq_dil_group} and that $A$ denotes its generator.
We also recall that the Hilbert spaces $\ltwo(\R_+;\H^s_t)$ and $\Hrond$ can be
naturally identified with the Hilbert spaces $\ltwo(\R_+)\otimes\H^s_t$ and
$\ltwo(\R_+)\otimes\HS$.

\begin{Theorem}\label{BigMama}
Take $\rho>7$ in \eqref{condV} and let $\,t\in(7/2,\rho-7/2)$. Then, one has in
$\B(\Hrond)$ the equality
\begin{equation*}
\F_0(W_--1)\;\!\F_0^*
=-2\pi i\;\!M\;\!\big\{
\12\big(1-\tanh(2\pi A)-i\cosh(2\pi A)^{-1}\big)
\otimes1_{\H_{\rho-t}}\big\}B,
\end{equation*}
with $M$ and $B$ defined in \eqref{defdeM} and \eqref{defdeB}.
\end{Theorem}

The proof of this statement is rather technical and we shall not reproduce it here. Let us however
mention the key idea. Consider $\varphi\in C_{\rm c}(\R_+;\HS)$ and
$\psi\in C_{\rm c}^\infty(\R_+)\odot C(\S^2)$ (the algebraic tensor product), and set $s:=\rho-t>7/2$. Then, we
have for each $\varepsilon>0$ and $\lambda\in\R_+$ the inclusions
$$
g_\varepsilon(\lambda):=\lambda^{1/4}\;\!T(\lambda+i\varepsilon)\;\!\F_0^*\;\!
\delta_\varepsilon(L-\lambda)\varphi\in\H_s
\qquad\hbox{and}\qquad
f(\lambda):=\lambda^{-1/4}\F_0(\lambda)^*\psi(\lambda)\in\H_{-s}\;\!.
$$
It follows that the expression \eqref{start} is equal to
\begin{align*}
&-\int_\R\d\lambda\,\lim_{\varepsilon\searrow0}\int_0^\infty\d\mu\,\big\langle
T(\lambda+i\varepsilon)\;\!\F_0^*\;\!\delta_\varepsilon(L-\lambda)\varphi,
(\mu-\lambda+i\varepsilon)^{-1}\;\!\F_0(\mu)^*\;\!\psi(\mu)
\big\rangle_{\H_s,\H_{-s}}\\
&=-\int_{\R_+}\d\lambda\,\lim_{\varepsilon\searrow0}\int_0^\infty\d\mu\,
\bigg\langle g_\varepsilon(\lambda),\frac{\lambda^{-1/4}\mu^{1/4}}
{\mu-\lambda+i\varepsilon}\;\!f(\mu)\bigg\rangle_{\H_s,\H_{-s}}.
\end{align*}
Then, once the exchange between the limit $\varepsilon \searrow 0$ and the integral with variable $\mu$
has been fully justified, one obtains that
$$
\big\langle\;\!\F_0(W_\pm-1)\F_0^*\varphi,\psi\big\rangle_{\!\Hrond}
= -\int_{\R_+}\d\lambda \int_0^\infty\d\mu\,
\bigg\langle g_0(\lambda),\frac{\lambda^{-1/4}\mu^{1/4}}
{\mu-\lambda+i0}\;\!f(\mu)\bigg\rangle_{\H_s,\H_{-s}}.
$$
It remains to observe that $g_0(\lambda) = [B\varphi](\lambda)$ and that $f = M^*\psi$,
and to derive a nice function of $A$ from the kernel $\frac{\lambda^{-1/4}\mu^{1/4}}
{\mu-\lambda+i0}$. We refer to the proof of \cite[Thm.~2.6]{RT_R3} for the details.

The next result is a technical lemma which asserts that
a certain commutator is compact. Its proof is mainly based on a result of Cordes which
states that if $f_1,f_2 \in C([-\infty,\infty])$, then the following inclusion holds: $[f_1(X),f_2(D)]\in\K\big(\ltwo(\R)\big)$.
By conjugating this inclusion with the Mellin transform as introduced in Section \ref{subsec_example},
one infers that
$\big[f_1(A),f_3(L)\big]\in\K\big(\ltwo(\R_+)\big)$ with
$f_3:=f_2\circ(-\ln)\in C([0,\infty])$.
Note finally that the following statement does not involve the potential $V$, but only some operators which
are related to $-\Delta$ and to its spectral representation.

\begin{Lemma}\label{Lemma_compact}
Take $s>-1$ and $t>3/2$. Then, the difference
$$
\big\{\big(\tanh(2\pi A)+i\cosh(2\pi A)^{-1}\big)\otimes1_{\HS}\big\}M
-M\big\{\big(\tanh(2\pi A)+i\cosh(2\pi A)^{-1}\big)\otimes 1_{\H^s_t}\big\}
$$
belongs to $\K\big(\ltwo(\R_+;\H_t^s),\Hrond\big)$.
\end{Lemma}

Before providing the proof of Theorem \ref{Java}, let us simply mention that the
following equality holds:
$$
\12\big(1+\tanh(\pi A)-i\cosh(\pi A)^{-1}\big)=\F_0^*\big\{ \12\big(1-\tanh(2\pi A)-i\cosh(2\pi A)^{-1}\big)\otimes 1_{\HS}\big\}\F_0.
$$
with the generator of dilations on $\R^3$ in the l.h.s.~and the generator of dilations on $\R_+$ in the r.h.s.

\begin{proof}[Proof of Theorem \ref{Java}]
Set $s=0$ and $t\in(7/2,\rho-7/2)$. Then, we deduce from Theorem \ref{BigMama},
Lemma \ref{Lemma_compact} and the above paragraph that
\begin{align}
\nonumber W_--1
&=-2\pi i\;\!\F_0^*M\big\{\12\big(1-\tanh(2\pi A)-i\cosh(2\pi A)^{-1}\big)\otimes1_{\H_{\rho-t}}\big\}B\;\!\F_0\\
\label{eq_repspe} &=-2\pi i\;\!\F_0^*\big\{\12\big(1-\tanh(2\pi A)-i\cosh(2\pi A)^{-1}\big)\otimes 1_\HS\big\}MB\;\!\F_0+K\\
\nonumber &=\12\big(1+\tanh(\pi A)-i\cosh(\pi A)^{-1}\big)\;\!\F_0^*(-2\pi iMB)\;\!\F_0+K,
\end{align}
with $K\in\K(\H)$. Comparing $-2\pi iMB$ with the usual expression for the scattering
matrix $S(\lambda)$ (see for example \cite[Eq.~(5.1)]{JK79}), one observes that
$-2\pi iMB=\int_{\R_+}^\oplus\d\lambda\,\big(S(\lambda)-1\big)$. Since $\F_0$ defines
the spectral representation of $H_0$, one obtains that
\begin{equation}\label{form_BPU}
W_--1=\12\big(1+\tanh(\pi A)-i\cosh(\pi A)^{-1}\big)[S-1]+K.
\end{equation}
The formula for $W_+-1$ follows then from \eqref{form_BPU} and the relation
$W_+=W_-\;\!S^*$.
\end{proof}

\subsection{The index theorem}\label{subsec_Lev3}

In order to figure out the algebraic framework necessary for this model, let us first look again at
the wave operator in the spectral representation of $H_0$. More precisely, one deduces from
\eqref{eq_repspe} that the following equality holds in $\Hrond \equiv \ltwo(\R_+)\otimes \HS$~:
$$
\F_0 W_- \F_0^* = 1 + \big\{\12\big(1-\tanh(2\pi A)-i\cosh(2\pi A)^{-1}\big)\otimes 1_\HS\big\} [S(L)-1] + K
$$
with $K\in \K(\Hrond)$. Secondly, let us recall some information on the scattering matrix
which are available in the literature.
Under the assumption on $V$ imposed in Theorem \ref{Java}, the map
$$
\R_+\ni \lambda\mapsto S(\lambda)-1 \in \K_2(\HS)
$$
is continuous, where $\K_2(\HS)$ denotes the set of Hilbert-Schmidt operators on $\HS$, endowed with the Hilbert-Schmidt norm.
A fortiori, this map is continuous in the norm topology of $\K(\HS)$, and in fact this map
belongs to $C\big(\bRp;\K(\HS)\big)$. Indeed, it is well-known that $S(\lambda)$ converges to $1$ as $\lambda \to \infty$,
see for example \cite[Prop.~12.5]{AJS}. For the low energy behavior, see \cite{JK79} where the norm convergence of $S(\lambda)$
for $\lambda \to 0$ is proved (under conditions on $V$ which are satisfied in our Theorem \ref{Java}).
The picture is the following: If $H$ does not possess a $0$-energy resonance, then $S(0)$ is equal to $1$,
but if such a resonance exists, then $S(0)$ is equal to $1-2P_{0}$, where $P_{0}$ denotes the orthogonal
projection on the one-dimensional subspace of spherically symmetric functions in $\HS\equiv \ltwo(\S^{2})$.

By taking these information into account, it is natural to define the unital $C^*$-subalgebra $\EE'$ of $\B\big(\ltwo(\R_+)\otimes \HS\big)$
by $\EE':=\big\{\EE_{(L,A)}\otimes \K(\HS)\big\}^+$ and to consider the short exact sequence
$$
0 \to \K\big(\ltwo(\R_+)\big)\otimes \K(\HS) \hookrightarrow \EE' \stackrel{q\;\;\!}{\to} \big\{C(\square)\otimes \K(\HS)\big\}^+ \to 0.
$$
However, we prefer to look at a unitarily equivalent representation of this algebra in the original Hilbert space
$\H$, and set $\EE:= \F_0^* \big\{\EE_{(L,A)}\otimes \K(\HS)\big\}^+ \F_0 \subset \B(\H)$. The corresponding short exact sequence reads
$$
0 \to \K(\H) \hookrightarrow \EE \stackrel{q\;\;\!}{\to} \big\{C(\square)\otimes \K(\HS)\big\}^+ \to 0,
$$
and this framework is the suitable one for the next statement:

\begin{Corollary}
Let $V \in \linf(\R^3)$ be real and satisfy $|V(x)|\le{\rm Const.}\;\!(1+|x|)^{-\rho}$ with $\rho>7$ for
almost every $x\in\R^3$. Then $W_-$ belongs to $\EE$ and its image
$\Gammaq := q(W_-)$ in $\big\{C(\square)\otimes \K(\HS)\big\}^+$ is given by
$$
\Gammaq = \Big(1+\12\big(1+\tanh(\pi \cdot)-i\cosh(\pi \cdot)^{-1}\big)[S(0)-1], S(\cdot),1,1\Big).
$$
In addition, the equality
\begin{equation}\label{eq_top3}
\ind[\Gammaq]_1 = -[E_{\rm p}(H)]_0
\end{equation}
holds, with
$[\Gammaq]_1\in K_1\big(\big\{C(\square)\otimes \K(\HS)\big\}^+ \big)$ and $[E_{\rm p}(H)]_0\in K_0\big(\K(\H)\big)$.
\end{Corollary}

Let us mention again that if $H$ has no $0$-energy resonance, then $S(0)-1$ is equal to $0$,
and thus the first term $\Gamma_1$ in the quadruple $\Gammaq$ is equal to $1$.
However, if such a resonance exists, then $\Gamma_1$ is not equal to $1$ but to
$$
1-\big(1+\tanh(\pi \cdot)-i\cosh(\pi \cdot)^{-1}\big)P_0 = P_0^\bot + \big(-\tanh(\pi \cdot)+i\cosh(\pi \cdot)^{-1}\big)P_0.
$$
This term will allow us to explain the correction which often appears in the literature for
$3$-dimensional Schr\"odinger operators in the presence of a resonance at $0$.
However, for that purpose we first need a concrete computable version of our topological Levinson's theorem, or in other words
a way to deduce an equality between numbers from the equality \eqref{eq_top3}.

The good point in the previous construction is that the $K_0$-group of $\K(\H)$ and the
$K_1$-group of $\big\{C(\square)\otimes \K(\HS)\big\}^+$ are both isomorphic to $\Z$.
On the other hand, since $S(\lambda)-1$ takes values in $\K(\HS)$ and not in $\MM_n(\C)$ for some fixed $n$,
it is not possible to simply apply the map $\Wind$ without a regularization process. In the next section,
we shall explain how such a regularization can be constructed, but let us already present the final
result for this model.

In the following statement, we shall use the fact that for the class of perturbations
we are considering the map $\R_+\ni \lambda \mapsto S(\lambda)-1 \in\K(\HS)$ is continuous in the Hilbert-Schmidt norm.
Furthermore, it is known that this map is even continuously differentiable in the norm topology.
In particular, the on-shell time delay operator
$i\;\!S(\lambda)^*\;\!S'(\lambda)$ is well defined for each $\lambda
\in \R_+$, see \cite{Jen81,Jen83} for details.
If we set $\K_1(\HS)$ for the trace class operators in $\K(\HS)$,
and denote the corresponding trace by $\tr$, then the following statement holds:

\begin{Theorem}\label{Levnous}
Let $V \in \linf(\R^3)$ be real and satisfy $|V(x)|\le{\rm Const.}\;\!(1+|x|)^{-\rho}$ with $\rho>7$ for
almost every $x\in\R^3$. Then for any $p\geq 2$ one has
\begin{equation*}
\frac{1}{2\pi}\Big\{
\int_{-\infty}^{\infty} \tr \big[i\big(1-\Gamma_1(\xi)\big)^p
\Gamma_1(\xi)^*\,\Gamma_1'(\xi)\big] \d \xi
+ \int_0^\infty  \tr \big[i\big(1-S(\lambda)\big)^p S(\lambda)^*\,S'(\lambda)\big] \d
\lambda\Big\}  = \sharp\sigma_{\rm p}(H).
\end{equation*}
In addition, if the map $\lambda \mapsto S(\lambda)-1$ is continuously differentiable in
the Hilbert-Schmidt norm, then the above equality holds also for any $p\geq 1$.
\end{Theorem}

The proof of this statement is a corollary of the construction presented in the next section.
For completeness, let us mention that in the absence of a resonance at $0$ for $H$, in
which case $\Gamma_1=1$, only the second term containing $S(\cdot)$ contributes to the l.h.s.
On the other hand, in the presence of a resonance at $0$ the real part of the integral of the term $\Gamma_1$ yields
$$
\Re\Big\{ \frac{1}{2\pi} \int_{-\infty}^{\infty} \tr
\big[i\big(1-\Gamma_1(\xi)\big)^p \Gamma_1(\xi)^*\;\!\Gamma_1'(\xi)\big] \d \xi\Big\} =  -\12
$$
which accounts for the correction usually found in
Levinson's theorem. Note that only the real part of this expression is of interest since its imaginary
part will cancel with the corresponding imaginary part of term involving $S(\cdot)$.

\subsection{A regularization process}\label{subsec_regul}

In this section and in the corresponding part of the Appendix, we recall and adapt some of the results and proofs from \cite{KR_R3}
on a regularization process.
More precisely, for an arbitrary Hilbert space $\HS$, we consider a unitary element
$\Gamma\in C\big(\S;\K(\HS)\big)^+$ of the form $\Gamma(t) - 1 \in \K(\HS)$ for any $t\in \S$.
Clearly, there is a certain issue about the possibility of computing a kind of winding
number on this element, as the determinant of $\Gamma(t)$ is not always defined.
Nevertheless, at the level of $K$-theory, it is {\it a priori} possible to define $\Wind$ on $[\Gamma]_1$
simply by evaluating it on a representative on which the pointwise determinant is well defined
and depends continuously on $t$. For our purpose this approach is
not sufficient, however, as it is not clear how to construct
for a given $\Gamma$ such a representative. We will therefore have to make
recourse to a regularization of the determinant.

Let us now explain this regularization in the case that
$\Gamma(t)-1$ lies in the $p$-th Schatten ideal $\K_p(\HS)$ for some integer $p$, that is, $|\Gamma(t)-1|^p$ belongs to $\K_1(\HS)$.
We also denote by $\{e^{i\theta_j(t)}\}_j$ the set of eigenvalues of $\Gamma(t)$.
Then the regularized Fredholm determinant $\det_p$, defined by \cite[Eq.~(XI.4.5)]{GGK}
$$
{\det}_p\big(\Gamma(t)\big)
=\prod_j e^{i\theta_j(t)}\exp\left(\sum_{k=1}^{p-1}\frac{(-1)^k}{k}
(e^{i\theta_j(t)} -1)^k\right)
$$
is finite and non-zero.
Thus, if in addition we suppose that $t\mapsto \Gamma(t)-1$ is continuous in the
$p$-th Schatten norm, then the map $t\to {\det}_p\big(\Gamma(t)\big)$ is continuous and hence
the winding number of the map $\S\ni t\mapsto {\det}_p\big(\Gamma(t)\big)\in\C^*$ can be defined.
However, in order to get an analytic formula for this winding number, stronger conditions are necessary,
as explicitly required in the following statements:

\begin{Lemma}\label{lem_reg_1}
Let $I\subset \S$ be an open arc of the unit circle, and assume that the map $I \ni t\mapsto \Gamma(t)-1 \in \K_p(\HS)$ is continuous in norm of $\K_p(\HS)$
and is continuously differentiable in norm of $\K(\HS)$. Then the map $I \ni t \mapsto \det_{p+1}\big(\Gamma(t)\big)\in \C$
is continuously differentiable and the following equality holds for any $t\in I$:
\begin{equation}\label{aobtenir}
\Big(\ln\det_{p+1}\big(\Gamma(\cdot)\big)\Big)'(t)
= \tr\big[\big(1-\Gamma(t)\big)^{p}\Gamma(t)^*
\Gamma'(t)\big].
\end{equation}
Furthermore, if the map $I \ni t\mapsto \Gamma(t)-1 \in \K_p(\HS)$ is continuously differentiable in norm of $\K_p(\HS)$,
then the statement already holds for $p$ instead of $p+1$.
\end{Lemma}

The proof of this statement is provided in Section \ref{subsec_app_regul}.
Based on \eqref{aobtenir}, it is natural to define
$$
\Wind(\Gamma):=\frac{1}{2\pi}\int_\S \tr\big[i\big(1-\Gamma(t)\big)^{p}\Gamma(t)^*
\Gamma'(t)\big] \d t
$$
whenever the integrant is well-defined and integrable.
However, such a definition is meaningful only if the resulting number does not depend on $p$, for sufficiently large $p$.
This is indeed the case, as shown in the next statement:

\begin{Lemma}\label{lem_reg_2}
Assume that the map $\S \ni t\mapsto \Gamma(t)-1 \in \K_p(\HS)$ is continuous in norm of $\K_p(\HS)$,
and that this map is continuously differentiable in norm of $\K(\HS)$, except on a finite subset of $\S$ (which can be void).
If the map $t\mapsto \tr\big[i\big(1-\Gamma(t)\big)^{p}\Gamma(t)^* \Gamma'(t)$ is integrable for some integer $p$, then for any integer $q> p$ one has:
\begin{equation*}
\frac{1}{2\pi} \int_\S \tr\big[i\big(1-\Gamma(t)\big)^{q}\Gamma(t)^*
\Gamma'(t)\big]\d t  = \frac{1}{2\pi} \int_\S \tr\big[i\big(1-\Gamma(t)\big)^{p}\Gamma(t)^*
\Gamma'(t)\big]\d t\ .
\end{equation*}
\end{Lemma}

The proof of this statement is again provided in Section \ref{subsec_app_regul}.
Clearly, Theorem \ref{Levnous} is an application of the previous lemma with $p=2$.

Before ending this section, let us add one illustrative example. In it, the problem does
not come from the computation of a determinant, but from an integrability condition. More precisely,
for any $a,b>0$ we consider $\varphi_{a,b}: [0,1]\to \R$ defined by
$$
\varphi_{a,b}(x)=x^a \sin\(\pi x^{-b}/2\), \qquad x\in [0,1].
$$
Let us also set $\Gamma_{a,b}:[0,2\pi]\to \T$ by
$$
\Gamma_{a,b}(x):=\e^{- 2\pi i \varphi_{a,b}(x/2\pi)}.
$$
Clearly, $\Gamma_{a,b}$ is continuous on $[0,2\pi]$ with $\Gamma_{a,b}(0)=\Gamma_{a,b}(2\pi)$, and thus can be considered as an element of $C(\S)$.
In addition, $\Gamma_{a,b}$ is unitary, and thus there exists an element $[\Gamma_{a,b}]_1$ in $K_1\big(C(\S)\big)$.
One easily observes that the equality $[\e^{-2\pi i\;\! {\rm id}}]_1 = [\Gamma_{a,b}]_1$ holds, meaning that the equivalence class
$[\Gamma_{a,b}]$ contains the simpler function $x\mapsto \e^{-2\pi i x}$.
However, the same equivalence class also contains some functions which are continuous but not differentiable
at a finite number of points, or even wilder continuous functions.

Now, the computation of the winding number of any of these functions
can be performed by a topological argument, and one obtains
$$
\Wind_t(\Gamma_{a,b}):=\varphi_{a,b}(1)-\varphi_{a,b}(0)=1.
$$
Here, we have added an index $t$ for emphasizing that this number is computed \emph{topologically}.
On the other hand, if one is interested in an explicit analytical formula for this winding number, one immediately faces some
troubles. Namely, let us first observe that for $x\neq 0$
$$
\varphi_{a,b}'(x)=ax^{a-1}\sin\(\pi x^{-b}/2\) - \frac{b\pi}{2}x^{a-b-1}\cos\(\pi x^{-b}/2\).
$$
Clearly, the first term is integrable for any $a,b>0$ while the second one is integrable only if $a>b$.
Thus, in such a case it is natural to set
$$
\Wind_a(\Gamma_{a,b}):=\frac{1}{2\pi}\int_0^{2\pi} i\Gamma_{a,b}(x)^* \Gamma_{a,b}'(x)\d x
=\frac{1}{2\pi}\int_0^{2\pi}\varphi_{a,b}'(x/2\pi)\d x
=\int_0^{1}\varphi_{a,b}'(x)\d x
$$
and this formula is well defined, even if one looks at each term of $\varphi_{a,b}'$ separately.
Note that in this formula, the index $a$ stands for \emph{analytic}.
Finally, in the case $b\geq a>0$ the previous formula is not well defined if one looks at both terms separately,
but one can always find $p\in \N$ such that $(p+1)a>b$. Then one can set
\begin{align*}
\Wind_r(\Gamma_{a,b})& :=\frac{1}{2\pi}\int_0^{2\pi} \(1-\Gamma_{a,b}(x)\)^{p} i\Gamma_{a,b}(x)^*\Gamma_{a,b}'(x)\d x \\
& = \frac{1}{2\pi}\int_0^{2\pi} \(1-\Gamma_{a,b}(x)\)^{p} \varphi_{a,b}'(x/2\pi)\d x \\
& = \int_0^{1} \(1-\Gamma_{a,b}(2\pi x)\)^{p} \varphi_{a,b}'(x)\d x ,
\end{align*}
and this formula is well defined, even if one looks at each term of $\varphi_{a,b}'$ separately.
Note that here the index $r$ stands for \emph{regularized}.
Clearly, the value which can be obtained from these formulas is always equal to $1$.

\section{Schr\"odinger operators on $\R^2$}\label{Sec_2D}
\setcounter{equation}{0}

In this section, we simply provide an explicit formula for the wave operators
in the context of Schr\"o\-din\-ger operators on $\R^2$. The statement is very similar
to the one presented in Section \ref{subsec_Wave3}, and its proof is based on the same scheme.
We refer to \cite{RT_R2} for a more detailed presentation of the result and for its
proof. Let us however mention that some technicalities have not been considered in $\R^2$,
and therefore our main result applies only in the absence of $0$-energy bound state
or $0$-energy resonance.

Let us be more precise about the framework. In the Hilbert space
$\ltwo(\R^2)$ we consider the Schr\"o\-din\-ger operator $H_0:=-\Delta$ and the perturbed operator $H:=-\Delta+V$,
with a potential $V\in\linf(\R^2;\R)$ decaying fast enough at infinity. In such a situation, it is well-known that
the wave operators $W_\pm$ for the pair $(H,H_0)$
exist and are asymptotically complete. As a consequence, the scattering operator $S:=W_+^* W_-$ is a unitary operator.

\begin{Theorem}\label{Java2}
Suppose that $V\in \linf(\R^2)$ is real and satisfies $|V(x)|\le{\rm Const.}\;\!(1+|x|)^{-\rho}$ with
$\rho>11$ for almost every $x\in\R^2$, and assume that $H$ has neither eigenvalues
nor resonances at $0$-energy. Then, one has in $\B\big(\ltwo(\R^2)\big)$ the equalities
\begin{equation*}
W_-=1+\12\big(1+\tanh(\pi A/2)\big)[S-1]+K
\qquad\hbox{and}\qquad
W_+=1+\12\big(1-\tanh(\pi A/2)\big)[S^*-1]+K',
\end{equation*}
with $A$ the generator of dilations in $\ltwo(\R^2)$ and $K,K'\in\K\big(\ltwo(\R^2)\big)$.
\end{Theorem}

We stress that the absence of eigenvalues or resonances at $0$-energy is generic.
Their presence leads to slightly more complicated expressions and this has not been considered
in \cite{RT_R2}. On the other hand, we note that no spherical symmetry is imposed on $V$.
Note also that in the mentioned reference, an additional formula for $W_\pm$ which does not involve any compact remainder
(as in Theorem \ref{BigMama}) has been exhibited.
For this model, we do not deduce any topological Levinson's theorem, since this has
already been performed for the $3$-dimensional case, and since the exceptional case
has not yet been fully investigated. Let us however mention that similar results
already exist in the literature, but that the approaches are completely different.
We refer to \cite{BGD88,EG12_0,EG12_1,JY02,Sch05,Yaj99} for more information on this model and for related results.

\section{A.-B. model and higher degree Levinson's theorem}\label{Sec_AB}
\setcounter{equation}{0}

In this section, we first introduce the Aharonov-Bohm model, and discuss some of the results obtained in \cite{PR}.
In order to extend the discussion about index theorems to index theorems for families, we then provide
some information on cyclic cohomology and explain how it can be applied to this model. This material mainly is borrowed from
\cite{KPR} to which we refer for more information.

\subsection{The Aharonov-Bohm model}\label{sec_AB_model}

Let us denote by $\H$ the Hilbert space $\ltwo(\R^2)$.
For any $\alpha \in (0,1)$, we set $A_\alpha: \R^2\setminus\{0\} \to \R^2$ by
\begin{equation*}
A_\alpha(x,y)= -\alpha \left(\frac{-y}{x^2+y^2},
\frac{x}{x^2+y^2}\right),
\end{equation*}
which formally corresponds to the magnetic field $B=\alpha\delta$ ($\delta$ is the Dirac
delta function), and consider the operator
\begin{equation*}
H_\alpha:=(-i\nabla -A_\alpha)^2,
\qquad \dom(H_\alpha)=C_{\rm c}^\infty\big(\R^2\setminus\{0\}\big)\ .
\end{equation*}
Here $C_{\rm c}^\infty(\Xi)$ denotes the set of smooth functions on $\Xi$ with compact support.
The closure of this operator in $\H$, which is denoted by the same symbol,
is symmetric and has deficiency indices $(2,2)$.

We briefly recall the parametrization of the self-adjoint extensions of $H_\alpha$ from \cite{PR}.
Some elements of the domain of the adjoint operator $H_\alpha^*$ admit singularities
at the origin. For dealing with them, one defines four linear functionals $\Phi_0$, $\Phi_{-1}$, $\Psi_0$, $\Psi_{-1}$
on $\dom(H_\alpha^*)$ such that for $\f\in\dom(H_\alpha^*)$ one has, with $\theta \in [0,2\pi)$ and $r \to 0_+$,
\[
2\pi \f(r\cos\theta,r\sin\theta)= \Phi_0(\f)r^{-\alpha}+\Psi_0(\f) r^\alpha
+e^{-i\theta} \Big(
\Phi_{-1}(\f)r^{\alpha-1}+\Psi_{-1}(\f) r^{1-\alpha}
\Big) +O(r).
\]
The family of all self-adjoint extensions of the operator $H_\alpha$ is then indexed by two matrices
$C,D \in \MM_2(\C)$ which satisfy the following conditions:
\begin{equation}\label{eq-mcd}
\text{(i) $CD^*$ is self-adjoint,\qquad  (ii) $\det(CC^* + DD^*)\neq 0$,}
\end{equation}
and the corresponding extensions $H^{\CD}_\alpha$ are the restrictions of $H_\alpha^*$
onto the functions $\f$ satisfying the boundary conditions
\[
C \left(\begin{matrix}
\Phi_0(\f)\\ \Phi_{-1}(\f)
\end{matrix}\right)
=2D \left(\begin{matrix}
\alpha \Psi_0(\f)\\  (1-\alpha)\Psi_{-1}(\f)
\end{matrix}\right).
\]
For simplicity, we call \emph{admissible} a pair of
matrices $(C,D)$  satisfying the conditions mentioned in \eqref{eq-mcd}.

\begin{Remark}\label{1to1}
The parametrization of the self-adjoint extensions of $H_\alpha$ with all admissible
pairs $(C,D)$ is very convenient but non-unique.
At a certain point, it will be useful to have a one-to-one parametrization of all
self-adjoint extensions.
So, let us consider $\U_2(\C)$ (the group of unitary $2\times 2$ matrices) and set
\begin{equation*}
C(U) := {\textstyle \frac{1}{2}}(1-U) \quad \hbox{ and }
\quad D(U) = {\textstyle \frac{i}{2}}(1+U).
\end{equation*}
It is easy to check that $C(U)$ and $D(U)$ satisfy both conditions \eqref{eq-mcd}.
In addition, two different elements $U,U'$ of $\U_2(\C)$ lead to two different self-adjoint
operators $H_\alpha^{C(U)\;\!D(U)}$ and $H_\alpha^{C(U')\;\!D(U')}$, {\it cf.}~\cite{Ha}.
Thus, without ambiguity we can write $H_\alpha^U$ for the operator $H_\alpha^{C(U)\;\!D(U)}$.
Moreover, the set $\{H_\alpha^U\mid U \in \U_2(\C)\}$ describes all
self-adjoint extensions of $H_\alpha$.
Let us also mention that the normalization of the above maps has been chosen such that
$H_\alpha^{-1}\equiv H_\alpha^{10}= H_\alpha^{\AB}$ which corresponds to the standard Aharonov-Bohm operator
studied in \cite{AB,Rui}.
\end{Remark}

For the spectral theory, let us mention that the essential spectrum of $H^{\CD}_\alpha$
is absolutely continuous and covers the positive half-line $[0,+\infty)$.
On the other hand, the discrete spectrum consists in at most two negative eigenvalues.
More precisely, the number of negative eigenvalues of $H^{\CD}_\alpha$ coincides with the number
of negative eigenvalues of the matrix $CD^*$.

\subsubsection{Wave and scattering operators}\label{sec_AB_wave}

One of the main results of \cite{PR} is an explicit description of the wave operators.
We shall recall this result below, but we first need
to introduce the decomposition of the Hilbert space $\H$ with respect to a special basis.
For any $m \in \Z$, let $\phi_m$ be the complex function defined by
$[0,2\pi)\ni \theta \mapsto \phi_m(\theta):= \frac{e^{im\theta}}{\sqrt{2\pi}}$.
One has then the canonical isomorphism
\begin{equation}\label{decomposition}
\H \cong \bigoplus_{m \in \Z} \H_r \otimes [\phi_m] \ ,
\end{equation}
where $\H_r:=\ltwo(\R_+, r\;\!\d r)$ and $[\phi_m]$ denotes the one dimensional space spanned by $\phi_m$.
For shortness, we write $\H_m$ for $\H_r \otimes [\phi_m]$, and often consider it as a subspace of $\H$.
Let us still set
\begin{equation}\label{eq_Hint}
\H_\2:=\H_0\oplus\H_{-1}
\end{equation}
which is clearly isomorphic to $\H_r\otimes \C^2$.

Let us also recall that the unitary dilation group
$\{U_t\}_{t \in \R}$ in $\H$ is defined on any $\f \in \H$ and $x \in \R^2$ by
$[U_t \f](x) = e^t \f(e^t x)$. Its self-adjoint generator is still denoted by $A$.
It is easily observed that this group as well as its generator leave each subspace $\H_m$ invariant.

Let us now consider the wave operators
\begin{equation*}
W^{\CD}_-\equiv W_-(H_\alpha^\CD,H_0):=s-\lim_{t\to - \infty}e^{itH^{\CD}_\alpha }\;\!e^{-itH_0 }\ .
\end{equation*}
where $H_0:=-\Delta$ denotes the Laplace operator on $\R^2$.
It is well-known that for any admissible pair $(C,D)$ the operator $W_-^{\CD}$ is reduced by
the decomposition $\H=\H_\2 \oplus \H_\2^\bot$ and that
$W_-^{\CD}|_{\H_\2^\bot} = W_-^{\AB}|_{\H_\2^\bot}$.
The restriction to $\H_\2^\bot$ is further reduced by the decomposition \eqref{decomposition}
and it is proved in \cite[Prop.~11]{PR} that the channel wave operators satisfy for each $m \in \Z$,
\begin{equation*}
W_{-,m}^{\AB} = \varphi_m^-(A)\ ,
\end{equation*}
with $\varphi_m^-$ explicitly given for $x \in \R$ by
\begin{equation*}
\varphi^-_m(x):=e^{i\delta_m^\alpha}\;\!
\frac{\Gamma\big(\frac{1}{2}(|m|+1+ix)\big)}{\Gamma\big(\frac{1}{2}(|m|+1-ix)\big)}
\;\!
\frac{\Gamma\big(\frac{1}{2}(|m+\alpha|+1-ix)\big)}{\Gamma\big(\frac{1}{2}(|m+\alpha|+1+ix)\big)}
\end{equation*}
and
\begin{equation*}
\delta_m^\alpha = \hbox{$\frac{1}{2}$}\pi\big(|m|-|m+\alpha|\big)
=\left\{\begin{array}{rl}
-\hbox{$\frac{1}{2}$}\pi\alpha & \hbox{if }\ m\geq 0 \\
\hbox{$\frac{1}{2}$}\pi\alpha & \hbox{if }\ m< 0
\end{array}\right.\ .
\end{equation*}
Note that here, $\Gamma$ corresponds to the usual Gamma function.
It is also proved in \cite[Thm.~12]{PR} that
\begin{equation*}
W_-^{\CD}|_{\H_\2} = \left(
\begin{matrix}\varphi^-_0(A) & 0 \\ 0 & \varphi^-_{-1}(A) \end{matrix}\right) +
\left(
\begin{matrix}\tilde{\varphi}_0(A) & 0 \\ 0 & \tilde{\varphi}_{-1}(A) \end{matrix}\right)\;\!
\widetilde{S}^{\CD}_\alpha\big(\sqrt{H_0}\big)
\end{equation*}
with $\tilde{\varphi}_m(x)$ given for $m \in \{0,-1\}$ by
\begin{eqnarray*}
\frac{1}{2\pi}\;\!e^{-i\pi|m|/2} \;\!e^{\pi x/2}\;\!
\frac{\Gamma\big(\frac{1}{2}(|m|+1+ix)\big)}{\Gamma
\big(\frac{1}{2}(|m|+1-ix)\big)}  \Gamma\big(\ud(1+|m+\alpha|-ix)\big)
\;\!\Gamma\big(\ud(1-|m+\alpha|-ix)\big)\ ,
\end{eqnarray*}
and with the function $\widetilde{S}^{\CD}_\alpha(\cdot)$ given
for $\lambda \in \R_+$ by
\begin{eqnarray*}
\widetilde{S}_\alpha^{\CD}(\lambda)
&:=& 2i\sin(\pi\alpha)
\left(\begin{matrix}
\frac{\Gamma(1-\alpha)\;\!e^{-i\pi\alpha/2}}{2^\alpha}\;\!\lambda^{\alpha} & 0 \\
0 & \frac{ \Gamma(\alpha)\;\!e^{-i\pi(1-\alpha)/2}}{2^{1-\alpha}}\;\!\lambda^{(1-\alpha)}
\end{matrix}\right)  \\
&&\cdot \left[ D\,
\left(\begin{matrix}
\frac{\Gamma(1-\alpha)^2 \;\!e^{ -i\pi\alpha}}{4^\alpha}\;\!\lambda^{2\alpha} & 0 \\
0& \frac{\Gamma(\alpha)^2\;\! e^{ -i\pi(1-\alpha)}}{4^{1-\alpha}}\;\!\lambda^{2(1-\alpha)}
\end{matrix}\right)
+\frac{\pi}{2\sin(\pi\alpha)}C\right]^{-1}\!\!\! D \\
&& \cdot
\left(\begin{matrix}
\frac{ \Gamma(1-\alpha)\;\!e^{-i\pi\alpha/2}}{2^\alpha}\;\!\lambda^{\alpha} & 0 \\
0 & -\frac{ \Gamma(\alpha)\;\!e^{-i\pi(1-\alpha)/2}}{2^{1-\alpha}}\;\!\lambda^{(1-\alpha)}
\end{matrix}\right)\ .
\end{eqnarray*}

Clearly, the functions $\varphi^-_m$ and $\tilde{\varphi}_m$ are continuous on $\R$. Furthermore,
these functions admit limits at $\pm \infty$: $\varphi^-_m(-\infty)=1$, $\varphi^-_m(+\infty)=e^{2i\delta^\alpha_m}$,
$\tilde{\varphi}_m(-\infty)=0$ and $\tilde{\varphi}_m(+\infty)=1$.
On the other hand, the relation between the usual scattering operator $S^{\CD}_\alpha:=\big(W^{\CD}_+\big)^* W^{\CD}_-$
and the function $\widetilde{S}^{\CD}_\alpha(\cdot)$ is provided by the formulas
\begin{equation*}
S_\alpha^{\CD}|_{\H_\2}=
S_\alpha^{\CD}(\sqrt{H_0})
\quad\hbox{with}\quad
S_\alpha^{\CD}(\lambda):=\left(\begin{matrix}
e^{-i\pi\alpha} & 0 \\
0 & e^{i\pi\alpha}
\end{matrix}\right)
+ \widetilde{S}_\alpha^{\CD}(\lambda)\ .
\end{equation*}

Let us now state a result which has been formulated in a more precise form in \cite[Prop.~14]{PR}.

\begin{Proposition}\label{propSurS}
The map
\begin{equation*}
\R_+\ni \lambda \mapsto S^{\CD}_\alpha(\lambda) \in \U_2(\C)
\end{equation*}
is continuous and has explicit asymptotic values for $\lambda=0$ and $\lambda = +\infty$
which depend on $C,D$ and $\alpha$.
\end{Proposition}

The asymptotic values $S^{\CD}_\alpha(0)$ and $S^{\CD}_\alpha(+\infty)$
are explicitly provided in the statement of \cite[Prop.~14]{PR}, but numerous cases have to be considered.
For simplicity, we do not provide these details here.
By summarizing the information obtained so far, one infers that:

\begin{Theorem}
For any admissible pair $(C,D)$ the following equality holds:
\begin{equation}\label{yoyo}
W_-^{\CD}|_{\H_\2} = \left(
\begin{matrix}\varphi^-_0(A) & 0 \\ 0 & \varphi^-_{-1}(A) \end{matrix}\right) +
\left(
\begin{matrix}\tilde{\varphi}_0(A) & 0 \\ 0 & \tilde{\varphi}_{-1}(A) \end{matrix}\right)
\left[S^{\CD}_\alpha\big(\sqrt{H_0}\big)-\left(\begin{matrix}
e^{-i\pi\alpha} & 0 \\
0 & e^{i\pi\alpha}
\end{matrix}\right)\right],
\end{equation}
with $\varphi^-_0, \varphi^-_1,\tilde\varphi_0,\tilde\varphi_1\in C([-\infty,\infty])$
and with $S^{\CD}_\alpha\in C([0,+\infty])$.
\end{Theorem}

Based on this result and on the content of Section \ref{sec_Kth}, one could easily deduce an index type theorem.
However, we prefer to come back to an \emph{ad hoc} approach, which looks more like the approach followed
for the baby model. Its interest is that individual contributions to the winding number can be computed,
and the importance of each of them is thus emphasized. A more conceptual (and shorter) proof
will be provided in Section \ref{sec_higher}.

\subsection{Levinson's theorem, the pedestrian approach}\label{sec_ped}

Let us start by considering again the expression \eqref{yoyo} for the operator $W_-^{\CD}|_{\H_\2}$.
Since the matrix-valued functions defining this operator have limits at $-\infty$ and $+\infty$, respectively at $0$ and $+\infty$,
one can define the quadruple $(\Gamma_1,\Gamma_2,\Gamma_3,\Gamma_4)$, with
$\Gamma_j$ given for $x \in \R$ and $\lambda\in \R_+$ by
\begin{eqnarray}
\nonumber \Gamma_1(x)\equiv \Gamma_1(C,D,\alpha,x)&:=&
\left(\begin{matrix}\varphi^-_0(x) & 0 \\ 0 & \varphi^-_{-1}(x) \end{matrix}\right) +
\left(\begin{matrix}\tilde{\varphi}_0(x) & 0 \\ 0 & \tilde{\varphi}_{-1}(x) \end{matrix}\right)
\widetilde{S}^{\CD}_\alpha(0)\ ,\\
\label{Gma2}\Gamma_2(\lambda)\equiv \Gamma_2(C,D,\alpha,\lambda)&:=&S^{\CD}_\alpha(\lambda)\ ,\\
\nonumber \Gamma_3(x)\equiv\Gamma_3(C,D,\alpha,x)&:=&
\left(\begin{matrix}\varphi^-_0(x) & 0 \\ 0 & \varphi^-_{-1}(x) \end{matrix}\right) +
\left(\begin{matrix}\tilde{\varphi}_0(x) & 0 \\ 0 & \tilde{\varphi}_{-1}(x) \end{matrix}\right)
\widetilde{S}^{\CD}_\alpha(+\infty)\ ,\\
\nonumber \Gamma_4(\lambda)\equiv\Gamma_4(C,D,\alpha,\lambda)&:=& 1.
\end{eqnarray}

Clearly, $\Gamma_1(\cdot)$ and $\Gamma_3(\cdot)$ are continuous functions on $[-\infty,\infty]$
with values in $\U_2(\C)$, and $\Gamma_2(\cdot)$ and $\Gamma_4(\cdot)$ are continuous functions on
$[0,\infty]$ with values in $\U_2(\C)$.
By mimicking the approach of Section \ref{sec_baby}, one sets
$\square = B_1\cup B_2 \cup B_3 \cup B_4$
with $B_1 = \{0\}\times[-\infty,+\infty]$, $B_2 =[0,+\infty]\times\{+\infty\}$, $B_3 = \{+\infty\}\times[-\infty,+\infty]$,
and $B_4=[0,+\infty]\times\{-\infty\}$,
and observes that the function $\Gamma_\alpha^{\CD}= (\Gamma_1,\Gamma_2,\Gamma_3,\Gamma_4)$
belongs to $C\big(\square; \U_2(\C)\big)$.
As a consequence, the winding number $\Wind\big(\Gamma_\alpha^{\CD}\big)$ based on the map
\begin{equation*}
\square \ni \xi \mapsto \det \big[\Gamma_\alpha^{\CD}(\xi)\big]\in \T
\end{equation*}
is well defined, and our aim is to relate it to the spectral properties of $H^{\CD}_\alpha$.

The following statement is our Levinson's type theorem for this model:

\begin{Theorem}\label{Thm_Lev_AB}
For any $\alpha\in (0,1)$ and any admissible pair $(C,D)$ one has
\begin{equation*}
\Wind \big(\Gamma_\alpha^{\CD}\big) = \# \sigma_p(H_\alpha^{\CD}).
\end{equation*}
\end{Theorem}

The proof of this equality can be obtained by a case-by-case study.
It is a rather long computation which has been performed in \cite[Sec.~III]{PR} and we shall only recall the detailed results.
Note that one has to calculate separately the contribution to the winding number from the functions
$\Gamma_1$, $\Gamma_2$ and $\Gamma_3$, the contribution of $\Gamma_4$ being always trivial.
Below, the contribution to the winding of the function $\Gamma_j$ will be denoted by  $w_j(\Gamma_\alpha^{\CD})$.
Let us also stress that due to \eqref{Gma2} the contribution of $\Gamma_2$ corresponds to the contribution of the scattering operator.
It will be rather clear that a naive approach of Levinson's theorem involving only the contribution of the scattering operator would
lead to a completely wrong result.

We now list the results for the individual contributions. They clearly depend on $\alpha$, $C$ and $D$.
The various cases have been divided into subfamilies.

\begin{center}
\begin{tabular}{c|c|c|c|c|c|}
Conditions & $\#\sigma_p(H_\alpha^{\CD})$ & $w_1(\Gamma_\alpha^{\CD})$ & $w_2(\Gamma_\alpha^{\CD})$ & $w_3(\Gamma_\alpha^{\CD})$  & $\Wind(\Gamma_\alpha^{\CD})$
\\ \hline\hline
$ D=0 $&$ 0 $&$ 0 $&$ 0 $&$ 0 $&$ 0 $\\\hline
$C=0 $&$ 0 $&$ -1 $&$ 0 $&$
 1 $&$ 0$ \\ \hline
\end{tabular}
\end{center}

\vspace{5mm}

Now, if $\det(D)\neq 0$ and $\det(C)\ne 0$, we set $E:=D^{-1}C=:(e_{jk})_{j,k=1}^2$ and obtains:
\begin{center}
\begin{tabular}{c|c|c|c|c|c|}
Conditions & $\#\sigma_p(H_\alpha^{\CD})$ & $w_1(\Gamma_\alpha^{\CD})$ & $w_2(\Gamma_\alpha^{\CD})$ & $w_3(\Gamma_\alpha^{\CD})$  & $\Wind(\Gamma_\alpha^{\CD})$
\\ \hline\hline
$ e_{11} e_{22}\ge 0$, $\tr(E)>0$, $\det(E)>0$ &$ 0 $&$ 0 $&$ -1 $&$ 1 $&$ 0$ \\\hline
$ e_{11} e_{22}\ge 0$, $\tr(E)>0$, $\det(E)<0$ &$ 1$ &$ 0 $&$ 0 $&$ 1 $&$ 1 $\\\hline
$ e_{11} e_{22}\ge 0$, $\tr(E)<0$, $\det(E)>0$ &$ 2 $&$ 0 $&$ 1 $&$ 1 $&$ 2$ \\\hline
$ e_{11} e_{22}\ge 0$, $\tr(E)<0$, $\det(E)<0$ &$ 1 $&$ 0 $&$ 0 $&$ 1 $&$ 1$ \\\hline
$ e_{11}=e_{22}=0,\det(E)< 0$&$ 1 $&$ 0 $&$ 0 $&$ 1 $&$ 1$ \\\hline
$ e_{11}\;\! e_{22}<0 $&$ 1 $&$ 0 $&$ 0 $&$ 1 $&$ 1$ \\\hline
\end{tabular}
\end{center}

\vspace{5mm}

If $\det(D)\neq 0$, $\det(C)=0$ and if we still set $E:=D^{-1}C$ one has:
\begin{center}
\begin{tabular}{c|c|c|c|c|c|}
Conditions & $\#\sigma_p(H_\alpha^{\CD})$ & $w_1(\Gamma_\alpha^{\CD})$ & $w_2(\Gamma_\alpha^{\CD})$ & $w_3(\Gamma_\alpha^{\CD})$  & $\Wind(\Gamma_\alpha^{\CD})$
\\ \hline\hline
$ e_{11}=0, \tr(E)>0 $&$ 0 $&$ -\alpha $&
$ \alpha-1 $&$ 1 $&$ 0 $\\\hline
$e_{11}\;\!e_{22}\neq 0,\tr(E)>0,\alpha<1/2 $&$ 0 $&$ -\alpha $&
$ \alpha-1 $&$ 1 $&$ 0 $\\\hline
$ e_{11}>0, \tr(E)<0 $&$ 1 $&$ -\alpha $&$ \alpha $&$ 1 $&$ 1$ \\\hline
$ e_{11}\;\!e_{22}\neq0, \tr(E)<0,\alpha<1/2 $&$ 1 $&$ -\alpha $&$ \alpha $&$ 1 $&
$ 1$ \\\hline
$ e_{22}=0, \tr(E)>0 $&$ 0 $&$ \alpha -1$&$ -\alpha $&$ 1 $&$ 0$ \\\hline
$ e_{11}\;\!e_{22}\neq0, \tr(E)>0,\alpha>1/2 $&$ 0 $&$ \alpha-1$&$ -\alpha $&
$ 1 $&$ 0$ \\\hline
$ e_{22}=0, \tr(E)<0 $&$ 1 $&$ \alpha-1 $&$ 1-\alpha $&$ 1 $&$ 1 $\\\hline
$ e_{11}\;\!e_{22}\neq0, \tr(E)<0,\alpha>1/2 $&$ 1 $&$ \alpha-1 $&$ 1-\alpha $&
$ 1 $&$ 1 $\\\hline
$ e_{11}\;\!e_{22}\neq0,\tr(E)> 0 ,\alpha=1/2$&$ 0 $&$ -1/2 $&$ -1/2 $&$ 1 $&$ 0$ \\\hline
$ e_{11}\;\! e_{22}\neq 0,\tr(E)<0,\alpha=1/2 $&$ 1 $&$ -1/2 $&$ 1/2 $&$ 1 $&$ 1$ \\\hline
\end{tabular}
\end{center}

\vspace{5mm}

On the other hand, if $\dim[\Ker(D)]=1$, let us define the identification map $I:\C\to \C^2$ with $\Ran(I)=\Ker(D)^\bot$.
We then set
\begin{equation}\label{ell}
\ell:=(DI)^{-1}CI:\C\to \C
\end{equation}
which is in fact a real number because of the condition of admissibility for the pair $(C,D)$.

In the special case $\alpha=1/2$ one has:
\begin{center}
\begin{tabular}{c|c|c|c|c|c|}
Conditions & $\#\sigma_p(H_\alpha^{\CD})$ & $w_1(\Gamma_\alpha^{\CD})$ & $w_2(\Gamma_\alpha^{\CD})$ & $w_3(\Gamma_\alpha^{\CD})$  & $\Wind(\Gamma_\alpha^{\CD})$
\\ \hline\hline
$ \ell>0 $&$ 0 $&$ 0 $&$ -1/2 $&$ 1/2 $&$ 0$ \\\hline
$ \ell=0 $&$ 0 $&$ -1/2 $&$ 0 $&$ 1/2 $&$ 0$ \\\hline
$ \ell<0 $&$ 1 $&$ 0 $&$ 1/2 $&$ 1/2$&$ 1$ \\\hline
\end{tabular}
\end{center}

\vspace{5mm}

If $\dim[\Ker(D)]=1$, $\alpha<1/2$ and if ${}^t(p_1,p_2)\in \Ker(D)$ with $p_1^2+p_2^2=1$ one obtains with $\ell$ defined in \eqref{ell}:
\begin{center}
\begin{tabular}{c|c|c|c|c|c|}
Conditions & $\#\sigma_p(H_\alpha^{\CD})$ & $w_1(\Gamma_\alpha^{\CD})$ & $w_2(\Gamma_\alpha^{\CD})$ & $w_3(\Gamma_\alpha^{\CD})$  & $\Wind(\Gamma_\alpha^{\CD})$
\\ \hline\hline
$ \ell<0,p_1\neq 0 $&$ 1 $&$ 0 $&$ \alpha $&$ 1-\alpha $&$ 1$ \\\hline
$ \ell<0,p_1=0 $&$ 1 $&$ 0 $&$ 1-\alpha $&$ \alpha $&$ 1$ \\\hline
$ \ell>0,p_1\neq 0 $&$ 0 $&$ 0 $&$ \alpha-1 $&$ 1-\alpha$&$ 0$ \\\hline
$ \ell>0,p_1= 0 $&$ 0 $&$ 0 $&$ -\alpha $&$ \alpha $&$ 0$ \\\hline
$ \ell=0,p_1\;\!p_2\neq 0 $&$ 0 $&$ -\alpha $&$ 2\alpha-1 $&$ 1-\alpha $&$ 0$ \\\hline
$ \ell=0,p_1= 0 $&$ 0 $&$ -\alpha $&$ 0 $&$ \alpha$&$ 0$ \\\hline
$ \ell=0,p_2= 0 $&$ 0 $&$ \alpha-1 $&$ 0 $&$ 1-\alpha$&$ 0$ \\\hline
\end{tabular}
\end{center}

\vspace{5mm}

Finally, if $\dim[\Ker(D)]=1$, $\alpha>1/2$ and ${}^t(p_1,p_2)\in \Ker(D)$ with $p_1^2+p_2^2=1$ one has with $\ell$ defined in \eqref{ell}:
\begin{center}
\begin{tabular}{c|c|c|c|c|c|}
Conditions & $\#\sigma_p(H_\alpha^{\CD})$ & $w_1(\Gamma_\alpha^{\CD})$ & $w_2(\Gamma_\alpha^{\CD})$ & $w_3(\Gamma_\alpha^{\CD})$  & $\Wind(\Gamma_\alpha^{\CD})$
\\ \hline\hline
$ \ell<0,p_2\neq 0 $&$ 1 $&$ 0 $&$ 1-\alpha $&$ \alpha $&$ 1$ \\\hline
$ \ell<0,p_2=0 $&$ 1 $&$ 0 $&$ \alpha $&$ 1-\alpha $&$ 1$ \\\hline
$ \ell>0,p_2\neq 0 $&$ 0 $&$ 0 $&$ -\alpha $&$ \alpha$&$ 0$ \\\hline
$ \ell>0,p_2= 0 $&$ 0 $&$ 0 $&$ \alpha-1 $&$ 1-\alpha $&$ 0$ \\\hline
$ \ell=0,p_1\;\!p_2\neq 0 $&$ 0 $&$ \alpha-1 $&$ 1-2\alpha $&$ \alpha $&$ 0$ \\\hline
$ \ell=0,p_1= 0 $&$ 0 $&$ -\alpha $&$ 0 $&$ \alpha$&$ 0$ \\\hline
$ \ell=0,p_2= 0 $&$ 0 $&$ \alpha-1 $&$ 0 $&$ 1-\alpha$&$ 0$ \\\hline
\end{tabular}
\end{center}

\vspace{5mm}

Once again, by looking at these tables, it clearly appears that singling out the contribution due to the scattering operator
has no meaning. An index theorem can be obtained only if the three contributions are considered on an equal footing.

\subsection{Cyclic cohomology, $n$-traces and Connes' pairing}\label{sec_cycle}

In this section we extend the framework which led to our abstract Levinson's theorem, namely to Theorem \ref{thm_Lev}.
In fact, this statement can then be seen as a special case of a more general result.
For this part of the manuscript, we refer to \cite{KPR} and \cite[Sec.~III]{Connes}, or to the short surveys presented in \cite[Sec.~5]{KRS}
and in \cite[Sec.~4 \& 5]{KS}.

Given a complex algebra $\B$ and any $n \in \N\cup\{0\}$, let $C^n_\lambda(\B)$ be the set of $(n+1)$-linear functional on $\B$
which are cyclic in the sense that any $\eta \in C^n_\lambda(\B)$ satisfies for each $w_0,\dots,w_n\in \B$:
\begin{equation*}
\eta(w_1,\dots,w_n,w_0)=(-1)^n \eta(w_0,\dots,w_n)\ .
\end{equation*}
Then, let $\b: C^n_\lambda(\B) \to C^{n+1}_\lambda(\B)$ be the Hochschild coboundary map defined for $w_0,\dots,w_{n+1}\in \B$ by
\begin{equation*}
[\b \eta](w_0,\dots,w_{n+1}):=
\sum_{j=0}^n (-1)^j \eta(w_0,\dots,w_jw_{j+1},\dots ,w_{n+1}) +
(-1)^{n+1}\eta(w_{n+1}w_0,\dots,w_n)\ .
\end{equation*}
An element $\eta \in C^n_\lambda(\B)$ satisfying $\b\eta=0$ is called a cyclic $n$-cocyle, and the cyclic cohomology $HC(\B)$ of $\B$ is the cohomology of the complex
\begin{equation*}
0\to C^0_\lambda(\B)\to \dots \to C^n_\lambda(\B) \stackrel{\b}{\to}
C^{n+1}_\lambda(\B) \to \dots \ .
\end{equation*}

A convenient way of looking at cyclic $n$-cocycles is in terms of characters of a graded
differential algebra over $\B$. So, let us first recall that a graded differential
algebra $(\A,\dd)$ is a graded algebra $\A$ together with a map $\dd:\A\to \A$ of degree $+1$.
More precisely, $\A:=\oplus_{j=0}^\infty \A_j$ with each $\A_j$ an algebra over $\C$
satisfying the property $\A_j \;\!\A_k \subset \A_{j+k}$, and $\dd$ is a graded
derivation satisfying $\dd^2=0$.
In particular, the derivation satisfies
$\dd(w_1w_2)=(\dd w_1)w_2 + (-1)^{\deg(w_1)}w_1 (\dd w_2)$, where $\deg(w_1)$
denotes the degree of the homogeneous element $w_1$.

A cycle $(\A,\dd,\int)$ of dimension $n$ is a graded differential algebra
$(\A,\dd)$, with $\A_j=0$ for $j>n$, endowed with a linear functional
$\int : \A_n \to \C$ satisfying $\int \dd w=0$ if $w \in \A_{n-1}$, and for
$w_j\in \A_j$, $w_k \in \A_k$ with $j+k=n$~:
$$
\int w_j w_k = (-1)^{jk}\int w_k w_j\ .
$$
Given an algebra $\B$, a cycle of dimension $n$ over $\B$ is a
cycle $(\A,\dd,\int)$ of dimension $n$ together with a homomorphism $\rho:\B \to \A_0$.
In the sequel, we will assume that this map is injective and hence identify $\B$ with a subalgebra of $\A_0$ (and do not write $\rho$ anymore).
Now, if $w_0, \dots, w_n$ are $n+1$ elements of $\B$, one can define the character $\eta(w_0,\dots,w_n) \in \C$ by the formula:
\begin{equation}\label{defeta}
\eta(w_0,\dots,w_n):=\int w_0\;\!(\dd w_1)\dots (\dd w_n)\ .
\end{equation}
As shown in \cite[Prop.~III.1.4]{Connes}, the map $\eta: \B^{n+1}\to \C$ is a cyclic $(n+1)$-linear functional on $\B$ satisfying
$\b\eta=0$, {\it i.e.} $\eta$ is a cyclic $n$-cocycle.
Conversely, any cyclic $n$-cocycle arises as the character of a cycle of dimension $n$ over $\B$.
Let us also mention that a third description of any cyclic $n$-cocycle is presented
in \cite[Sec.~III.1.$\alpha$]{Connes} in terms of the universal graded differential algebra associated with $\B$.

We can now introduce the precise definition of a $n$-trace over a Banach algebra.
Recall that for an algebra $\B$ that is not necessarily unital, we denote by
$\B^+$ the canonical algebra obtained by adding a unit to $\B$.

\begin{Definition}\label{ntrace}
A $n$-trace on a Banach algebra $\B$ is the character of a cycle $(\A,\dd,\int)$ of dimension $n$
over a dense subalgebra $\B'$ of $\B$ such that for all $w_1,\dots,w_n\in \B'$ and any
$x_1,\dots,x_n \in (\B')^+$ there exists a constant $c= c(w_1,\dots,w_n)$ such that
$$
\left|\int (x_1 \dd w_1)\dots (x_n \dd w_n)\right| \leq c \|x_1\|\dots \|x_n\|\ .
$$
\end{Definition}

\begin{Remark}
Typically, the elements of $\B'$ are suitably smooth elements of $\B$
on which the derivation $\dd$ is well defined and for which the r.h.s.~of \eqref{defeta}
is also well defined.
However, the action of the $n$-trace $\eta$ can sometimes be extended to
more general elements $(w_0, \dots,w_n) \in \B^{n+1}$ by a suitable reinterpretation
of the l.h.s.~of \eqref{defeta}.
\end{Remark}

The importance of $n$-traces relies on their duality relation with $K$-groups.
Recall first that $\MM_q(\B)\cong \MM_q(\C)\otimes \B$ and that $\tr$ denotes the standard
trace on matrices. Now, let $\B$ be a $C^*$-algebra and let $\eta_n$ be a $n$-trace
on $\B$ with $n \in \N$ even.
If $\B'$ is the dense subalgebra of $\B$ mentioned in Definition \ref{ntrace} and if
$p$ is a projection in $\MM_q(\B')$, then one sets
\begin{equation*}
\langle \eta_n,p\rangle := c_n\;\![\tr\otimes \eta_n](p,\dots,p) .
\end{equation*}
Similarly, if $\B$ is a unital $C^*$-algebra and if $\eta_n$ is a $n$-trace with $n \in \N$ odd,
then for any unitary $u$ in $\MM_q(\B')$  one sets
\begin{equation*}
\langle \eta_n,u\rangle := c_n \;\![\tr\otimes \eta_n](u^*,u,u^*,\dots,u)
\end{equation*}
the entries on the r.h.s.~alternating between $u$ and $u^*$. The constants $c_n$ are given by
\begin{equation}\label{eq_const_Connes}
c_{2k}
\;=\;
\frac{1}{(2\pi i)^k}\,\frac{1}{k!}
\mbox{ , }
\qquad
c_{2k+1}
\;=\;
\frac{1}{(2\pi i)^{k+1}}\,
\frac{1}{2^{2k+1}}\,
\frac{1}{(k+\frac{1}{2})(k-\frac{1}{2})
\cdots\frac{1}{2}}
\mbox{ . }
\end{equation}

There relations are referred to as Connes' pairing between $K$-theory and cyclic cohomology
of $\B$ because of the following property, see \cite[Thm.~2.7]{Connes86} for a precise
statement and for its proof:
In the above framework, the values $\langle \eta_n,p\rangle$ and $\langle \eta_n,u\rangle$
depend only on the $K_0$-class $[p]_0$ of $p$ and on the $K_1$-class $[u]_1$ of $u$,
respectively.

We now illustrate these notions by revisiting Example \ref{ex_K_groups}.

\begin{Example}\label{exam1}
If $\B=\K(\H)$, the algebra of compact operators on a Hilbert space $\H$,
then the linear functional $\int$ on $\B$ is given by the usual trace $\Tr$
on the set $\K_1(\H)$ of trace class elements of $\K(\H)$. Furthermore, since any projection
$p\in \K(\H)$ is trace class, it follows that $\langle \eta_0,p\rangle\equiv
\langle \Tr,p\rangle$ is well defined
for any such $p$ and that this expression gives the dimension of the projection $p$.
\end{Example}

For the next example, let us recall that $\det$ denotes the usual determinant of elements of $\MM_q(\C)$.

\begin{Example}\label{exam2}
If $\B=C(\S,\C)$, let us fix $\B':=C^1(\S,\C)$. We parameterize $\S$ by the real numbers modulo $2\pi$ using $\theta$ as local coordinate.
As usual, for any $w \in \B'$
(which corresponds to an homogeneous element of degree $0$), one sets
$[\dd w](\theta):=w'(\theta)\;\!\d \theta$ (which is now an homogeneous element of degree $1$).
Furthermore, we define the graded trace $\int v \;\!\d \theta :=\int_{0}^{2\pi} v(\theta)\;\!\d \theta$
for an arbitrary element $v \;\!\d \theta$ of degree $1$. This defines the $1$-trace $\eta_1$.
A unitary element in $u\in C^1\big(\S,\MM_q(\C)\big)\equiv \MM_q\big(C^1(\S;\C)\big)$ pairs as follows
\begin{equation}\label{wn1}
\langle \eta_1,u \rangle =  c_1[\tr\otimes\eta_1](u^*,u) := \frac{1}{2\pi i}
\;\!\int_{0}^{2\pi} \tr[u(\theta)^*\;\!u'(\theta)]\;\! \d \theta =
-\frac{1}{2\pi}
\;\!\int_{0}^{2\pi} \tr[iu(\theta)^*\;\!u'(\theta)]\;\! \d \theta\ .
\end{equation}
But this quantity has already been encountered at several places in this text
and corresponds to analytic expression for the computation of (minus)\footnote{Unfortunately,
due to our convention for the computation of the winding number, the expressions computed with the constants provided in \eqref{eq_const_Connes}
differ from our expressions by a minus sign.} the winding number
of the map $\theta \mapsto \det[u(\theta)]$.
Since this quantity is of topological nature, it only requires
that the map $\theta \mapsto u(\theta)$ is continuous.
Altogether, one has thus obtained that the pairing $\langle \eta_1,u \rangle$ in \eqref{wn1} is nothing but
the computation of (minus) the winding number of the map $\det[u]: \S \to \T$, valid for any unitary $u \in
C\big(\S,\MM_q(\C)\big)$. In other words, one has obtained that $\langle \eta_1,u \rangle=-\Wind(u)$.
\end{Example}

\subsection{Dual boundary maps}\label{sec_dual}

We have seen that a $n$-trace $\eta$ over $\B$ gives rise to a functional on $K_i(\B)$ for $i=0$ or $i=1$,
{\it i.e.}~the map $ \langle \eta,\cdot \rangle$ is an element of $Hom(K_i(\B),\C)$.
In that sense $n$-traces are dual to the elements of the $K$-groups. An important question is whether this dual relation
is functorial in the sense that morphisms between the $K$-groups of different algebras yield dual morphisms on higher traces.
Here we are in particular interested in a map on higher traces which is dual to the index map,
{\it i.e.}~a map $\#$ which assigns to an even trace $\eta$ an odd trace $\#\eta$ such that
\begin{equation}\label{comp}
\langle \eta,\ind (\cdot) \rangle = \langle \#\eta,\cdot \rangle.
\end{equation}
This situation gives rise to equalities between two numerical topological invariants.

Such an approach for relating two topological invariants has already been used at few occasions.
For example, our abstract Levinson's theorem (Theorem \ref{thm_Lev}) corresponds to a equality of the form \eqref{comp}
for a $0$-trace and a $1$-trace.
In addition, in the following section we shall develop such an equality for a $2$-trace and a $3$-trace.
On the other hand, let us mention that similar equalities have also been developed for the exponential
map in \eqref{comp} instead of the index map. In this framework, an equality involving a $0$-trace and a $1$-trace
has been put into evidence in \cite{Johannes}. It gives rise to a relation
between the pressure on the boundary of a quantum system and the integrated density of states.
Similarly, a relation involving $2$-trace and a $1$-trace was involved in the proof of the equality between
the bulk-Hall conductivity and the conductivity of the current along the edge of the sample, see \cite{KRS,KS}.

\subsection{Higher degree Levinson's theorem}\label{sec_higher}

In order to derive a higher degree Levinson's theorem, let us first introduce
the algebraic framework which will lead to a much shorter new proof of Theorem \ref{Thm_Lev_AB}.
The construction is similar to the one already proposed in Section \ref{subsec_Lev3}
for Schr\"odinger operators on $\R^3$.

We recall from Section \ref{sec_AB_model} that $\H$ denotes the Hilbert space $\ltwo(\R^2)$,
that $\H_\2$ has been introduced in \eqref{eq_Hint},  and let us set $\HS:=\ltwo(\S)$.
We also denote by $\F_0:\H\mapsto \ltwo(\R_+;\HS)$ the usual spectral representation of
the Laplace operator $H_0=-\Delta$ in $\H$.
Then, we can define
$$
\EE:=\big\{\F_0^* \big[\EE_{(L,A)}\otimes \K(\HS)\big] \F_0\big\} |_{\H_\2}\subset
\B\big(\H_\2\big)\equiv \B(\H_r)\otimes \MM_2(\C).
$$
Clearly, $\EE$ is made of continuous functions of $H_0$ having limits at $0$ and $+\infty$,
and of continuous function of the generator $A$ of dilations in $\ltwo(\R^2)$ having
limits at $-\infty$ and at $+\infty$. One can then consider the short exact sequence
$$
0 \to \K(\H_\2) \hookrightarrow \EE \stackrel{q\;\;\!}{\to} C\big(\square; \MM_2(\C)\big) \to 0,
$$
and infer the following result directly from the construction presented in Section \ref{sec_Kth}.
Note that this result corresponds to \cite[Thm.~13]{KPR} and provides an alternative
proof for Theorem \ref{Thm_Lev_AB}.

\begin{Theorem}\label{Ktheo_AB}
For any $\alpha \in (0,1)$ and any admissible pair $(C,D)$, one has
$W_-^{\CD}|_{\H_\2}\in \EE$. Furthermore,
$q\big(W_-^{\CD}|_{\H_\2}\big)=\Gamma^{\CD}_\alpha \in C\big(\square; \U_2(\C)\big)$
and the following equality holds
\begin{equation*}
\Wind \big(\Gamma_\alpha^{\CD}\big) = \# \sigma_p(H_\alpha^{\CD}).
\end{equation*}
\end{Theorem}

Let us stress that the previous statement corresponds to a pointwise Levinson's theorem
in the sense that it has been obtained for fixed $C,D$ and $\alpha$.
However, it clearly calls for making these parameters degrees of freedom and thus to include them into the description of the algebras.
In the context of our physical model this amounts to considering
families of self-adjoint extensions of $H_\alpha$.
For that purpose we use the one-to-one parametrization of these extensions with elements $U \in \U_2(\C)$
introduced in Remark \ref{1to1}. We denote the self-adjoint extension corresponding to $U \in \U_2(\C)$ by $H_\alpha^U$.

So, let us consider a smooth and compact orientable $n$-dimensional manifold $\X$ without boundary.
Subsequently, we will choose for $\X$ a two-dimensional submanifold of $\U_2(\C)\times (0,1)$.
Taking continuous functions over $\X$ we get a new short exact sequence
\begin{equation}\label{degresup}
0 \to C\big(\X;\K(\H_\2)\big)\stackrel{}\hookrightarrow C(\X;\EE) \stackrel{}\to C\big(\X; C\big(\square;\MM_2(\C)\big)\big)\to 0\ .
\end{equation}
Furthermore, recall that $\K(\H_\2)$ is endowed with a $0$-trace $\Tr$ and the algebra $C\big(\square;\MM_2(\C)\big)$ with
a $1$-trace $\Wind$.
There is a standard construction in cyclic cohomology, the cup
product, which provides us with
a suitable $n$-trace on the algebra $C\big(\X,\K(\H_\2)\big)$ and a corresponding $n+1$-trace
on the algebra $C\big(\X; C\big(\square;\MM_2(\C)\big)\big)$, see \cite[Sec.~III.1.$\alpha$]{Connes}.
We describe it here in terms of cycles.

Recall that any smooth and compact manifold $\Y$ of dimension $d$
naturally defines a structure of a graded differential algebra $(\A_\Y,\dd_\Y)$,
the algebra of its smooth differential $k$-forms.
If we assume in addition that $\Y$ is orientable so that we can choose
a global volume form,
then the linear form $\int_\Y$ can be defined by integrating the
$d$-forms over $\Y$.
In that case, the algebra $C(\Y)$ is naturally endowed with the $d$-trace defined
by the character of the cycle $(\A_\Y,\dd_\Y,\int_\Y)$ of dimension $d$
over the dense subalgebra $C^\infty(\Y)$.

For the algebra $C\big(\X;\K(\H_\2)\big)$, let us denote by $\K_1(\H_\2)$ the trace class elements of $\K(\H_\2)$.
Then, the natural graded differential algebra associated with
$C^\infty\big(\X,\K_1(\H_\2)\big)$ is given by $(\A_\X\otimes \K_1(\H_\2),\dd_\X)$.
The resulting $n$-trace on $C\big(\X;\K(\H_\2)\big)$ is then defined by the character
of the cycle $(\A_\X\otimes \K_1(\H_\2),\dd_\X,\int_\X\otimes \Tr)$ over the dense
subalgebra $C^\infty\big(\X,\K_1(\H_\2)\big)$ of $C\big(\X;\K(\H_\2)\big)$. We denote it by $\eta_\X$.

For the second algebra, let us observe that
$$
C\big(\X; C\big(\square;\MM_2(\C)\big)\big)\cong C\big(\X\times\S;\MM_2(\C)\big)
\cong C(\X \times \S) \otimes \MM_2(\C).
$$
Since $\X\times \S$ is a compact orientable manifold without
boundary, the above construction
applies also to $C\big(\X\times\S;\MM_2(\C)\big)$. More precisely,
the exterior derivation on $\X\times \S$ is the sum of  $\dd_\X$ and
$\dd_{\S}$ (the latter was denoted simply by $\dd$ in Example
\ref{exam2}). Furthermore, we consider the natural volume form on $\X\times \S$.
Note because of the factor $\MM_2(\C)$ the graded trace of the cycle
involves the usual matrix trace $\tr$.
Thus the resulting $n+1$-trace is the character of the cycle
$(\A_{\X\times \S}\otimes \MM_2(\C),\dd_\X+\dd_{\S},\int_{\X\times
\S}\otimes \tr)$. We denote it by $\#\eta_\X$.

Having these constructions at our disposal we can now state the main result of this section.
For the statement, we use the one-to-one parametrization of the extensions $H_\alpha^U$ of $H_\alpha$
introduced in Remark \ref{1to1}. We also consider a family
$\{W_-^{U,\alpha}\}_{(U,\alpha)\in \X} \subset \B(\H_\2)$, with
$W_-^{U,\alpha}:=W_-(H_\alpha^U,H_0)$,
parameterized by some compact orientable and boundaryless submanifold $\X$ of $\U_2(\C)\times (0,1)$.
This family defines several maps, namely
$$
\bW_-:\X\ni (U,\alpha)\mapsto W_-^{U,\alpha}\in \EE
$$
as well as
$$
\bG:\X\ni (U,\alpha)\mapsto \Gammaq^{U,\alpha}\in C\big(\square;\MM_2(\C)\big),
$$
with $\Gammaq^{U,\alpha}:= q\big(W_-^{U,\alpha}\big)$,
and also
$$
\bE:\X\ni (U,\alpha)\mapsto E_{\rm p}\big(H_\alpha^U\big).
$$

\begin{Theorem}[Higher degree Levinson's theorem]\label{thm-ENN}
Let $\X$ be a smooth, compact and orientable $n$-dimensional submanifold of $\U_2(\C)\times (0,1)$ without boundary.
Let us assume that the map $\bW_-:\X\to \EE$
is continuous. Then the maps $\bG$ and $\bE$ are continuous, and the following equality holds:
\begin{equation*}
\ind[\bG]_1 = -[\bE]_0
\end{equation*}
where $\ind$ is the index map from the $K_1$-group of the algebra $C\big(\X; C\big(\square;\MM_2(\C)\big)\big)$
to the $K_0$-group of the algebra $C\big(\X;\K(\H_\2)\big)$.
Furthermore, the numerical equality
\begin{equation}\label{eq23}
\big\langle \#\eta_{\X},[\bG]_1 \big\rangle =-\big \langle \eta_{\X},[\bE]_0 \big\rangle
\end{equation}
also holds.
\end{Theorem}

The proof of this statement is provided in \cite[Thm.~15]{KPR} and is based on the earlier work \cite{KRS}.
Let us point out that r.h.s.~of \eqref{eq23} corresponds to the Chern number of the vector bundle given by the
eigenvectors of $H_\alpha^U$.
On the other hand, the l.h.s.~corresponds to a $(n+1)$-trace applied to $\bG$ which is constructed
from the scattering theory for the operator $H_\alpha^U$. For these reasons, such an equality has been named
\emph{a higher degree Levinson's theorem}.
In the next section we illustrate this equality by a special choice of the manifold $\X$.

\subsubsection{A non-trivial example}\label{sec_non_trivial}

Let us now choose a $2$-dimensional manifold $\X$ and show that the previous relation
between the corresponding $2$-trace and $3$-trace is not trivial.
More precisely, we shall choose a manifold $\X$ such that the r.h.s. of \eqref{eq23} is not equal to $0$.

For that purpose, let us fix two complex numbers  $\lambda_1,\lambda_2$ of modulus $1$
with $\Im \lambda_1<0< \Im \lambda_2$ and consider the set $\X \subset \U_2(\C)$ defined by :
\begin{equation*}
\X = \left\{ V
\left(\begin{smallmatrix}
\lambda_1 & 0 \\
0 & \lambda_2
\end{smallmatrix}\right)
 V^* \mid V\in \U_2(\C)\right\}.
\end{equation*}
Clearly, $\X$ is a two-dimensional smooth and compact manifold without boundary, which can be parameterized by
\begin{equation}\label{param}
\X= \left\{
\left(\begin{matrix}
\rho^2 \lambda_1 + (1-\rho^2) \lambda_2 &
\rho(1-\rho^2)^{1/2} \;\!e^{i\phi}(\lambda_1-\lambda_2) \\
\rho(1-\rho^2)^{1/2} \;\!e^{-i\phi}(\lambda_1-\lambda_2) &
(1-\rho^2) \lambda_1 + \rho^2 \lambda_2
\end{matrix}\right)
\mid \rho \in [0,1] \hbox{ and } \phi \in [0,2\pi)\right\}.
\end{equation}
Note that the $(\theta,\phi)$-parametrization of $\X$ is complete in the sense that it covers all the manifold
injectively away from a subset of codimension $1$, but it has coordinate singularities at $\rho\in \{0,1\}$.

By \cite[Lem.~16]{PR}, for each $U\equiv U(\rho,\phi)\in \X$ the operator $H^U_\alpha$ has a single negative eigenvalue.
It follows that the projection $E_{\rm p}(H_\alpha^U)$ is non-trivial for any $\alpha \in (0,1)$ and any $U\in \X$,
and thus the expression $\big \langle \eta_{\X},[\bE]_0 \big\rangle$ can be computed.
This rather lengthy computation has been performed in \cite[Sec.~V.D]{KPR}, and it turns out that the following
result has been found for this example:
$$
\big \langle \eta_{\X},[\bE]_0 \big\rangle = 1.
$$
As a corollary of Theorem \ref{thm-ENN} one can then deduce that:

\begin{Proposition}\label{propfinal}
Let $\lambda_1,\lambda_2$ be two complex numbers of modulus $1$
with $\Im \lambda_1<0< \Im \lambda_2$ and consider the set $\X \subset \U_2(\C)$ defined by \eqref{param}.
Then the map $\bW_-: \X \to \EE$ is continuous and the following equality holds:
\begin{equation*}
\frac{1}{24\pi^2} \int_{\X\times \square}\tr\big[\bG^* \; \big(\dd_{\X\times\square}\bG\big) \wedge
\big(\dd_{\X\times\square}\bG^*\big) \wedge \big(\dd_{\X\times\square}\bG\big) \big]
= 1.
\end{equation*}
\end{Proposition}

\section{Appendix}\label{sec_App}
\setcounter{equation}{0}

\subsection{The baby model}\label{subsec_baby}

In this section, we provide the proofs on the baby model which have not been presented in Section \ref{sec_baby}.
The notations are directly borrowed from this section, but we shall mainly review, modify and extend
some results obtained in \cite[Sec.~3.1]{Yaf}.

First of all, it is shown in \cite[Sec.~3.1]{Yaf} that the wave operators $W_\pm^\alpha$ exist and are asymptotically complete.
Furthermore, rather explicit expressions for them are proposed in \cite[Eq.~3.1.15]{Yaf}.
Let us also mention that an expression for the scattering operator $S^\alpha$ is given in \cite[Sec.~3.1]{Yaf}, namely
\begin{equation*}
S^\alpha = \frac{\alpha +i\sqrt{H_0}}{\alpha-i\sqrt{H_0}}.
\end{equation*}

In the following lemma, we derive new expressions for the wave operators. They involve the scattering operator $S^\alpha$
as well as the Fourier sine and cosine transforms $\F_{\rm s}$ and $\F_{\rm c}$ defined for $x,k \in \R_+$ and any $f \in C_{\rm c}(\R_+)\subset \ltwo(\R_+)$ by
\begin{align}
\label{Fouriers} [\F_{\rm s} f](k)&:= (2/\pi)^{1/2} \int_0^\infty \sin(k x) f(x)\;\!\d x \\
\label{Fourierc} [\F_{\rm c} f](k)&:= (2/\pi)^{1/2} \int_0^\infty \cos(k x) f(x)\;\!\d x.
\end{align}

\begin{Lemma}\label{lem_baby_W}
The following equalities hold:
\begin{align*}
W_-^\alpha&= 1+ \12 \big(1-i\F_{\rm c}^*\F_{\rm s}\big) [S^\alpha-1],\\
W_+^\alpha &= 1+ \12 \big(1+i\F_{\rm c}^*\F_{\rm s}\big) \big[(S^\alpha)^*-1\big].
\end{align*}
\end{Lemma}

\begin{proof}
We use the notations of \cite[Sec.~3.1]{Yaf} without further explanations.
For $W_-^\alpha$, it follows from \cite[Eq.~3.1.15 \& 3.1.20]{Yaf} that
\begin{align*}
W_-^\alpha&= iP\F_-^*\F_{\rm s}
= \i2 (2\F_-P)^*\F_{\rm s}
= \i2 \big(\Pi_+ -S^*\Pi_-\big)^* \F_{\rm s}\\
&= \i2 \big(2i\F_{\rm s} + \Pi_- -S^*\Pi_-\big)^* \F_{\rm s}
= \i2 \big(-2i\F_{\rm s}^* - \Pi_-^*(S-1)\big) \F_{\rm s} \\
&= 1- \i2 \Pi_-^* \F_{\rm s} (S^\alpha-1).
\end{align*}
Thus, one obtains
\begin{equation*}
W_-^\alpha
=1 - \i2 \big(\F_{\rm c}^*+i\F_{\rm s}^*\big) \F_{\rm s} [S^\alpha-1] =1+ \12 \big(1-i\F_{\rm c}^*\F_{\rm s}\big) [S^\alpha-1].
\end{equation*}

A similar computation leads to the mentioned result for $W_+^\alpha$.
\end{proof}

Now, we provide another expression for the operator $-i\F_{\rm c}^*\F_{\rm s}$. For that purpose, let $A$ denote the generator of dilations
in $\ltwo(\R_+)$.

\begin{Lemma}\label{lemmaA}
The following equality holds
\begin{equation*}
-i\F_{\rm c}^*\F_{\rm s} = \tanh(\pi A) -i \cosh(\pi A)^{-1}.
\end{equation*}
\end{Lemma}

\begin{proof}
This proof is inspired by the proof of \cite[Lem.~3]{KR1D}.
Let us first define for $x,y \in \R_+$ and $\varepsilon>0$ the kernel of the operator $I_\varepsilon$ by
\begin{equation*}
I_\varepsilon(x,y):=(1/\pi)\Big[\frac{x+y}{(x+y)^2+\varepsilon^2} - \frac{x-y}{(x-y)^2+\varepsilon^2}\Big].
\end{equation*}
Then, an easy computation shows that
$I_\varepsilon(x,y) = (2/\pi)\int_0^\infty \cos(xz)\sin(yz)\e^{-\varepsilon z} \d z$,
and an application of the theorems of Fubini and Lebesgues for $f \in C_{\rm c}(\R_+)$ leads to the equality
\begin{equation*}
\lim_{\varepsilon \searrow 0}I_\varepsilon f = \F_{\rm c}^*\F_{\rm s} f.
\end{equation*}

Now, by comparing the expression for $[I_\varepsilon f](x)$ with the following expression
\begin{equation*}
[\varphi(A)f](x)= \frac{1}{\sqrt{2\pi}}\int_0^\infty \check{\varphi}\Big(\ln\big(\frac{x}{y}\big)\Big)\;\!\Big(\frac{x}{y}\Big)^{1/2} f(y) \frac{\d y}{x},
\end{equation*}
valid for any essentially bounded function $\varphi$ on $\R$ whose inverse Fourier transform is a distribution on $\R$,
one obtains that
\begin{equation*}
\check{\varphi}(s) =\frac{1}{\sqrt{2\pi}}\Big[\frac{1}{\cosh(s/2)}-\Pv\frac{1}{\sinh(s/2)}\Big],
\end{equation*}
where $\Pv$ means principal value.
Finally, by using that the Fourier transform of the distribution $s \mapsto \Pv\frac{1}{\sinh(s/2)}$
is the function $-i\sqrt{2\pi}\tanh(\pi \cdot)$ and the one of $s \mapsto \frac{1}{\cosh(s/2)}$ is
the function $\sqrt{2\pi}\cosh(\pi \cdot)^{-1}$, one obtains that
\begin{equation*}
\varphi(A)=\cosh(\pi A)^{-1} + i\tanh(\pi A).
\end{equation*}
By replacing $\F_{\rm c}^*\F_{\rm s}$ with this expression, one directly obtains the stated result.
\end{proof}

\begin{Corollary}
The following equalities hold:
\begin{align*}
W_-^\alpha&=1+ \12 \big(1+\tanh(\pi A) -i \cosh(\pi A)^{-1}\big) \Big[\frac{\alpha +i\sqrt{H_0}}{\alpha-i\sqrt{H_0}}-1\Big],\\
W_+^\alpha &=1+\12 \big(1-\tanh(\pi A) +i \cosh(\pi A)^{-1}\big) \Big[\frac{\alpha -i\sqrt{H_0}}{\alpha+i\sqrt{H_0}}-1\Big].
\end{align*}
\end{Corollary}

\subsection{Regularization}\label{subsec_app_regul}

Recall that $\HS$ stands for an arbitrary Hilbert space and that $\U(\HS)$ corresponds to the set of unitary operators on $\HS$.
Let $\Gamma$ be a map $\S \to \U(\HS)$ such that $\Gamma(t)-1 \in  \K(\HS)$ for all $t \in \S$.
For $p \in \N$ we set $\K_p(\HS)$ for the $p$-th Schatten ideal in $\K(\HS)$.

\begin{proof}[Proof of Lemma \ref{lem_reg_1}]
For simplicity, let us set $A(t):=1-\Gamma(t)$ for any $t \in I$ and recall from \cite[Eq.~(XI.2.11)]{GGK} that
$\det_{p+1}\big(\Gamma(t)\big) = \det\big(1+R_{p+1}(t)\big)$
with
\begin{equation*}
R_{p+1}(t):= \Gamma(t)\exp\Big\{\sum_{j=1}^p\frac{1}{j}A(t)^j\Big\} -1\ .
\end{equation*}
Then, for any $t,s \in I$ with $s\neq t$ one has
\begin{eqnarray*}
\frac{\det_{p+1}\big(\Gamma(s)\big)}{\det_{p+1}
\big(\Gamma(t)\big)}
&=&\frac{\det\big(1+R_{p+1}(s)\big)}{\det
\big(1+R_{p+1}(t)\big)} \\
&=& \frac{\det\big[\big(1+R_{p+1}(t)\big)\big(1+B_{p+1}(t,s)\big)\big]}{\det
\big(1+R_{p+1}(t)\big)}\\
&=& \det\big(1+B_{p+1}(t,s)\big)
\end{eqnarray*}
with
$B_{p+1}(t,s) = \big(1+R_{p+1}(t)\big)^{-1} \big(R_{p+1}(s)-R_{p+1}(t)\big)$.
Note that $1+R_{p+1}(t)$ is invertible in $\B(\HS)$ because $\det_{p+1}\big(\Gamma(t)\big)$ is non-zero.
With these information let us observe that
\begin{eqnarray}\label{presque2}
\frac{\frac{
\det_{p+1}(\Gamma(s))- \det_{p+1}(\Gamma(t))}{|s-t|}
}{\det_{p+1}(\Gamma(t))}
= \frac{1}{|s-t|}
\big[\det\big(1+B_{p+1}(t,s)\big)-1\big]\ .
\end{eqnarray}
Thus, the statement will be obtained if the limit $s\to t$ of this expression exists and if this limit is equal to the r.h.s. of \eqref{aobtenir}.

Now, by taking into account the asymptotic development of $\det(1+\varepsilon X)$ for $\varepsilon$ small enough, one obtains that
\begin{eqnarray}\label{presque}
\nonumber &&\lim_{s \to t} \frac{1}{|s-t|}
\big[\det\big(1+B_{p+1}(t,s)\big)-1\big] \\
\nonumber &=& \lim_{s \to t} \tr\bigg[\frac{B_{p+1}(t,s)}{|s-t|}\bigg] \\
&=&
\lim_{s \to t} \tr\bigg[H_{p+1}(t)^{-1}\frac{H_{p+1}(s)-H_{p+1}(t)}{|s-t|}\bigg]
\end{eqnarray}
with $H_{p+1}(t):=\big(1-A(t)\big)\exp\big\{\sum_{j=1}^p\frac{1}{j}A(t)^j\big\}$.
Furthermore, it is known that the function $h$ defined for $z \in \C$ by $h(z):=z^{-(p+1)}(1-z)\exp\big\{\sum_{j=1}^p \frac{1}{j}z^j\big\}$ is an entire function,
see for example \cite[Lem.~6.1]{S1}.
Thus, from the equality
\begin{equation}\label{bellavista}
H_{p+1}(t) = A(t)^{p+1}h\big(A(t)\big)
\end{equation}
and from the hypotheses on $A(t)\equiv 1-\Gamma(t)$ it follows that the map $I \ni t \mapsto H_{p+1}(t)\in \K_1(\HS)$
is continuously differentiable in the norm of $\K_1(\HS)$. Thus, the limit \eqref{presque} exists,
or equivalently the limit \eqref{presque2} also exists.
Then, an easy computation using the geometric series leads to the expected result, {\it i.e.}~the limit
in \eqref{presque} is equal to the r.h.s.~of \eqref{aobtenir}.

Finally, for the last statement of the lemma, it is enough to observe from \eqref{bellavista} that
the map $I \ni t \mapsto H_{p}(t)\in \K_1(\HS)$ is continuously differentiable in the norm of $\K_1(\HS)$ if the map
$I \ni t\mapsto \Gamma(t)-1 \in \K_p(\HS)$ is continuously differentiable in norm of $\K_p(\HS)$.
Thus the entire proof holds already for $p$ instead of $p+1$.
\end{proof}

\begin{proof}[Proof of Lemma \ref{lem_reg_2}]
Let us denote by $\S_0$ the open subset of $\S$ (with full measure) such
that $\S_0\ni t \mapsto \Gamma(t) \in \K(\HS)$ is continuously differentiable.
One first observes that for any $t\in \S_0$ and $q>p$ one has
\begin{eqnarray*}
M_{q}(t)&:=&\tr\big[\big(1-\Gamma(t)\big)^{q}\Gamma(t)^*
\Gamma'(t)\big] \\
&=& \tr\Big[\big(1-\Gamma(t)\big)^{q-1}\Gamma(t)^*
\Gamma'(t)- \Gamma(t)\big(1-\Gamma(t)\big)^{q-1}\Gamma(t)^*
\Gamma'(t)\Big] \\
&=& M_{q-1}(t) -\tr\big[\big(1-\Gamma(t)\big)^{q-1}
\Gamma'(t)\big]
\end{eqnarray*}
where the unitarity of $\Gamma(t)$ has been used in the third  equality. Thus the statement will be proved by reiteration if one shows that
the map $\S_0\ni t\mapsto \tr\big[\big(1-\Gamma(t)\big)^{q-1} \Gamma'(t)\big]\in \K(\HS)$ is integrable, with
\begin{equation}\label{aintegrer}
\int_{\S_0} \tr\big[\big(1-\Gamma(t)\big)^{q-1}
\Gamma'(t)\big] \d t =0.
\end{equation}

For that purpose, let us set for simplicity $A(t):=1-\Gamma(t)$ and observe that for $t,s$ in the same arc of $\S_0$ and with $s\neq t$ one has
\begin{eqnarray*}
\tr[A(s)^q]-\tr[A(t)^q] =
\tr\big[A(s)^q-A(t)^q\big] =
\tr\Big[P_{q-1}\big(A(s), A(t)\big)\, \big(A(s)-A(t)\big) \Big]
\end{eqnarray*}
where $P_{q-1}\big(A(s), A(t)\big)$ is a polynomial of degree $q-1$ in the two non commutative variables $A(s)$ and $A(t)$.
Note that we were able to use the cyclicity because of the assumptions $q-1\geq p$ and $A(t)\in \K_p(\HS)$ for all $t \in \S$. Now, let us observe that
\begin{eqnarray*}
&&\bigg|\frac{1}{|s-t|}
\tr\Big[P_{q-1}\big(A(s), A(t)\big)\, \big(A(s)-A(t)\big)\Big]- \tr\Big[ P_{q-1}\big(A(t),A(t)\big)\,A'(t)\Big]\bigg|
\\
&\leq& \Big\|\frac{A(s)-A(t)}{|s-t|}\Big\|  \,\Big|
\tr\big[P_{q-1}\big(A(s), A(t)\big)-P_{q-1}\big(A(t), A(t)\big)\big]
\Big| \\
&& + \Big\|\frac{A(s)-A(t)}{|s-t|}-A'(t)\Big\|
\,\Big|
\tr\big[P_{q-1}\big(A(t), A(t)\big)\big]\Big|\ .
\end{eqnarray*}
By assumptions, both terms vanish as $s \to t$. Furthermore, one observes that $P_{q-1}\big(A(t),A(t)\big)=q A(t)^{q-1}$.
Collecting these expressions one has shown that
\begin{equation*}
\lim_{s\to t}\frac{\tr[A(s)^q]-\tr[A(t)^q]}{|s-t|}- q\,\tr[A(t)^{q-1}A'(t)] =0\ ,
\end{equation*}
or in simpler terms $\frac{1}{q}\big(\tr[A(\cdot)^q]\big)'(t) = \tr[A(t)^{q-1}A'(t)]$. By inserting this equality into \eqref{aintegrer} and
by taking the continuity of $\S\ni t\mapsto \Gamma(t)$ into account, one directly obtains that this integral is equal to $0$, as expected.
\end{proof}



\end{document}